\documentclass[11pt]{article}

\usepackage{amsfonts}
\usepackage{amssymb}
\usepackage{amsmath}
\usepackage{epsfig}
\usepackage{color}
\usepackage{graphicx}
\usepackage{cite}
\usepackage{epsfig}
\usepackage{subfig}

\let\oldbibliography\thebibliography
\renewcommand{\thebibliography}[1]{%
\oldbibliography{#1}%
\setlength{\itemsep}{0pt}%
}
\newlength{\off}\newlength{\temp}
\newlength{\laboff} \newcounter{stepno}

\newtheorem{theorem}{Theorem}[section]
\newtheorem{corollary}[theorem]{Corollary}

\newtheorem{lemma}[theorem]{Lemma}

\newtheorem{proposition}[theorem]{Proposition}

\newtheorem{claim}[theorem]{Claim}
\newtheorem{fact}[theorem]{Fact}

\newcommand{\sq}{\hbox{\rlap{$\sqcap$}$\sqcup$}}
\newcommand{\qed}{\hspace*{\fill}\sq}
\newenvironment{proof}{\noindent {\em Proof.}\ }{\qed\par\vskip 4mm\par}
\newenvironment{proofof}[1]{\bigskip \noindent {\em Proof of #1:}\quad }{\qed\par\vskip 4mm\par}

\newcommand{\ignore}[1]{}

\newcommand{\MyQuote}[1]{\bigskip \refstepcounter{equation}%
     \parbox{0.83\columnwidth}{#1}\hfill(\arabic{equation})\\[0.2cm] }

\def\pr{\mbox{P}}
\def\supp{\mbox{supp}}
\def\ex{\mbox{E}}

\def\CC{\mathcal{C} }
\def\CA{\mathcal{A} }
\def\CL{\mathcal{L} }
\def\CF{\mathcal{F} }

\def\CG{\mathcal{G} }

\def\RR{\mathbb{R}}
\def\ZZ{\mathbb{Z}}
\def\PP{\mathbb{P}}

\def\NN{\mathbb{N}}

\def\dcg{\mbox{DC-WSN}}

\def\rgg{\mbox{RGG}}
\def\vbrgg{\mbox{VB-RGG}}
\def\prgg{\mbox{RGG}_P}
\def\pvbrgg{\mbox{VB-RGG}_P}

\def\dcc{\mbox{DC-C-WSN}}
\def\dcr{\mbox{DC-R-WSN}}

\def\Broad{\mbox{Send}}

\newif\ifcomments
\commentsfalse
\newif\ifrevision
\revisionfalse
\newif\ifrevisiontwo
\revisiontwofalse

\begin{document}

\title{Optimal Radius for Connectivity in Duty-Cycled Wireless Sensor Networks}
\author{Amitabha Bagchi \and Cristina Pinotti \and Sainyam Galhotra \and Tarun Mangla}

%\category{C.2.1}{Network Architecture and Design}{Wireless Communication}

%\terms{Theory, Model, Design, Performance}

%\keywords{Wireless sensor networks, duty-cycled, connectivity radius}

%\acmformat{Amitabha Bagchi, Sainyam Galhotra, Tarun Mangla,
%and Cristina M. Pinotti, 20xx. Optimal Radius for Connectivity in Duty-Cycled W%ireless Sensor Networks.}
% At a minimum you need to supply the author names, year and a title.
% IMPORTANT:
% Full first names whenever they are known, surname last, followed by a period.
% In the case of two authors, 'and' is placed between them.
% In the case of three or more authors, the serial comma is used, that is, all author names
% except the last one but including the penultimate author's name are followed by a comma,
% and then 'and' is placed before the final author's name.
% If only first and middle initials are known, then each initial
% is followed by a period and they are separated by a space.
% The remaining information (journal title, volume, article number, date, etc.) is 'auto-generated'.

\maketitle

\begin{abstract}
We investigate the condition on transmission radius needed to achieve
connectivity in duty-cycled wireless sensor networks (briefly,
DC-WSN).  First, we settle a conjecture of Das et. al. (2012) and
prove that the connectivity condition on Random Geometric Graphs
(RGG), given by Gupta and Kumar (1989), can be used to derive a weak
sufficient condition to achieve connectivity in DC-WSN. To find a
stronger result, we define a new vertex-based random connection model
which is of independent interest. Following a proof technique of
Penrose (1991) we prove that when the density of the nodes approaches
infinity then a finite component of size greater than 1 exists with
probability 0 in this model. We use this result to obtain an optimal
condition on node transmission radius which is both necessary and
sufficient to achieve connectivity and is hence {\em optimal}. The
optimality of such a radius is also tested via simulation for two
specific duty-cycle schemes, called the {\em contiguous} and the {\em
  random selection} duty-cycle scheme. Finally, we design a
minimum-radius duty-cycling scheme that achieves connectivity with a
transmission radius arbitrarily close to the one required in Random
Geometric Graphs. The overhead in this case is that we have to spend
some time computing the schedule.
\end{abstract}

\section{Introduction}
\label{sec:intro}

Wireless Sensor Networks (WSNs) have a wide range of applications from
wildlife monitoring to critical infrastructure monitoring, from
traffic management to individual health management \cite{Yick}. The
three primary functions of a sensor are to sense, process and
communicate. After being deployed randomly over a limited area,
sensors start to sense a phenomenon on a regular basis. Then, they
process the raw data, and wirelessly forward it to a base station,
connected to the external world, via multihop paths. Since sensors
deployments are often made in environments where regular power supply
cannot be guaranteed, they have to rely on batteries and are therefore
constrained by a limited energy budget. Their monitoring activities,
however, tend to have a long time line, and so energy consumption is
the overarching problem for WSN operations.

For conserving energy in WSNs, firstly, transmission power
can be carefully controlled. This allows to
save energy for the sending node, but also it avoids to loose energy
at neighboring nodes for interferences.
However, enough transmission power has to be used to ensure that the
basic communication function of the WSN--relaying data to the base
station--can be completed successfully. This trade off translates into
a question of optimal radius for the connectivity of the WSN
graph. This has been studied using the Random Geometric Graph (RGG)
model by Gupta and Kumar~\cite{gupta-cdc:1998} among others.  

Sensors can also save energy during sensing and processing activities
by turning off the radio/sensing sensor module when possible. This
fact has been exploited by passive power conservation
mechanisms~\cite{Pantazis07}.  In fact the basic idea behind the
notion of Duty-Cycled Wireless Sensor Networks (briefly, \dcg s) is
that sensors do not need to sense and process all the time. 
%CRISTINA
\ifrevisiontwo {\bf \fi 
This is not only an option, but it might be a necessity if the sensors
harvest energy from the environment.  For example, a sensor powered by
a solar cell must harvest its energy only during the daytime, and can
release it other times.  Currently, sensors powered by solar cells are
produced which, after recharging for few hours hours in daytime, are
fully functional (i.e., transmitting a measured value every 15-30
minutes) for one day, even in complete darkness.  While these solar
cells are macro devices--they require long recharge times and offer
long self-discharging periods--we also have micro energy harvesters,
available in sizes ranging from centimeters down to micrometers, which
store enough energy just for one measurement. So, they offer very short
charge and self-discharge periods. The functioning of a sensor powered
with a micro energy harvester implies intermittent measuring and data
sending followed by scavenging and storing the energy for the next
measurement in a buffer (capacitor, battery).  Hence recharge
opportunities impact individual node operations as well as system
design considerations.  Indeed, to exploit the possible added
benefits, the nodes must optimize their capability by tuning different
node parameters, like the duration of the recharging period, and its
starting point, in a manner that the available energy is not exhausted
before the next recharge cycle~\cite{Sudevalayam:2011}.  Since these
parameters depend on the sensor technology and on the applications, we
assume that they are injected in the sensors at the time of their
deployment.  
\ifrevisiontwo } \fi
%CRISTINA

For sensors such as these and others we provide the following model:
In our duty-cycling paradigm, sensors repeat a cycle of fixed length
$L$, during which they switch between the {\em awake} and {\em sleep}
mode. During the sleep mode, the sensors recharge or conserve their
batteries by turning-off their sense/processing/radio modules; during
the active mode, the sensors sense, process and communicate regularly.
A natural method for deciding when a sensor node sleeps and when it
wakes is to probabilistically choose its sleeping times.  However, for
the duty-cycled network to function as it should, we need two
properties: (a) {\em time coverage,} i.e. data generated at any time
must be sensed and relayed by the network, and (b) {\em connectivity,}
i.e. every node should be connected to every other and to the base
station.  The study of the conditions that guarantee these properties
is the focus of this paper.

This problem was initially investigated
in~\cite{das-icuimc:2012}, where it was conjectured that the \dcg\ is
connected if in every time slot the nodes awake form a RGG connected
in that time slot. In this paper we prove this
conjecture, and also show that the radius of connectivity that this
conjecture implies is not optimal i.e. a lower connection radius (and
hence lower transmission power) is sufficient. This lower connection
radius is also shown to
be optimal in the sense that it provides a necessary condition for
connectivity. We call it the {\em optimal radius} as opposed to the {\em
  weak radius} conjectured in~\cite{das-icuimc:2012},  We present two natural duty-cycling schemes, called the
{\em contiguous} scheme and the {\em random selection} scheme that both satisfy
the time coverage property. Apart from being useful duty-cycling
schemes for real applications, these schemes also highlight the
contribution of this paper since they have the same weak radius of
connectivity but very different optimal radii.  
We also show
that if we are willing to spend some preprocessing time in defining the duty-cycle scheme, 
we can compute a \emph{deterministic duty-cycling scheme} that achieves connectivity at the minimum
possible radius, i.e. the \rgg\ radius.

In order to prove the optimal radius result for duty-cycled WSNs we
introduce a new continuum percolation model that we call the {\em
  vertex-based random connection model} which is a natural
generalization of the random connection model as defined
in~\cite{meester:1996}. In this model each node instantiates a random
variable independently of all other nodes and a connection between two
nodes (that are within transmission radius of each other) is made by
computing a function of the random values present at the two nodes. To
explain by example, we could say that each node chooses one colour at
random out of Red, Green and Blue and two nodes that are within
transmission radius of each other are connected only if they both have
different colours. Clearly, in this model edges are not formed
independently since in the example cited above we cannot have a clique
of size 4 since there are only three colors and so there must be at
least one pair of vertices that has the same colour. In other words,
in this example the probability of a 4 clique existing is 0 whereas in
a model where edges are formed independently of each other a clique of
size 4 could form with non-zero probability. To the best of our
knowledge the vertex-based random connection model has not been
studied in this generality before. We present basic results about this
model, including a high density result following Penrose's result for
the simple random connection model~\cite{penrose-aap:1991}, which
allows us to prove the sufficient condition of the optimal radius. 

Finally, our vertex-based random connection model can also be
considered a generalization to a model considered by researchers in
the area of key presharing for secure communication. In the key
presharing setting Eschenauer and Gligor~\cite{eschenauer-ccs:2002}
proposed a scheme in which each node receives a randomly selected
subset of keys and two nodes can communicate if they share a key. This
is similar to our model if we think of time slots as keys and the time
slots that a vertex is awake as those keys assigned to a
vertex. In~\cite{yagan-tit-a:2012}, the author stated a specific
conjecture regarding the connectivity of \rgg s operating with the
Eschenauer-Gligor scheme. Our Theorem~\ref{thm:vb-gamma} settles this
conjecture. Therefore, our contribution is a general and foundational
contribution, as well as a detailed and in-depth study of the
particular setting of duty-cycled WSNs.

A preliminary version of this work has appeared as a short 4 page paper
in the proceedings of ACM MSwim 2013~\cite{bagchi-mswim:2013}.  That version contains none of the proofs presented here
and contains only a few of simulation results of
Section~\ref{sec:simulation}. 

The paper is organized as follows.  Section~\ref {sec:related} relates
the previous work in this area.  Section ~\ref{sec:model}, after
introducing our duty-cycle wireless sensor network view, describes its
model and its challenges.  Section~\ref{sec:weak} introduces the weak
radius, while Section~\ref{sec:strong} presents the optimal radius
defining a new ``vertex-based'' random connection
model. Section~\ref{sec:simulation} highlights the significance of our
results applying them to the two natural contiguous and random
selection duty-cycling schemes.  Finally, in
Section~\ref{sec:deterministic}, we present a method for computing a
periodic duty-cycling scheme that achieves connectivity at the minimum
possible radius, i.e. the \rgg\ radius.  Conclusions and wider
implications of our results are discussed in
Section~\ref{sec:conclusion}.

\section{Related Work}
\label{sec:related}
Duty-cycling is a passive power conservation mechanism widely adopted
in WSNs~\cite{Pantazis07}.  The basic idea of duty-cycling is to
reduce the time a node is idle or spends overhearing unnecessary
activities by turning off radio/sensing sensor modules and thus
putting the node in the so called sleep mode.  \ifrevision {\bf \fi
  Early research on duty-cycling in WSNs considered this technique
  tightly integrated with the design of communication protocols at the
  MAC layer \cite{PolastreHC04,YeHE02,Ye06a}. The S-MAC duty-cycle
  protocol \cite{YeHE02} was proposed to minimize energy consumption
  in battery-powered wireless sensor nodes. B-MAC aims to reduce costs
  due to synchronization in S-MAC by means of long preambles and
  low-power listening~\cite{PolastreHC04}. SCP-MAC is a hybrid
  solution between S-MAC and B-MAC which relies on scheduled channel
  pollings instead of asynchronous preambles~ \cite{Ye06a}. These. and
  other advanced duty-cycle solutions for the MAC layer, are revised
  in the comprehensive A-MAC architecture, proposed in \cite{Dutta10}.
  Subsequently, sleep/wakeup protocols have been implemented at the
  network or application level because they permit a greater
  flexibility and, in principle, can be used with any MAC protocol.
  These latter protocols can be subdivided into three main categories,
  {\em on-demand}, {\em period scheduling}, and {\em asynchronous
    scheme}~\cite{Anastasi09}.  The basic idea behind on-demand
  protocols is that a node should wake up only when another node wants
  to communicate with it. This requires a way to inform the sleeping
  node that some other node is ready to communicate with it.
  Typically in such on-demand schemes multiple radios with different
  energy/performance trade offs (i.e. a low-rate and a low power radio
  for signaling, and a high-rate but more power hungry radio for data
  communication) are used, and thus they require that sensor hardware
  characteristics are adapted to the adopted duty-cycle scheme.  In
  periodic scheduling, nodes wake up according to a wakeup schedule,
  and remain active (listening to the radio) for a short time interval
  to communicate with their neighbors.  Finally, in asynchronous
  sleep/wakeup protocols a node can wake up when it wants and still be
  able to communicate with its neighbors. Both periodic and
  asynchronous schemes must guarantee that nodes are able to
  communicate with neighbors without any explicit information exchange
  among nodes.  Thus, the main challenge in these schemes is to
  guarantee that the network is connected and that there is always a
  sufficient number of awake sensors.  A detailed survey of the
  sleep/wakeup schedules up to 2008 can be found in
  \cite{Anastasi09}. In more recent years, flexible periodic
  duty-cycling schemes have been proposed. These schemes vary the
  length of the awake period to react to external conditions, like the
  amount of energy drained so far or the overall operation latency.
  For example, in \cite{Guo2009}, the authors consider duty cycling in
  an energy harvesting WSN and adapt the length of the awake period to
  the amount of available energy which varies depending on the space
  and time.  In \cite{Ghidini2012}, the flexibility idea is pushed
  even forward by proposing a Markov chain-based duty-cycling scheme.
  In that paper, the authors assume that the sensors are locally
  time-synchronized and feature a common time-slot length, but the
  time-slot length is computed along with other input parameters, like
  working schedule duty cycle and memory coefficient of the
  Markov-chain process, so as to improve the network efficiency while
  keeping a constant connection delay, or to improve connection delay
  yet not negatively affecting efficiency.

Many other works in the literature address specific communication operations, including localization, one-to-all communication, data dissemination
and collection, in duty-cycled wireless sensor networks. 
%In this paper, we concentrate on  periodic scheme, that alternates between sleep and awake periods.
%The duty cycle ratio is defined as the ratio of  the length of the  awake periods to the cycle length
%and gives an indication on how much a sensor spends in the active period.
For example, in \cite{TCS08,TPDS09,TMC10}, the benefit of duty-cycling 
is studied for training duty-cycle sensors
to learn their position with respect to a central sink either when the sensors  cooperate amongst themselves or 
when the sensors adopt different periodic duty-cycle schemes.  The length of the cycle and the
length of sensor awake period are analytically determined in such a way
that the energy consumed during the training process is minimized and
all the sensors are guaranteed to learn their position.  

For networks to function as they should under power conservation mechanisms, 
connectivity needs to be maintained. Such power conservation mechanisms 
generally exploit WSN redundancy to extract a subset of active sensors that form a connected communication graph, 
like a near-optimal dominating set or a sub-optimal broadcast-tree,
which guarantees network functionalities 
~\cite{Zhou:2009,survey:2012}.
%Many solutions focus on developing (distributed) algorithms to construct a near-optimal
%connected dominating set or sub-optimal broadcast tree using fixed or variable transmission radius~\cite{Zhou:2009,survey:2012}.
%
To the best of our knowledge, in \dcg s, the problem of maintaining connectivity by constructing near-optimal communication graphs has been addressed 
in few papers~\cite{XWXZ10,WL11,LaiR13,Han13}. %Guo09,AWK11,
However,
in those works, the  assumptions of their models, like the sensors density in the network, or the type of duty-cycle, substantially differ from the environment we deal with.
%In~\cite{Guo09}, the problem of maintaining an energy-optimal broadcast tree is addressed for DC-WSNs.
%The key novelty lies in the forwarding decision that
%preserves the reliable nature of traditional flooding
%by allowing packets to opportunistically travel along the links outside
%the broadcast tree to reduce both flooding delays and the level of
%redundancy. This method does not directly apply to dense, large-scale WSNs.
In~\cite{XWXZ10}, the one-to-all and the all-to-all paradigms have been addressed in DC-WSNs. However, sensors can transmit messages at any time, not only when they are active (awake), and the duty-cycle is considered only with respect to the receiving capabilities.  %, i.e., a sensor can receive only when it is active.
In~\cite{WL11}, the broadcast problem in DC-WSN with unique identifiers is shown to be equivalent to the shortest path problem in a time-coverage graph, and accordingly an optimal centralized solution has been presented. %The obtained solution has also motivated a distributed implementation that has been experimentally examined under diverse network configurations. The investigated approach may result unfeasible when scaling to dense WSNs as we do consider.
%Recently, in~\cite{AWK11}, experimental studies on enhanced heuristics for solving the broadcast in DC-WSNs with unique identifiers have been presented. The aim of the authors was to prolong the network lifespan by evenly distributing the energy consumption required to accomplish the broadcast, hence improving the opportunistic flooding method presented in~\cite{Guo09}. Other aspects of DC-WSNs have been discussed in \cite{Ghidini2011} and \cite{Ghidini2010}.
In \cite{LaiR13}, the problem of least-latency end-to-end routing over asynchronous and heterogeneous \dcg s\   
is modeled as the time-dependent Bellman-Ford problem. 
In \cite{Han13}, the minimum-energy multicasting problem is studied in duty-cycle wireless sensor networks
again modeling the network as an undirected graph. 
These last three investigated approaches may result in infeasible solutions when scaling to dense 
WSNs i.e. the case that we consider in this paper.
\ifrevision
{\em [revision id \# 11]}}
\fi

Seminal work for connectivity in the area of scaling radio networks
whose sensors are uniformly and at random placed over a unit area are
reported in~\cite{gupta-cdc:1998,gupta-chapter:1999}. The authors
study scaling laws for connectivity when the sensors are always awake
(i.e., no duty-cycle) and use a result of
Penrose~\cite{penrose-aap:1991} to show that the \rgg\ is connected
with probability tending to 1 as $n \rightarrow \infty$ if and only if
\begin{equation}
\label{eq:radius}
\pi r(n)^2 = (\log n + c(n))/n, \mbox { where} \lim_{n \rightarrow
  \infty} c(n) = \infty. 
\end{equation}
Obviously, these results do not directly apply to the duty-cycle scenario.
Nonetheless, we will use them to prove the conjecture in
\cite{das-icuimc:2012}, where a preliminary
study of connectivity in uniformly and randomly distributed DC-WSNs
was initiated by modeling the DC-WSNs as a temporal series of random geometric graphs of
only awake sensors.

Gupta and Kumar also conjectured
~\cite{gupta-cdc:1998,gupta-chapter:1999} that if each edge between
two vertices that are at most $r(n)$ apart is formed independently
with probability $p(n)$ then connectivity can be obtained if $\pi
r(n)^2 p(n) = (\log n + c(n))/n$ and $c(n) \rightarrow \infty$ as $n
\rightarrow \infty$. This conjecture was recently proved in a slightly
more general setting in~\cite{mao-tit:2013}, but always assuming that
edges are formed independently. Although in our \dcg s\ the edge
between two sensors are formed with a certain probability which is the
same for all pairs of sensors within transmission radius, 
the edges are not formed independently. Hence, the results
in~\cite{mao-tit:2013,yi-tcomm:2006} do not apply to our
case. Moreover, the necessary condition of the Gupta and Kumar's
conjecture had been earlier proved by Yi et. al.~\cite{yi-tcomm:2006}
who used a geometric approach which was closer in spirit to the
approach used by Gupta and Kumar themselves to prove a result about the
distribution of isolated nodes. We will say more about the technique
used in~\cite{yi-tcomm:2006} in Section~\ref{sec:strong:main-theorem}
and using some key aspects of it we prove the necessary condition
in Theorem~\ref{thm:vb-gamma}.   In addition, in~\cite{yi-dmaa:2010} a slightly more general model that
includes random independent node removals is studied and a result on
the distribution of isolated nodes similar to that
of~\cite{yi-tcomm:2006} was shown. However this model too falls short
of the generality of our vertex-based random connection model.

Connectivity has also been studied  in~\cite{Shakkottai:2003} for the random grid model
assuming that each sensor fails with independent probability $1-p(n)$. 
Although the failure probability may be referred to the duty-cycle ratio,
this paper does not require that  all nodes  be connected, but only the nodes that are active at a certain time. 
This model is different from ours. Besides, their sufficient condition for
connectivity is weaker than ours, even in
the grid case (which is considered easier to analyze than the uniformly
distributed case).
 
\iffalse
See their Prop 4.1 equation (16). Note
that for this condition to be satisfied for general p(n), we must have
that

\pi p(n) r(n)^2 = 2 (log(n) + c(n))/n

where c(n) tends to infinity as n tends to infinity. The extra factor of 2
is because of the 2 in the denominator of the exponential in (16). So,
their sufficient condition is weaker than ours by a factor of 2 even in
the grid case (which is considered easier to analyze than the uniformly
distributed case).
\fi
 
Finally, our model can be considered a generalization of the so called
key graph of the Eschenauer-Gligor scheme, which can be seen as an
intersection of a random geometric graph with an Erd\"os-Renyi
graph~\cite{eschenauer-ccs:2002} .  Our Theorem~\ref{thm:vb-gamma}
settles a specific conjecture for connectivity stated for such graphs
in ~\cite{yagan-tit-a:2012}, thus improving on the connectivity
conditions previously known for such a model.  Moreover, our result
uses a new continuum percolation model, which does not spring from the
usual techniques applied for the Eschenauer-Gligor scheme.

\section{Modeling Dense duty-cycled wireless sensor networks}
\label{sec:model}
%CRISTINA
%In this section, 
%we describe our foreseen of dense wireless sensor networks employed, for example, in space exploration and condition based monitoring applications.
%
%We envision a large population of tiny low-cost sensors which follow a periodic sleep-awake cycle,
%sense the physical world in their proximity, and communicate by a wireless multihop network. 
In this section, we describe the network setting that we are studying
(Section~\ref{sec:model:network}) and then explain the graph model
by which we try to capture the properties of the setting that are
relevant to the study of connectivity under the family of duty-cycling
schemes we consider (Section~\ref{sec:model:model}.)

\ifrevisiontwo
{\bf 
\fi
\subsection{The network setting}
\label{sec:model:network}
In our view, {\em duty-cycled wireless sensor networks} (DC-WSNs)
consist of a large population of tiny, anonymous, mass produced
commodity sensors, uniformly and randomly deployed on a vast
geographical area, perhaps via an unmanned vehicle. 
%CRISTINA
The sensors must work unattended for long periods of times.
They can be either provided with a limited and nonrenewable power supply
or with a limited and rechargeable power supply.
Each sensor is equipped with a processing unit, a sensing unit, a
short-range radio transceiver and, if it applies, with a circuit to harvest the energy.
 In order to save or store energy, the sensors
follow a periodic pattern of sleeping and waking, known as a {\em duty
  cycle}. When a sensor sleeps, only its internal clock and its timer
are on.  During the awake periods, the sensors sense in their
proximity and, if required, they process the collected data and send
radio messages.
%CRISTINA

%In our setting no centralized or distributed algorithm is deployed to
%create a connected network of sensors. We aim to study the situation
%where at deployment time sensors opportunistically make connections
%with every other sensor that they are able to communicate with. Hence
%our scheme for network creation is a very simple and greedy scheme
%involving making all possible connections. At deployment time each
%sensor sends requests to connect at every waking slot and handshakes
%with those neighbors who are close enough to receive and transmit and
%also awake for long enough to communicate meaningfully. While this
%scheme may appear overly simple it has the major advantage of not
%incurring any computational overhead. The goal of our paper is to give
%bounds on the transmission power (expressed as transmission radius)
%that allow densely placed sensors operating with a periodic sleep
%schedule to form a connected network.

We assume that just prior to the deployment (perhaps onboard of the
vehicle that drops them in the terrain), the sensors are provided with
the parameters required to set a functioning network.
%their duty-cycle  and its transmission range radius. 
As will be discussed later in this paper, the sensors need to know the
total number $n$ of sensors deployed, the adopted periodic duty-cycled
scheme along with its period \emph{length} $L$, the \emph{number of
  waking slots} $d$ where $d < L$, and the probability $\gamma$ that
two nodes within the range of transmission can communicate.  In fact,
radio messages sent by a sensor can reach only the sensors in its
immediate proximity that are awake at transmission time.  Namely, only
a fraction $\gamma$ of the overall sensor population in the sensor
proximity can hear the radio message.

Moreover, each sensor is provided with a standard public domain
pseudo-random number generator, which is used for generating the
random information of the selected duty-cycle scheme, and with an
$L$-bit register $R$ where the generated duty cycle scheme is
memorized.  Precisely, each bit in $R$ represents a time slot of the
period, and it is set to $1$ if the sensor is awake and $0$ otherwise.
To make this description more concrete, we will consider how the awake
period is selected in the two duty-cycle schemes presented in
Section~\ref{sec:simulation}.  The first scheme, called the {\em
  contiguous} model has been studied in~\cite{das-icuimc:2012}. For
this scheme, by means of the pseudo-random number generator, each
sensor $u$ independently chooses an integer $i_u$ from the set
$\{0,1,\ldots,L-1\}$ and it sets the entries of register $R$ from
$i_u$ to $i_u + d-1$ to $1$ because it is awake for $d$ consecutive
time slots.  The remaining entries of $R$ are set to $0$ because the
sensors sleeps. In the following, we will denote this model \dcc.  The
second scheme, called the {\em independent random selection} model, is
one in which each node chooses the set of the awake $d$ time slots at
random from $\{0,1,\ldots,L-1\}$ and sets these entries of $R$ to $1$.
We will use the notation \dcr\ to refer to this scheme.
%CRISTINA
Note that the \dcc\ well models a network of rechargeable sensors,
while the \dcr \ a network of sensors equipped with unrenewable energy. 
%CRISTINA

In addition, before deployment the sensors receive an {\em initial
  time}.  On the terrain, at the initial time, all the clocks are
synchronous and share the same slot length.  The sensors follow their
periodic scheme in a totally distributed way.
%Although the sleeping and waking patterns of different sensors
%may be different, since all the sensors share the same period $L$, each duty cycle 
%begins and ends at the same time for all the nodes.
Each sensor computes the time slot number using its internal clock and
autonomously follows the duty-cycle scheme memorized in its register
$R$.  Specifically, the sensor indexes the register $R$ by the time
slot number modulo $L$ and stays awake if $R$ is equal to $1$ or goes
to sleep if $R$ is equal to $0$.  As long as the clocks remain
synchronized, the time slot number is the same for all sensors, each
duty cycle begins and ends at the same time for all the nodes, while
the sleeping and waking patterns of different sensors may be
different.
%
%stays awake if the entry of $R$ indexed by the time slot 
%number modulo $L$ is $1$, and it goes to sleep otherwise. 
During the network lifetime, due to clock drift, synchronization may
become weaken and it may happen that the sensors no longer share the 
 same time slot length or the same time slot number.
%This becomes a trouble only if it affects the overlapping probability of the awake periods.  
Nonetheless, this can be tolerated as long as the probability $\gamma$ that two sensors communicate remains the same. 
Hence, we conclude that synchronization is {\em not} critical
for our study. Our results hold for synchronized and non-synchronized
settings. The only thing that the entire of family of duty-cycle
schemes that come in our ambit require is that there should be a
well-defined probability $\gamma$ for the event that two nodes within
transmission range of each other overlap in such a way that they can
communicate. If this is defined and is the same for all pairs of nodes
then our results hold.

%CRISTINA: I believe it is redundant now
%While on the terrain, the fundamental goal of the DC-WSN is to produce
%globally meaningful information from local raw data to be stored
%inside the WSN or to be transferred to the outside world.  Namely,
%during the awake periods, each sensor collects data through its
%sensing capabilities and, if the raw data need to be routed to the
%final destination, the sensor transmits them in its proximity by means
%of the radio.  In their turn, the awake sensors that receive the
%message forward it, hence building multihop wireless paths.  It is
%worth noting that the sensors involved in such paths are those that
%receive the radio transmission, and they are not selected based on
%location criteria.  Hence, no localization algorithm is required for
%establishing the wireless paths. It is a separate matter that the
%sensing application may itself in many cases require localization so
%that the raw data can be associated with the location from which it is
%collected.
%CRISTINA

In our setting no centralized or distributed algorithm is deployed to
create a connected network of sensors. We aim to study the situation
where at deployment time sensors opportunistically make connections
with every other sensor that they are able to communicate with. Hence
our scheme for network creation is a very simple and greedy scheme
involving making all possible connections. At deployment time each
sensor sends requests to connect at every waking slot and handshakes
with those neighbors who are close enough to receive and transmit and
also awake for long enough to communicate meaningfully. 
While this
scheme may appear overly simple it has the major advantage of not
incurring any computational overhead and offering a basic communication mechanism. 
Moreover, it is
worth noting that no localization algorithm is required for
establishing the wireless multihop communications. It is a separate matter that the
sensing application may itself in many cases require localization so
that the raw data can be associated with the location from which it is
collected.

The goal of our paper is to give
bounds on the transmission power (expressed as transmission radius)
that allow densely placed sensors operating with a periodic sleep
schedule to form a connected network. Once this prequisite of any communication
protocol is guaranteed,
communication protocols can be
adopted for optimizing the routing process.
However, this is  a separate matter, not studied in this paper.

%Then, for optimizing the routing process, communication protocols can be
%adopted.  However, in this paper, we do not design any specific
%communication protocol.  We instead study under which duty-cycle and
%radius transmission conditions, the interconnected multihop dense
%DC-WSN is connected, which is a prerequisite for any communication
%protocol.  

In the following, we model the DC-WSNs and formalize the
challenges we encounter.

\ifrevisiontwo
}
\fi

\subsection{The duty-cycled graph model}
\label{sec:model:model}
We first define our notation and the model
following~\cite{das-icuimc:2012}. A {\em random geometric graph}
$\rgg(n,r)$ is a graph with vertex set $V$ of $n$ points distributed
uniformly at random in the unit circle centred at the origin. 
\ifrevision
{\bf 
\fi
These
points model the sensor nodes distributed randomly through the area of
interest. 
\ifrevision
{\em [revision id \#1]}}
\fi
We also
superimpose one point at the origin itself. There are edges between
any two $u,v \in V$ such that $d(u,v) \leq r$ where $d(\cdot,\cdot)$
is a distance metric defined on $\RR^2$. 
\ifrevision
{\bf 
\fi
The quantity $r$ models the
transmission radius of the sensor nodes.
\ifrevision
{\em [revision id \#1]}}
\fi

%Although the sleeping and waking patterns of different sensors
%may be different, they all share the same period, and each cycle of
%this periodic pattern (known as a {\em duty cycle}) begins and ends at
%the same time for all the nodes.  
As described in~\cite{das-icuimc:2012}, the primary parameters of the periodic duty-cycle
are $L$, the \emph{length} of the duty cycle, and
$d$, the \emph{number of waking slots} where $d < L$. We use the
notation $\delta = \lceil d/L \rceil$ to indicate the {\em duty-cycle
ratio}, which is a measure of the energy spent by each sensor in
each cycle. 
In addition,
we provide a more general definition of a duty-cycled
graph than that given in~\cite{das-icuimc:2012}. Each sensor $u$
chooses its waking slots which we denote by the set $A_u$ where $A_u
\subseteq \{0,1,\ldots,L-1\}$ and $|A_u|=d$. Given a scheme $\CA$ for
choosing these waking slots, we define the {\em duty-cycle graph}
$\dcg_{\CA}(n,r,\delta,L)$ as follows: it has the same vertex set as
$\rgg(n,r)$ and its edge set is: $E = \{(u,v) : d(u,v) \leq r, A_u
\cap A_v \ne \emptyset\}$. Namely, for two vertices $u$ and $v$ that
are within transmission range of each other to be connected, they must
share a slot where they are both awake. 
Specifically, in the previously introduced \dcc s, sensor $u$ chooses $A_u= i_u, \ldots, i_u+d-1$,
where $i_u$ is a random number from the set
$\{0,1,\ldots,L-1\}$; whereas, in \dcr s,  the set
$A_u$ is a set of size $d$ selected at random in $\{0,1,\ldots,L-1\}$.

\subsubsection{Connectivity}
\label{sec:model:model:connectivity}

As explained, a fundamental property desired of any duty-cycled sensor network is
connectivity. More precisely, connectivity means that {\em it should
  be possible to send data generated at any time at any node to any
  other node in the network (within reasonable time)}. 

We make a simple observation about connectivity. 
\begin{fact}
\label{fct:deltahalf}
If $\delta > 1/2$ then
$\dcg(n,r,\delta,L)$ is connected whenever $\rgg(n,r)$ is connected.
\end{fact}
\ifrevision
{\bf 
\fi
To see why this is the case note that whenever the awake period of
each sensor is (strictly) more than half the duty cycle then each edge
of the original graph $\rgg(n,r)$ is available for at least one time
slot because any two waking periods must, by the Pigeonhole Principle,
share a slot.
\ifrevision
{\em [revision id \#3]}}
\fi

However, for $\delta \le 1/2$, connectivity is not guaranteed under
our current definition. Not only is it a random event whose
probability needs to be determined, it may also be an event which
occurs with probability 0. Consider a scheme where $d=3, L=10$ and
each node chooses either $A = \{0,1,2\}$ or $B = \{3,4,5\}$ as its set
of waking slots (with probability 1/2 each, independently of all other
nodes). In this scheme all the nodes with waking cycle $A$ can {\em
  never} communicate with all the nodes which have waking cycle
$B$. Hence we need a condition on the duty-cycling scheme. We call
this the reachability condition.

\paragraph{The reachability condition} Consider a scheme $\CA$ for
selecting the waking slots of nodes. Given the set $\CL = \{A : A
\subset \{0,1,\ldots,L-1\}, |A| = d\}$, let us denote by $\CL(\CA)$
  all those subsets in $\CL$ that have non-zero probability of being
  selected as a waking schedule for a node. Then, the reachability
  condition on $\CA$ is the following:

\MyQuote{
\label{cond:reachability}
{\em Reachability}. There is a finite $k \geq 0$ such that for any
$A_1, A_2 \in \CL(\CA)$, there exists a sequence $B_0,\ldots,B_k$
where $B_0$ is $A_1$ and $B_k$ is $A_2$ such that $B_i \in \CL(\CA), 0
\leq i \leq k$, and $B_i \cap B_{i+1} \ne \emptyset, 0 \leq i < k$.  }
Clearly, as the example above shows, the reachability condition is
necessary for connectivity.

It is easy to see that for the contiguous
duty cycle scheme, given two nodes whose duty cycles begin at $i$ and
$j$, it is possible to find a chain of overlapping awake cycles
beginning at $i + m(d-1)$, $1 \leq m < \lceil (j-i)/(d-1) \rceil$. And
since the probability of picking an awake cycle beginning at $i +
m(d-1)$ for the relevant value of $m$ is greater than 0 (in fact
$1/L$), the reachability condition is easily satisfied. A similar
argument can be made for the random selection scheme, where, in fact
$k = 2$ is sufficient since for any two non-overlapping awake periods
we can always pick a third awake period that overlaps with both with
non-zero probability.

\subsubsection{Time Coverage}
\label{sec:model:model:coverage}

Since the sensor network's primary function is to sense data from the
environment, it is essential that a sleep schedule should keep a
significant fraction of the sensors awake at any time point in such a
way that the area being sensed is covered. This is different from the
notion of spatial coverage which is widely studied in the literature:
there the problem is to ensure that a static set of randomly
distributed sensors is able to sense each point in the region of
interest. For us the notion of time coverage is this: {\em a
  significant fraction of the nodes of the network should be awake in
  each time slot.}

Since we primarily work with probabilistic duty-cycling schemes, we
state the coverage requirement in probabilistic terms.
\begin{quote}
{\em Time coverage}. For each $k \in \{0,1,\ldots,L-1\}$, the probability
that a node $u$ is awake in slot $k$ is $\delta_k > 0$, where
$\delta_k$ may be a function of $d$ and $L$ but is not dependent on
the number of nodes in the network.
\end{quote}
Since each node in the network remains awake for $d$ out of $L$ slots
we can also think of a stronger condition on the duty-cycling scheme
which ensures symmetry across all the slots in $\{0,1,\ldots,L\}$.
\begin{quote}
{\em Uniform time coverage}. For each $k \in \{0,1,\ldots,L-1\}$, the probability
that a node $u$ is awake in slot $k$ is $\delta = \lceil d/L \rceil$.
\end{quote}
We will see that the distinction between coverage and uniform coverage
plays a role in determining the connectivity radius.

\section{A weak connectivity result for duty-cycled WSNs}
\label{sec:weak}

In this section we will show that Gupta and Kumar's result on
connectivity in high density random geometric graphs gives us a
condition on the transmission radius of a node that is sufficient to
achieve connectivity. The main result of this section,
Theorem~\ref{thm:conn-weak} is a generalization of the result first
presented in Das et. al.~\cite{das-icuimc:2012}. 
\ifrevision
{\bf 
\fi
We note that
Theorem~\ref{thm:conn-weak} generalizes the earlier theorem that was
proved in the previous paper for only \dcc\ to a whole family of
duty-cycling schemes restricted only by the reachability and coverage
conditions. The former is necessary for connectivity, as
discussed above, and the latter is necessary for the sensing
application to not drop any data.
\ifrevision
{\em [revision id \#5]}}
\fi
\begin{theorem}
  \label{thm:conn-weak} Given a duty-cycling scheme $\CA$ with $0 < \delta \leq 1/2$ and  $d = \lceil \delta  L\rceil > 1$, and the marginal probability of a node being awake in slot $i$ denoted by $\delta_i$, the probability that $\dcg_{\CA}(n,r(n),\delta,L)$ is connected tends to 1 as $n \rightarrow \infty$ if $\CA$  satisfies the reachability condition and the
  coverage condition and if
\begin{equation}
\label{eq:das-radius}
\pi r^2(n) \delta_{\min} = (\log n + c(n))/n,
\end{equation}
such that $c(n) \rightarrow \infty$ as $n \rightarrow \infty$, where
$\delta_{\min} = \min_{k=0}^{L-1} \delta_k$.
\end{theorem}

Before we prove this theorem we note that the form of this result is
non-trivial, especially the role of the quantity $c(n)$. We provide a
discussion of the role of $c(n)$ in
Section~\ref{sec:strong:main-theorem} right after the statement of
Theorem~\ref{thm:vb-gamma} which is, in our view, the appropriate
place for this discussion.

\begin{proof}
We prove the theorem by considering a set of $L$ subgraphs of
$\dcg_{\CA}(n,r,\delta,L)$, one for each time slot in a typical duty
cycle. Let us denote these by $\CG_i, 0 \leq i < L$. To be clear, the
vertices of $\CG_i$ are $V_i = \{u : u \in V, i \in A_u\}$ i.e. the vertices
that are awake in time slot $i$, and the edges of $\CG_i$ are $E_i =
\{(u,v) : u,v\in A_i, d(u,v) \leq r\}$. 

The scheme of the proof is as follows. We will first show that if the
condition given in the theorem holds then each $\CG_i$ is connected
with probability $1 - o(1)$. However
this is not enough because it may be that there is some time slot $j$
such that $V_j$ and $V_{j+1}$ are completely disjoint leading to a
partition in the graph. To complete the proof we will show that this
happens with probability $o(1)$.

Consider the vertex set $V_i$ of $\CG_i$. By the coverage condition,
if $|V| = n$ then $\ex(|V_i|) = \delta_i n$. Using Gupta and Kumar's
result on connectivity~\cite{gupta-cdc:1998}, it is clear that for the
case that $|V_i| \geq \delta_i n$, then subgraph $\CG_i$ is connected
with probability tending to 1 as $n \rightarrow \infty$ if
condition~(\ref{eq:das-radius}) is satisfied, since $\delta_i \geq
\delta_{\min}$ by definition. This can be seen by mechanically
substituting $\delta_i n$ in place of $n$ in Gupta and Kumar's theorem
and observing that since $\delta_i$ is a constant w.r.t. $n$, so
$(\log \delta_i)/n \rightarrow 0$ as $n \rightarrow \infty$.

Now we note that the probability $|V_i| < \delta_i n$ tends to 0 as $n
\rightarrow \infty$. This is a straightforward application of the law
of large numbers, but we formalize it anyway using Chernoff bounds:
For each $u \in V$ we define an $X_u^i$ which takes value 1 if $u \in
V_i$ and is 0 otherwise. Hence:
\[ |V_i| = \sum_{u\in V} X_u^i,\]
and, by Chernoff bounds, for any $\epsilon > 0$,
\[ \pr(|V_i| < (1 - \epsilon) \delta_i n ) \leq  e^{-(\epsilon^2\delta_i
  n)/2},\] which tends to 0 as $n \rightarrow \infty$. 
\ifrevision
{\bf 
\fi
We note that
Chernoff bounds are applicable in this case since each node chooses
its awake cycle independent of all other nodes, and hence for any
given $i$, the probability that $u$ is awake at time slot $i$ is
independent of the corresponding event for all other nodes and so the
collection of random variables $\{X_u^i : u \in V\}$ is an independent
collection.
\ifrevision
{\em [revision id \#4]}}
\fi

Suppose we denote by $H_i$ the event that $\CG_i$ is connected. Since
\[\pr(H_i^c) \leq \pr(H_i^c \mid |V_i| \geq \delta_i n)\cdot \pr(|V_i|
\geq \delta_i n) + \pr(|V_i| < \delta_i n),\] 
we get that $\pr(H_i^c) = o(1)$ since Gupta and Kumar's theorem tells
us that the first term goes to 0 and the Chernoff bound argument tells
us that the second term vanishes as $n \rightarrow \infty$.

Now let us define the event that there is time partitioning among the
$\CG_i$s. Let $F_i^c$ be the event that $V_i \cap V_{i+1 \mod L} =
\emptyset$. Recall that we denote by $\CL(\CA)$ the set of all subsets
of $\{0,1,\ldots,L-1\}$ that have non-zero probability of being chosen
as a waking cycle under the duty-cycle scheme $\CA$. Consider the sets
$\CC_i, \CC_{i+1} \subset \CL(\CA)$ such that $i$ is in each of the
sets in $\CC_i$ and $i+1$ is in each set of $\CC_{i+1}$. Note that the
nodes that choose their waking schedule from $\CC_i$ are precisely the
nodes of $V_i$. By the coverage condition these $\CC_i$ and
$\CC_{i+1}$ are non-empty. The reachability condition guarantees that
for every $A \in \CC_i$ and every $B \in \CC_{i+1}$ there is a $k$ and
$B_1, \ldots, B_{k-1}$ such that all the $B_i$ belong to $\CL(\CA)$
and $A \cap B_1 \ne \emptyset$, $B_i \cap B_{i+1} \ne \emptyset, 1
\leq i \leq k-1$ and $B_{k-1} \cap B \ne \emptyset$. From this
condition we can deduce the following:
\begin{claim}
\label{clm:indices}
There is a sequence of indices $i = j_0, j_1, \ldots j_l = i+1$ such
that for every $0 \leq k < l$
there is an $A \in \CC_{j_k}$ and a $B \in \CC_{j_{k+1}}$ such that
$A \cap B \ne \emptyset$. 
\end{claim}
We can build this sequence of indices constructively. Take any set
from $\CC_i$ and find a $B_1$ with respect to any set in $\CC_{i+1}$
as given by the reachability condition. Choose any index from $B_1
\cap B_2$ and call it $j_1$. Similarly pick an index from $B_2 \cap
B_3$ and call it $j_2$ and continue all the way till we reach $B_k \in
\CC_{i+1}$. 

Since we have seen earlier that under the
condition~(\ref{eq:das-radius}) each $\CG_i$ is connected with
probability tending to 1 as $n \rightarrow \infty$, and that the nodes
that choose their waking schedules from $\CC_i$ are exactly the nodes
$V_i$, by the definition of $\CC_i$, hence the implication of the
claim is that if the sequence of subgraphs $\CG_{j_0},
\CG_{j_1},\ldots, \CG_{j_l}$ are connected to each other then there is
a path from $\CG_i$ to $\CG_{i+1}$ in
$\dcg_{\CA}(n,r,\delta,L)$. Hence the probability that $V_i$ and
$V_{i+1}$ are disconnected is upper bounded by the probability that
there is an $i$ such that $\CG_{j_i}$ is disconnected from
$\CG_{j_{i+1}}$. Since the sequence $\{j_i\}_{i=0}^l$ is constructed
using overlapping schedules, the disconnection of $\CG_{j_i}$ and
$\CG_{j_{i+1}}$ can happen if either (a) one of $\CG_{j_i}$ or
$\CG_{j_{i+1}}$ are disconnected, or (b) none of the nodes of
$V_{j_i}$ that choose $B_i$ as their waking schedule have a node of
$V_{j_{i+1}}$ as neighbor that chooses $B_{i+1}$ as its schedule
(denote this event $M_{j_i}$).  We have already shown that the
probability of (a) is $o(1)$. So let us consider the event $M_{j_i}$.

Let us denote the probability of $B_i$  being chosen as a waking
schedule by $\beta_i$ for all relevant $i$.
The event $M_{j_i}$ occurs if none of the nodes of $V_{j_i}$ that
choose $B_i$ have a
neighbour in $V_{j_{i+1}}$ that chooses $B_{i+1}$. Now if a node of $V_{j_i}$ has $k$
neighbours then the probability that none of them chooses $B_{i+1}$
is $(1 - \beta_{i+1})^k$. Since $\beta_i > 0$, there is with
probability tending to 1 as $n \rightarrow \infty$ at least one node
$u \in V_{j_i}$ that chooses $B_i$ as its waking schedule. We denote
by $\Gamma_u$ the set of points of $V$ that lie within distance
$r(n)$ of $u$. Denote by $V^{B_{i+1}}$ those nodes of $V$ that
choose $B_{i+1}$ as their waking schedule.

Now, conditioning on the size of $\Gamma_u$ and using a Chernoff bound
argument to upper bound the probability of $|\Gamma_u|$ being smaller
than its expected value as done above we get 
\begin{eqnarray*}
\pr(M_{j_i}) & \leq&  \pr\left(\Gamma_u \cap V^{B_{i+1}} = \emptyset  | |\Gamma_u| \geq \frac{\log(n) + c(n)}{\delta_{\min}}\right)
\\ 
& & + \ o(1)\\
& \leq & (1 - \beta_{i+1})^{\frac{\log(n) + c(n)}{\delta_{\min}}} +
o(1)\\
& \leq & \exp - \left\{ \frac{ \beta_{i+1} \cdot (\log(n) +
  c(n))}{\delta_{\min}}\right\} + o(1),\\
\end{eqnarray*}

which is $o(1)$ since $\beta_{i+1}$ and $\delta_{\min}$ do not depend
on $n$. Also since the index $l$ given in
Claim~\ref{clm:indices} is constant with respect to $n$ (it depends
only on $d$ and $L$), we have shown that the probability that
$M_{j_i}$ occurs for any $i$ such that $0 \leq i \leq l$ is $o(1)$.
\end{proof}
In the uniform time coverage situation, i.e. for duty-cycling schemes like
the contiguous and the  random  schemes,
Theorem~\ref{thm:conn-weak} yields the following corollary:
\begin{corollary}
  \label{cor:conn-weak} For
  any $0 < \delta \leq 1/2$ and $L > 0$ such that $d = \lceil \delta
  L\rceil > 1$, the probability that $\dcg_{\CA}(n,r(n),\delta,L)$ is
  connected tends to 1 as $n \rightarrow \infty$ for a duty-cycling
  scheme $\CA$ that satisfies the reachability condition and the
  uniform time coverage condition if
\begin{equation}
\label{eq:unif-das-radius}
\pi r^2(n) \delta = (\log n + c(n))/n,
\end{equation}
such that $c(n) \rightarrow \infty$ as $n \rightarrow \infty$.
\end{corollary}

\section{A strong connectivity result for duty-cycled WSNs}
\label{sec:strong}

In this section we develop and present our optimal connectivity result
for duty-cycled WSNs i.e. the most important contribution of our
paper. Proving this result involves defining a new ``vertex-based''
random connection model and claiming certain properties for it. We
begin by motivating the need for this new definition.

\subsection{Stochastic domination and the duty-cycling graph}
\label{sec:strong:stochastic}

\ifrevision
{\bf 
\fi
It is natural to believe that the connectivity properties of the
Gupta-Kumar graph are sufficient to prove the stronger theorem that we
want.  Let us consider the following simple generalization of the Gupta-Kumar
graph that was mentioned in~\cite{gupta-chapter:1999}: Given the
random geometric graph $\rgg(n,r)$ and a parameter $\gamma$ such that
$0 < \gamma < 1$, retain each edge of $\rgg(n,r)$ with probability
$\gamma$ independent of all other edges. Let us denote this model
$\rgg(n,r,\gamma)$. 

In $\rgg(n,r\gamma)$, just like in the duty-cycling graph, it is not
necessary that two nodes that are within transmission range of each
other are able to communicate. They can communicate with a probability
$\gamma$. Hence it is tempting to believe that by choosing the correct
values of $\gamma$ we can use the properties of the generalized
Gupta-Kumar graph to determine under what necessary and sufficient
conditions the duty-cycling graph is connected. However, that would
require us to be able to compare the probability of certain events
(like the event that a subset of nodes is isolated) across the two
models. In general this is done using the theory of stochastic
domination that allows us to compare probabilities of classes of
events across two probability spaces.  But the interesting thing here,
which pushed us to define the vertex-based random connection model
separately, is that there is no stochastic domination between these
models, and hence we cannot use what we know about $\rgg(n,r,\gamma)$
to tell us what we need to know about $\dcg(n,r,\delta,L)$. We
document the details of this stochastic non-domination now. The reader
who does not want to be weighed down by the formal proof can skip the
rest of Section~\ref{sec:strong:stochastic}.  
\ifrevision {\em
  [revision id \#7]}} \fi

First, we recall the definition of stochastic domination. Suppose we
have a lattice $\Omega$ whose partial order is $\preceq$. A function
$f : \Omega \rightarrow \RR$ is called increasing if $f(\omega_1) \leq
f(\omega_2)$ whenever $\omega_1 \preceq \omega_2$, for all $\omega_1,
\omega_2 \in \Omega$. Now, suppose we have a measurable space
$(\Omega,\CF)$. For two probability measures $\mu_1$ and $\mu_2$
defined on this space, we say that $\mu_1$ stochastically dominates
$\mu_2$, denoted $\mu_2 \preceq \mu_1$, if $\ex_{\mu_2}(f) \leq
\ex_{\mu_1}(f)$ for all increasing functions $f$ where $\ex_\mu(f)$
denotes the expectation of the function $f$ under measure $\mu$. 

Suppose we denote the probability measure defined on
$\rgg(n,r,\gamma)$ as $\mu_\gamma$ and the probability measure on
$\dcg(n,r,\delta,L)$ as $\mu_\delta$. Consider the event
$A(u_1,\ldots,u_k)$ to be the event that all the points $u_1, \ldots,
u_k$ are isolated (i.e. have no edges incident on them). This is a
decreasing event in the sense that $- I_{A(u_1,\ldots,u_k)}$ (i.e. the
negative of the indicator function of the event) is an increasing
event. 
\ifrevision
{\bf 
\fi
To see why this is the case we need to understand the lattice
structure of the space on which these graphs are defined. Note that
every configuration contains a set of points and some edges between
these points i.e. each configuration $\omega$ can be described by a
tuple $\omega = (V,E)$. We define a relation $\preceq$ as follows:
$\omega = (V,E) \preceq \omega' = (V', E')$ if $V \subseteq V'$ and $E
\subseteq E'$. Now it is easy to see that if $u_1, \ldots, u_k$ are
isolated in configuration $\omega'$ then they must be isolated in
configuration $\omega$ whenever $\omega \preceq \omega'$. If some of
$u_1, \ldots, u_k$ do not exist in $\omega$ then they can trivially be
assumed to be isolated since they have no neighbors.
\ifrevision
{\em [revision id \#6]}}
\fi

Now consider the case where
we have $k > 1/\delta + 1$ points, and they are all within distance
$r$ of each other. In this case $\mu_\gamma(A(u_1,\ldots,u_k))$ is some
non-zero value, whereas $\mu_\delta(A(u_1,\ldots,u_k))$ is 0 since it is
not possible to have more than $1/\delta$ non-overlapping waking
periods in the duty-cycled network. Therefore $\ex_{\mu_\gamma}(-
  I_{A(u_1,\ldots,u_k)}) < 0 = \ex_{\mu_\delta}(- I_{A(u_1,\ldots,u_k)})$
    for this value of $k$, which implies that $\mu_\delta \not \preceq
    \mu_\gamma$.

The argument presented above is general for any duty-cycling scheme
with parameter $\delta$. To the prove that there is no stochastic
domination in the other direction we consider the contiguous
duty-cycling scheme \dcc\ defined in Section~\ref{sec:model} where
each node $u$ chooses a value $i_u$ uniformly at random from
$\{0,\ldots, L-1\}$ and is awake at time slots $i_u, i_u + 1 \bmod L,
\ldots i_u + d - 1\bmod L$.  With this scheme operating, consider the
case that there are three point $u,v$ and $w$ that all lie within
distance $r$ of each other in the point process. If we denote the
event that any set of edges $e_1, \ldots, e_k$ is in $E$ by
$B(e_1,\ldots,e_k)$, we have that $\mu_\gamma(B((u,v),(v,w),(u,w)))) =
\gamma^3$. Now, note that that if $u$ is connected to $v$, then for
$w$ to be connected to both $u$ and $v$, $w$'s waking cycle must
overlap with the slots of the duty cycle which are common to both $u$
and $v$. We fix the position of $u$ and condition on the event that
$u$ and $v$ have exactly $i$ slots in common (which happens with
probability $2/L$ for all $1\leq i\leq d-1$ and with probability $1/L$
for $i=d$), we get
\[\mu_\delta(B((u,v),(u,w),(v,w))) = 
\sum_{i=1}^{d} \frac{d+i-1}{L} \cdot \frac{2}{L} - \frac{2d-1}{L^2} =
\frac{3d^2 - 3d +1}{L^2}.\]
It is easy to verify that there are settings $d$ and $L$ for which
this value is actually strictly less than \[\mu_\gamma(B((u,v),(u,w),(v,w))) = \gamma^3
= ((2d-1)/L)^3.\] Hence, for those settings $\ex_{\mu_\delta}(
I_{B((u,v),(v,w),(u,w))}) <
\ex_{\mu_\gamma}(I_{B((u,v),(v,w),(u,w))})$, and since
$I_{B((u,v),(v,w),(u,w))}$ is an increasing function:  $\mu_\gamma \not \preceq
    \mu_\delta$.

Hence we have found that the two models are not related through
stochastic domination, and this motivates us to define a new model
that can describe the duty-cycled setting better and in which
connectivity results have to be proved anew. We define a general model
of this nature, we call it the {\em vertex-based random connection
  model}, in Section~\ref{sec:strong:vertex-based}.

\subsection{A vertex-based random connection model}
\label{sec:strong:vertex-based}

We now formally define our vertex-based random geometric graph
model. This model has four parameters. There are two finite positive
real numbers $\lambda, r$.  The third parameter is a random variable
$Z$ defined on some probability space $(\Omega,\CF,\pr)$, that is a
function of the form $Z : \Omega \rightarrow Q$ where $Q$
is some domain. The fourth parameter is a function $f: Q \times
Q \rightarrow \{0,1\}$. The vertex set $V$ is a Poisson point
process in $\RR^2$ with density $\lambda$ with an additional point at
the origin. Now we define the edge set $E$. With each $u \in V$ we
associate a random variable $Z_u$ which is a copy of $Z$. All the
random variables in the collection $\{Z_u : u \in V\}$ are independent
of each other. Moreover.
\begin{equation}
\label{eq:def}
 g(u,v) =
  \begin{cases}
   \pr'(f(Z_u,Z_v) = 1)  & \text{if } d(u,v) \leq r, \\
   0      & \text{otherwise}, 
  \end{cases}
\end{equation}
where $\pr'$ is the product measure defined on the product space of
the two random variables $Z_u$ and $Z_v$.  In other words, the edge
$(u,v)$ exists if $f(Z_u,Z_v)$ is 1, but only if $d(u,v) \leq
r$. Clearly, for this model to be useful, there should be non-zero
probability of an edge being formed between two points that are within
distance $r$ of each other. Also note that if $\pr'(f(Z_u,Z_v)=1) =1$
then the model reduces to the \rgg.

\subsubsection{Some restrictions on the vertex-based random connection
  model}
\label{sec:strong:vertex-based:restrictions}

Since this model is defined in a fairly general setting, we now define
some restrictions which make it more useful for us.

\paragraph{\sf Non-triviality} 
In this model the edge $(u,v)$ exists if $f(Z_u,Z_v)$ is 1, but only
if $d(u,v) \leq r$. A basic condition we need on the function $f$ and
the probability space on which $Z$ is defined is that the probability
of making a connection between two points should be non-zero. We call
this the ``non-triviality condition''.

\MyQuote{
\label{cond:non-triviality}
{\em Given two independent copies $Z_1$ and $Z_2$
of $Z$,  $0 < \pr'(f(Z_1,Z_2) = 1) < 1$.}
}
\paragraph{Finite reachability}
\label{sec:strong:vertex-based:restrictions:reachability}
The model as defined so far admits a serious anomaly. Consider the
case where $Z$ takes values from $\{0,1\}$ with equal probability and
$f$ is the equality function i.e. $f(x,y) = 1$ if $x=y$ and 0
otherwise. Clearly the non-triviality condition is satisfied. However
with this definition of $f$, the random graph that will be formed will
have two distinct classes of points: $\{u \in V: Z_u = 0\}$ and $\{u
\in V: Z_u = 1\}$. In this case, it will be like we have two random
graph models superposed on the same space (with appropriately thinned
Poisson processes) with no possibility of any edge between these
points. If both these processes may be supercritical
independently, there are two infinite components. To mitigate this
problem and to ensure that the uniqueness of the infinite component
that is seen in the random connection model is seen here as well, we
introduce a condition on $f$ and $Z$ that we call {\em finite
  reachability}. 

Let us first consider the case where
$\Lambda$ is a finite or countable set. We denote the support of
$\pr(\cdot)$ by $\supp(\Lambda,\pr)$ i.e. the set $\{ x \in \Lambda:
\pr(x) > 0\}$. Now, given $x,y \in \supp(\Lambda,\pr)$ we say that $x$
and $y$ are {\em 0-reachable from each other} if $f(x,y) = 1$, and are
$k$-reachable from each other if there exists $w \in \supp(\Lambda,P)$
such that $x,w$ are 0-reachable from each other and $w,y$ are $k-1$
reachable from each other. The finite reachability condition on $f$
and $Z$ is that all $x,y \in \supp(\Lambda,\pr)$ are $k$-reachable
from each other for some finite $k$ i.e. 

\MyQuote{
\label{cond:finite-reachability}
\em  $x,y \in \supp(\Lambda,\pr)$ are said to be $k$-reachable from each
  other, there is a sequence $w_0, w_1, \ldots, w_k $ such that $w_0$
  is $x$ and $w_k$ is $y$ and
  $w_i \in \supp(\Lambda,\pr)$, $0 < i < k$ and $f(w_i, w_{i+1}) = 1$,
  $0 \leq i < k$.
}

\paragraph{Connection Diversity} Non-triviality and the assumption
that $\{Z_u: u \in V\}$ is an independent collection implies a
property we call the ``connection diversity condition.'' We are
stating it separately for convenience.  Consider $k+1$ copies of $Z$,
$Z_0,Z_1,\ldots,Z_k$, all independent of each other. There is a
constant $c \in (0,1]$, depending only on $Z$ and $f$, such that
\begin{equation}
\label{eq:conn-div}
\pr\left(f(Z_0,Z_1) = 0 \cap \bigcup_{i=2}^{k} f(Z_0,Z_i) = 1\right) >c, \forall k
\geq 2,
\end{equation}
i.e. given a copy of $Z$ called $Z_0$ and $k$ independent copies of
$Z$, $Z_1,\ldots,Z_k$, there is non-zero probability that even if
$f(Z_0,Z_1)$ is 0 there is at least one $Z_i$ in the remaining
$Z_2,\ldots,Z_k$ such that $f(Z_0,Z_i)$ is 1.

In the following we will assume that whenever we talk of the
vertex-based random connection model, we are talking about a model
where $f$ satisfies the non-triviality condition, and hence the
connection diversity condition as well. 

\subsubsection{Basic properties of the vertex-based random connection model}

We now state some fundamental properties of this model. Since this
model is a generalization of the random connection model defined
in~\cite{meester:1996}, it is natural to ask whether it shares
some properties with that model. In fact, under the restrictions
described in Section~\ref{sec:strong:vertex-based:restrictions}, the
vertex-based random connection model has a non-trivial critical
density and has at most one infinite component. We state these
properties formally now. 

We will denote by $W(x)$ the connected component containing the point
$x \in V$. For the special case $W(0)$ i.e. the connected component
containing the origin, we will simply write $W$. We will use the
notation $\theta_g(x,\lambda)$ to denote the probability that the
point $x \in V$ is part of an infinite cluster. We will drop the $x$
when $x$ is the origin, writing simply $\theta_g(\lambda)$.

\begin{proposition} (Critical phenomena and non-triviality of critical density)
\label{prp:criticality}
For the vertex-based random connection model with parameters $\lambda$
and $r < \infty$ and a connection function $g$ based on a function
$f$ and a random variable $Z$ that satisfy the non-triviality
condition~(\ref{cond:non-triviality}), there is a critical value
$\lambda_c$ such that $\theta_g(\lambda) > 0$ whenever $\lambda >
\lambda_c$ and $\theta_g(\lambda) = 0$ whenever $\lambda <
\lambda_c$. Moreover $0 < \lambda_c < \infty$.
\end{proposition}

The proof of the first part of this proposition follows by observing
that a standard coupling argument (see e.g. Meester and
Roy~\cite{meester:1996}, ) implies $\theta_g(\lambda) \leq
\theta_g(\lambda')$ whenever $\lambda \leq \lambda'$, and then
applying Kolmogorov's 0-1 law.  The second part, the non-triviality of
the critical probability, involves a proof, but it is a standard proof
not very different from that presented in~\cite{penrose-aap:1991}, so
we omit it here.

As discussed in
Section~\ref{sec:strong:vertex-based:restrictions:reachability}, it is
easy to see that for a general choice of $Z$ and $f$ the vertex-based
random connection model could contain multiple infinite-sized
connected components (clusters). However, the finite reachability
restriction disallows this and forces the model to behave in a
reasonable manner similar to the random connection model, thereby
making it of some use to us.

\begin{proposition} (Uniqueness of the infinite component)
\label{prp:uniqueness}
The vertex-based random connection model with parameters $\lambda$
and $r < \infty$ and a connection function $g$ based on a function
$f$ and a random variable $Z$ that satisfy the non-triviality
condition~(\ref{cond:non-triviality}) and the finite reachability
condition~(\ref{cond:reachability}) contains at most one infinite
connected component.
\end{proposition}

The proof of Proposition~\ref{prp:uniqueness} proceeds by first noting
that due to the ergodicity of the process the number of infinite
components is almost surely constant. Then we proceed by
contradiction, assuming that there are greater than 2 infinite
components and showing how multiple components can be connected with
positive probability. This proof is long and involved and is almost
exactly similar to the proof of the same result for the Boolean model
(Proposition 3.3 of~\cite{meester:1996}) so we do not repeat it here,
only noting that a critical part of the proof involves showing that
for a large enough but finite sized box, multiple infinite components
enter it and so can be connected within it with finite
probability. Connecting multiple infinite components within a finite
sized box in the vertex-based random connection model requires the
finite reachability condition, which establishes that this condition
is not just necessary but also sufficient to establish the uniqueness
of the infinite component (if it exists).

\subsubsection{A high density result for the vertex-based random
  connection model}

We now come to our main result. Define the quantity $q_k(\lambda) =
\pr_\lambda(|W| = k), k \geq 1$ where $W$ is the connected component
containing the origin i.e. $q_k(\lambda)$ is the probability that the
component containing the origin has size $k$.  Our key contribution is
that we can show that the following result proved by
Penrose~\cite{penrose-aap:1991} for the high-density setting of the
random connection model also holds for the vertex-based random
connection model:
\begin{lemma}
\label{lem:penrose-modified}
\[ \lim_{\lambda \rightarrow \infty} \frac{\sum_{k = 1}^{\infty}
  q_k(\lambda)}{q_1(\lambda)} = 1\]
\end{lemma}
Since $\theta_g(\lambda) = 1 - \sum_{k = 1}^{\infty} q_k(\lambda)$,
the implication of this theorem is that as $\lambda \rightarrow
\infty$, the origin is either isolated or part of the infinite
component. Any other situation occurs with probability 0. The proof of
this Theorem is involved and technical so we move it to the Appendix.

\subsection{The strong connectivity result}
\label{sec:strong:main-theorem}

Denote by $\vbrgg(n,r,\gamma)$, the vertex-based random connection
model graph with vertex set consisting of $n$ points uniformly
distributed in the unit circle centred at the origin, with radius
bound $r$, and a connectivity function $g$ as defined
in~(\ref{eq:def}) using a function $f$ and and random variable $Z$
such that for any $Z_1$ and $Z_2$ that are independent copies of $Z$,
$\gamma = \pr(f(Z_1,Z_2) = 1)$. Now, we are ready to state our optimal
connectivity result.

\begin{theorem}
\label{thm:vb-gamma}
$\pr(\vbrgg(n,r,\gamma) \mbox{ is connected})\rightarrow 1$ as $n
\rightarrow \infty$ if and only if
\begin{equation}
\label{eq:vb-radius}
 \pi r(n)^2 \gamma = (\log n + c(n))/n,
\end{equation}
where $c = \lim_{n \rightarrow \infty} c(n) = \infty$ as $n \rightarrow \infty$.
\end{theorem}

\ifrevisiontwo
{\bf 
\fi
\paragraph{\sf Discussion on $c(n)$} Before we get to the proof, let
us briefly discuss the quantity $c(n)$ that appears
in~(\ref{eq:vb-radius}) and has appeared before in~(\ref{eq:das-radius})
and~(\ref{eq:unif-das-radius}). Since the probability that a node is connected to a
neighbor is $\gamma$ in $\vbrgg(n,r,\gamma)$ and since the density of
the points in the unit disk is $n$, we can reorganize the terms
of~(\ref{eq:vb-radius}) to see that the average degree of a vertex in
this graph is $\log n + c(n)$. Hence, what this result is saying is
that this graph is connected if and only if the average degree is
greater than $\log n$ by a quantity that is asymptotically significant
i.e. that tends to infinity when $n \rightarrow \infty$. Another way
of writing this is that average degree of $(1 + \omega(1/(\log
n))) \log n$ is necessary and sufficient for connectivity.
\ifrevisiontwo
}
\fi

\begin{proof} The theorem says that the condition~(\ref{eq:vb-radius})
  is both necessary and sufficient to establish connectivity. We begin
  with the sufficient condition. 

The proof of the sufficient condition uses
Lemma~\ref{lem:penrose-modified} which has been proved for the
vertex-based random connection model in the infinite plane. We first
need to show that this result can be used in the finite disc, a step
in the proof that was omitted by Gupta and
Kumar~\cite{gupta-chapter:1999}. We begin by fixing our
notation. Given $n, r \in \RR_+$ and $\gamma \in [0,1]$ we define
three random geometric graph models that are clearly related to each
other. $B(r)$ denotes the disk of radius $r$ centred at the origin and
$\PP(\lambda)$ denotes a Poisson Point Process of density $\lambda$ in
$\RR$. We define three graphs with the following vertex sets: (1)
$\vbrgg(n,r,\gamma)$: $n$ points distributed uniformly at random in
$B(1)$ and one point at the origin.  (2) $\vbrgg(n,r,\gamma,\ell)$:
$n$ points distributed uniformly at random in $B(\ell)$ and one point
at the origin.  (3) $\pvbrgg(n,r,\gamma, \infty)$: $\PP(n)$ and one
point at the origin.  The edge set of these random graphs is defined
as specified above for $\vbrgg(n,r,\gamma)$ i.e. using a connection
function $g$ that uses a function $f$ and a random variable $Z$ whose
independent copies are associated with each vertex.

We will use the notation $W^{\lambda}_\ell(x)$ to denote the connected
component containing the point $x$, and use only the notation
$W^{\lambda}_\ell$ when $x$ is the origin, where $\lambda$ will be the
density of the model ($n$ for all three models in the description
above) and $\ell$ will be the radius of the disc around the origin in
which the points of the model are placed.
\begin{lemma}
\label{lem:isolated}
For $\vbrgg(n,r(n),\gamma)$ if we denote by $A_n$ the event that there
exists a sequence $\{s_n\}_{n\geq 1}$ such that $|W^n_1| \geq s_n$ and $1 \leq s_n \leq n,
\forall n\geq 1$ and $s_n \rightarrow \infty$ as $n \rightarrow \infty$, then
\[\lim_{n\rightarrow \infty} (1 - \pr(A_n)) = \lim_{n \rightarrow
  \infty} \pr(|W^n_1| = 1),\]
as long as 
\begin{equation}
\label{eq:condition-radius}
\exists c : r(n)\cdot n^{1/3} \leq c,
\forall n \geq 1
\end{equation}
\end{lemma}
\begin{proof} By scaling we couple
$\vbrgg(n,r,\gamma)$ to a random graph model on a disc of larger radius
such that the probability that the component containing
the origin is of any particular size remains
exactly the same in the coupled model. This coupled model is
$\vbrgg(n^{1/3},r\cdot n^{1/3}, \gamma, n^{1/3})$ i.e. the random graph
model with a lower density than $\vbrgg(n,r,\gamma)$ by a factor of
$n^{2/3}$ and a radius longer by a factor of $n^{1/3}$ on a disc of
radius $n^{1/3}$. This basically involved expanding the unit disc with
density to a
disc of radius $n^{1/3}$. All the edges and non-edges are preserved
since the connection radius increases in exactly the same proportion
as the distances between points. The increase in distances brings the
density down by a factor of the square of the increase in distances
i.e. by $n^{2/3}$. Hence it is easy to see that:
\begin{equation*}
\pr(|W^n_1| = k) = \pr(|W^{n^{1/3}}_{n^{1/3}}| = k), \forall n,k \geq 1.
\end{equation*}
 In particular, taking the limit as $n \rightarrow \infty$ on both
 sides for $k = 1$,
\begin{equation}
\label{eq:lim-1}
\lim_{n \rightarrow \infty} \pr(|W^n_1| = 1) = \lim_{n \rightarrow \infty}\pr(|W^{n^{1/3}}_{n^{1/3}}| = 1).
\end{equation}
But $\vbrgg(n^{1/3},r\cdot n^{1/3}, \gamma, n^{1/3})$ has the
property that as $n \rightarrow \infty$, the density of the process
tends to infinity, and the disc it covers expands to the entire
plane. In other words it converges to
$\vbrgg(\lambda,r(\lambda),\gamma,\infty)$ in the limit $\lambda
\rightarrow \infty$. So (\ref{eq:lim-1}) implies that 
\begin{equation}
\label{eq:lim-1-infty}
\pr(|W^n_1| = 1) = \lim_{\lambda \rightarrow \infty} \pr(|W^{\lambda}_{\infty}| = 1).
\end{equation}
The condition on the connection function ($r(n)\cdot n^{1/3} < c$)
implies that the connection function has bounded support (i.e. beyond
a radius that is at most $c$ the probability of forming an edge is
0). Hence we can use Lemma~\ref{lem:penrose-modified}.
This lemma along with (\ref{eq:lim-1-infty}) implies that 
\begin{equation}
\label{eq:lim-meester-infty}
\lim_{n \rightarrow \infty} \pr(|W^n_1|=1) = \lim_{\lambda \rightarrow
  \infty} (1 - \pr(|W^\lambda_\infty| = \infty)).
\end{equation}
Noting that as $n \rightarrow \infty$, the event $A_n$ tends to the
event $\{|W^\lambda_\infty| = \infty\}$ as $\lambda \rightarrow
\infty$, the lemma follows from (\ref{eq:lim-meester-infty}).
\end{proof}

The following lemma gives upper bound on the probability of an isolated node
existing. 
\begin{lemma}
\label{lem:vb-poisson}
For $\vbrgg_P(n,r,\gamma)$, if (\ref{eq:vb-radius}) holds then
\[\lim \sup_{n \rightarrow \infty} \pr(\exists x \in V : |W^n_1(x)| = 1)
\leq e^{-c},\]
where $c = \lim_{n\rightarrow \infty} c(n)$.
\end{lemma}
\begin{proof}
Let us assume that the Poisson process places $j$ points, $x_1, \ldots,x_j$ in the unit
disc. For any given point out of these $j$, the probability that it is
isolated (i.e. its component has size 1) is computed by observing that
this happens only when the points lying in disc of
radius $r(n)$ around it are not
connected to it. To compute this probability we observe that if $k
\leq j-1$ of these points lie inside this disc then they must all be
not connected which happens with probability $(1 - \gamma)^k$ for a
fixed set of $k$ points. The number of ways of choosing $k$ points out
of $j-1$ is ${j-1 \choose k}$ and for a fixed set of $k$ points out of
$j-1$, the probability that they lie within the disc of radius $r(n)$
around the point of interest while the other $j-1-k$ do not is given
by $(\pi r^2(n))^k (1 - \pi
r^2(n))^{j-1-k}$.  Hence we have that 
\begin{equation}
\label{eq:vb-poisson}
\pr(x_1 \mbox{ is isolated}) \leq
\sum_{k=0}^{j-1} {j-1 \choose k} (\pi r^2(n))^k (1 - \pi
r^2(n))^{j-1-k}\cdot (1 - \gamma)^{k} =(1 - \gamma\pi
r^2(n))^{j-1},
\end{equation}
and so 
\[ \pr(\exists i: 1\leq i\leq j, x_i \mbox{ is isolated}) \leq j \cdot
(1 - \gamma\pi r^2(n))^{j-1}.\] 
From here we get the proof of the lemma just as Gupta and Kumar do by
conditioning on the event that the Poisson point process places $j$
points in the unit disc. 
\end{proof}
We note that to be precise we must observe that the disc of radius
$r(n)$ centred on an arbitrary point in the unit disc may not lie
entirely within the unit disc. It is easy to see that the this problem
occurs in a ring of width $r(n)$ at the boundary of the unit
disc. This complication disappears in the limit since $r(n)
\rightarrow 0$ as $n \rightarrow \infty$. Gupta and Kumar have handled
this complication in precise and tedious detail in the appendix
of~\cite{gupta-chapter:1999} and so we don't repeat that here.

Finally we show that the bound on $\pvbrgg(n,r(n),\gamma)$ containing
an isolated vertex translates into a bound on $\vbrgg(n,r(n),\gamma)$
being disconnected if (\ref{eq:vb-radius}) holds.  Since the radius
bound of (\ref{eq:vb-radius}) satisfies the condition
(\ref{eq:condition-radius}), we can apply Lemma~\ref{lem:isolated} to
claim that for any $\epsilon > 0$ there is a sufficiently large $n$
such that $\pr(\pvbrgg(n,r(n),\gamma)\mbox{ is disconnected})$ is
upper bounded by $(1 + \epsilon) \cdot \pr(\exists x \in V :|W^n_1(x)|
= 1)$.  We follow Gupta and Kumar's calculations, noting only that in
our case $\pr(\mbox{node }k\mbox{ is isolated in
}\vbrgg(k,r(n),\gamma))$ is upper bounded by $(1 - \gamma \pi
r^2(n))^{k-1}$, %as argued in the proof of
%Lemma~\ref{lem:vb-poisson}. 
Hence, 
\begin{align*}
\pr  (&\vbrgg(n,r(n),\gamma)\mbox{ is disconnected})  \\
& \leq 2(1 - 4 \epsilon)\left(\pr(\exists x \in V :
|W^n_1(x)| = 1) + \frac{e^{-\gamma \pi r^2(n)}}{\gamma \pi
  r^2(n)}\right)
\end{align*}
Under condition (\ref{eq:vb-radius}) and using Lemma~\ref{lem:vb-poisson}
we get 
\begin{align*}
\lim_{n \rightarrow \infty} \pr(\vbrgg(n,r(n),\gamma)&\mbox{ is
  disconnected}) \\
& \leq 2(1 - 4
\epsilon)\left( e^{-c}\right).
\end{align*}
Since $\epsilon$ can be taken to be arbitrarily small, the sufficient
part of the theorem follows since  $c(n) \rightarrow \infty$ as $n
\rightarrow \infty$ implies that $e^{-c} = 0$ .

We now move on to the necessary condition, to establish which we will
show that when $c(n)$ is a positive constant, the probability of an
isolated node existing in $\vbrgg(n,r(n), \gamma)$ is non-zero as $n
\rightarrow \infty$. This implies that $\vbrgg(n,r(n),\gamma)$ is
disconnected with positive probability when $c$ is a positive
constant. Since it can easily be shown using a coupling argument that
the probability of disconnection increases as $c$ decreases, this is
sufficient to show that $\vbrgg(n,r(n),\gamma)$ is disconnected with
non-zero probability for all values $c$ which are less than $+\infty$.

Specifically we will prove the following proposition:
\begin{proposition}
\label{prp:necessary-poisson}
If 
\begin{equation}
\label{eq:yi}
\pi r(n)^2 \gamma = (\log n + c)/n,
\end{equation}
for some constant $c \geq 0$
then the distribution of the number isolated nodes in
$\vbrgg(n,r(n),\gamma)$ is Poisson with mean $e^{-c}$.
\end{proposition}
\begin{proof}%[Proof of Proposition~\ref{prp:necessary-poisson}]
We note that the basic idea of the proof in the \rgg\ setting is
already present in Gupta and Kumar's
paper~\cite{gupta-chapter:1999}. Yi et. al. have put this proof on a
more rigorous basis and extended it to more general scenarios such as
\rgg s with nodes being ``active'' independently with probability $p$,
what they call Bernoulli nodes~\cite{yi-tcomm:2006}, and also to the
case where nodes and edges are both active or not
independently~\cite{yi-dmaa:2010}. Mao and Anderson have obtained the
same result using the Chen-Stein method~\cite{mao-tit:2013} but we
will follow the proof of~\cite{yi-tcomm:2006} since it is much more
direct and geometrical and hence easier to adapt for our purposes.

The proof of Yi et. al.~\cite{yi-tcomm:2006} is based on a
probabilistic version of Brun's sieve theorem (see
e.g.~\cite{alon:2000}) which is as follows: 
\begin{lemma}
\label{lem:brun}
Given a sequence $B_1, B_2, \ldots, B_n$ of events, define $Y_n$ to be
the (random) number of $B_i$ that hold. Now, if for any set $\{i_1,
\ldots i_k\}$ it is true that
\begin{equation}
\label{cond:brun1}
\pr\left(\cap_{j=1}^k B_j \right) = \pr\left(\cap_{j=1}^k B_{i_j}
\right),
\end{equation}
and there is a constant $\mu$ such that for any fixed $k$
\begin{equation*}
\label{cond:brun2}
n^k \pr\left(\cap_{j=1}^k B_j \right) \rightarrow \mu^k \mbox{ as }
n \rightarrow \infty,
\end{equation*}
then the sequence $\{Y\}_n$ converges in distribution to a Poisson
random variable with mean $\mu$.
\end{lemma}

As in the proof of Yi et. al.~\cite{yi-tcomm:2006}, we define the
event $B_i$ as the event that the $i$-th vertex of
$\vbrgg(n,r(n),\gamma)$, denoted $X_i$, is isolated. In this setting
it is clear that the condition (\ref{cond:brun1} holds because for any
$k$ the joint probability of any $k$ vertices being isolated is the
same as that of any other $k$ vertices due to the fact that all the
points are placed uniformly at random in $B(1)$ independent of each
other. Hence to prove Proposition~\ref{prp:necessary-poisson} using
Lemma~\ref{lem:brun}, we only need to show that
\begin{equation}
\label{cond:sieve}
n^k \pr\left(\cap_{j=1}^k B_j \right) \rightarrow e^{-ck} \mbox{ as }
n \rightarrow \infty, \mbox{ for all } k \geq 1.
\end{equation}

We reproduce some necessary notation from Yi
et. al.~\cite{yi-tcomm:2006}.  Given a finite set of points $x_1,
\ldots, x_k$ from $B(1)$ (i.e. the unit disk centred at the origin),
$G_r(x_1,\ldots,x_k)$ is the graph formed by placing an edge between
each pair of points that is at a distance of at most $r$. $C_{k,m}$ is
defined as the set of $k$-tuples $(x_1,\ldots,x_k) \in B(1)^k$ such
that $G_{2r}(x_1,\ldots,x_k)$ has exactly $m$ connected components. We
note that the set $C_{k,k}$ consists of those tuples of $k$ points
which have the property that a disk of radius $r$ around each of the
points contains none of the other points of the tuple. For a set of
points $S \subseteq B(1)$, Yi et. al. denote by $\nu_r(S)$ the area of
the union of the $r$-radius disks centred at the points of $S$
intersected with $B(1)$ i.e. the Lebesgue measure of the set of points
from $B(1)$ that are at most distance $r$ from one of the points of
$S$.

When $r$ satisfies~(\ref{eq:yi}), Yi et. al.~\cite{yi-tcomm:2006}
prove the following geometric properties:
\begin{equation}
\label{eq:yi-1}
n \int_{B(1)} (1 - \gamma \nu_r(x))^{n-1} dx  \rightarrow  e^{-c}
\mbox{ as } n \rightarrow \infty.
\end{equation}
\begin{equation}
\label{eq:yi-2}
n^k \int_{C_{k,m}} (1 - \gamma \nu_r(x_1,\ldots,x_k))^{n-k}
\prod_{i=1}^k dx_i  \rightarrow  0 \mbox{ as } n \rightarrow \infty,
\mbox{ for } k \geq 2, 1 \leq m < k.
\end{equation}
\begin{equation}
\label{eq:yi-3}
n^k \int_{C_{k,k}} (1 - \gamma \nu_r(x_1,\ldots,x_k))^{n-k}
\prod_{i=1}^k dx_i  \rightarrow  e^{-kc} \mbox{ as } n \rightarrow
\infty, \mbox{ for } k \geq 2.
\end{equation}
Of these three properties, (\ref{eq:yi-1}) was previously demonstrated
in~\cite{gupta-chapter:1999}. To be able to use these properties to
prove that (\ref{cond:sieve}) holds in our case  we
need to show certain properties of the probabilities of the events
$B_i$. We state these as a claim.

\begin{claim}
\label{clm:sieve-vbrgg}
\begin{enumerate}
\item For any  $x \in B(1)$, 
\begin{equation}
\label{eq:sieve-prob-1}
 \pr(B_1 \mid X_1 = x) = (1 - \gamma \nu_r(x))^{n-1},
\end{equation}
\item For any  $k \geq 2$ and $(x_1, \ldots, x_k) \in B(1)^k$,
\begin{equation}
\label{eq:sieve-prob-2}
\pr\left(\bigcap_{i=1}^k B_i \mid X_i = x_i, 1 \leq i \leq k\right) \leq (1 -
\gamma \nu(x_1, \ldots, x_k))^{n-k}.
\end{equation}
\item The equality in (\ref{eq:sieve-prob-2}) is achieved for
$(x_1,\ldots,x_k) \in C_{k,k}$.
\end{enumerate}
\end{claim}
%\proof[Proof of Claim~\ref{clm:sieve-vbrgg}]
\begin{proof}
We have already demonstrated that (\ref{eq:sieve-prob-1}) holds for
$\vbrgg(n,r(n),\gamma)$ in the proof of Lemma~\ref{lem:vb-poisson}
(see (\ref{eq:vb-poisson})). So we move to (\ref{eq:sieve-prob-2}).

Let us abuse  notation slightly and use $\nu_r(x_1,\ldots,x_k)$ to
refer to the region that lies within distance $r$ of any point of
$x_1, \ldots, x_k$ as well as the area of the region.
Clearly if the event $\{\bigcap_{i=1}^k B_i \mid X_i = x_i, 1 \leq i
\leq k\}$ occurs then it must be the case that even if any of the $n -
k$ nodes $X_{k+1}, \ldots, X_n$ lies in $\nu_r(x_1, \ldots, x_k)$  it must not be connected to any point in that
subset. If we assume that a point $X_l$ lies within the
$\nu_r(x_1,\ldots,x_k)$, it must lie with distance $r$ of at least one
of the points of $x_1, \ldots, x_k$ and at most $k$ of them. Since
each of these $k$ points choose $Z_{x_i}, 1 \leq i \leq k$
independently, the probability of $X_l$ being not connected to each
one of them is at most $(1-\gamma)$ and at least $(1 - \gamma)^k$. The
upper bound is relevant to us here. If we say that some $j$ of the
points $X_{k+1}, \ldots, X_n$ lie within $\nu_r(x_1,\ldots,x_k)$,
since each of them choose their $Z$ independently, the probability
that all $j$ of them are not connected to any of $x_1, \ldots, x_k$ is
at most $(1-\gamma)^j$. Hence we get:
\begin{equation}
\label{eq:sieve-prob-helper}
\pr\left(\bigcap_{i=1}^k B_i \mid X_i = x_i, 1 \leq i \leq k\right)
\leq \sum_{j = 0}^{n-k} {n - k \choose j} (1 - \nu_r(x_1, \ldots,
x_k))^{n-k-j} \nu_r(x_1,\ldots,x_k)^j (1 - \gamma)^j.
\end{equation}
The RHS simplifies to give us (\ref{eq:sieve-prob-2}).

To see that the equality holds for the case where $(x_1,\ldots,x_k)
\in C_{k,k}$ we note that for any $x_i, 1 \leq i \leq k$, it cannot be
connected to any of the $x_j,  1 \leq j \leq k, i\ne j$ since by the
definition of $C_{k,k}$ the distance between any pair of these points
is at least $2r$ which is more than $r$. Also, since, by the
definition of $C_{k,k}$ these points are spaced $2r$ apart, the disc
of radius $r$ around each of these points overlaps with none of the
other discs. Hence, any point of $X_{k+1}, \ldots, X_n$ that lands in
$\nu_r(x_1,\ldots,x_k)$ is within distance $r$ of exactly one point of
$x_1, \ldots, x_k$ and hence has probability exactly $(1-\gamma)$ of
not being connected to it. From this argument we can see that the two
upper bound approximations we used to calculate the RHS of
(\ref{eq:sieve-prob-helper}) actually hold exactly whenever
$(x_1,\ldots,x_k) \in C_{k,k}$.
\end{proof}

Now we see that (\ref{cond:sieve}) is satisfied for the case when
$k=1$ by combining (\ref{eq:yi-1}) and (\ref{eq:sieve-prob-1}). For
the case when $k \geq 2$, (\ref{eq:yi-2}) and (\ref{eq:sieve-prob-2})
imply that when $(x_1, \ldots,x_k) \notin C_{k,k}$ $n^k
\pr\left(\cap_{j=1}^k B_j \right) \rightarrow 0$ as $n \rightarrow
\infty$, but when $(x_1, \ldots,x_k) \in C_{k,k}$, (\ref{eq:yi-3}) and
part 3 of Claim~\ref{clm:sieve-vbrgg} give us that $\pr\left(\cap_{j=1}^k B_j
\right) \rightarrow e^{-kc}$. This concludes the proof of
Proposition~\ref{prp:necessary-poisson}. \sq
\end{proof}

As argued above Proposition~\ref{prp:necessary-poisson} implies that
if $c(n)$ does not grow to $\infty$ as $n \rightarrow 0$, then there
is a non-zero probability of $\vbrgg(n,r(n),\gamma)$ being
disconnected as $n \rightarrow \infty$ i.e. the condition that $c(n)
\rightarrow \infty$ as $n \rightarrow \infty$ is necessary as well as
sufficient. \sq
\end{proof}

\section{Simulation results for the Contiguous and Random Duty-Cycle Schemes}
\label{sec:simulation}

\ifrevision
{\bf 
\fi
In this section we present the results of an extensive simulation
study. The main aim of this study is to support the theoretical
results we have presented so far. Hence we present simulations that
show that the weak radius presented in Theorem~\ref{thm:conn-weak} is
indeed sufficient. We also show through a series of experiments that
the strong radius presented in Theorem~\ref{thm:vb-gamma} is optimal
in the sense that it is both necessary and sufficient. 

In order to demonstrate these results through simulation we need
concrete duty-cycling schemes. For this purpose we use the {\em
  contiguous} and the {\em random selection} duty-cycle schemes
(\dcc\ and \dcr,\ resp.), introduced at the end of
Section~\ref{sec:model}.  In the former, the sensor selects a slot in
the cycle period and stays awake for $d$ consecutive time slots.  In
the latter, the sensor selects at random $d$ awake slots during the
cycle period. And while the prime focus of the simulations is to
validate our theoretical results we also study the sensitivity of
\dcc\ and \dcr\ to the duty-cycling parameters $\delta$ and $L$, as
well as to the number of nodes $n$.
\ifrevision 
{\em  [revision id \#8]}} 
\fi

\subsection{The two duty-cycle schemes}
\label{sec:simulation:discussion}

Before describing the experiments, observe that the contiguous duty
cycle scheme satisfies the reachability condition given
in~(\ref{cond:reachability}), which is an essential condition for the
optimal connectivity result (i.e. Theorem~\ref{thm:vb-gamma}) to apply.
To see this note that given two nodes $u$ and $v$ whose duty cycles
begin at $i_u$ and $i_v$, it is possible to find a chain of
overlapping awake periods beginning at $i_u+m(d-1)$, $1 \leq m <
\lceil (i_v-i_u)/(d-1) \rceil$. And since the probability of picking
an awake cycle beginning at $i_u + m(d-1)$ for the relevant value of
$m$ is greater than 0 (in fact $1/L$), the reachability condition is
easily satisfied.  

In regard to the random selection scheme,
whenever $d \geq 2$ the reachability condition can be easily achieved
for any two non-overlapping awake periods. In fact, we can always pick
a third awake period that overlaps with both with non-zero
probability. Hence, Theorem~\ref{thm:vb-gamma} can be applied here.

\ifrevision
{\bf} 
\fi
The contiguous scheme is very natural and has several advantages.
First it is easy and cheap to hard code in each sensor. Namely, the
scheme for each sensor $u$ is uniquely defined by the three constants
$d$, $L$ and $i_u$, consuming a negligibly amount of memory. Only
$\log L$ random bits are needed to generate $i_u$. Moreover, since the
$d$-awake periods are consecutive, the scheme can tolerate 
clock drift. To clarify this let us take a concrete example. Say $d =
5$ and $L = 100$ and there are nodes $u$ and $v$ such that $i_u = 10$
and $i_v = 12$ i.e. they have an overlap of 3 slots. Even if the
clock of $u$ drifts backward by one whole slot and the clock of $v$
drifts forward by 1 whole slot, they still overlap for 1 slot, which
is all they need to communicate. Hence the contiguous scheme is robust
to problems in synchronization. %Similarly, slight deviations in slot
%length can also be tolerated as long as the two communicating sensors
%have sufficient waking time in common to be able to make a
%connection.   
%CRISTINA
%CRISTINA
Finally, the energy spent in commuting between the sleep
and the active mode is minimal since only two state-transitions occur
in each cycle.

The random selection scheme is less thrifty than the
contiguous one. Firstly, $\theta(d \log L)$ bits are needed to
generate a schedule for each sensor.  Secondly, each node needs to
memorize $d$ constants in addition to the values of $d$ and $L$ and it
incurs in more than $2$ mode-transitions in each cycle, hence spending
more energy. However, as we will see, the extra expense incurred is
justified since its optimal radius for connectivity is definitely
smaller than the one required by the contiguous model. This highlights
the main contribution of our paper: Theorem~\ref{thm:conn-weak}
conjectured by Das et. al.~\cite{das-icuimc:2012} gives a connectivity
condition which is same for both the random and contiguous schemes
although the random scheme is much more complex than the contiguous
scheme, whereas our optimal connectivity result,
Theorem~\ref{thm:vb-gamma}, establishes that a much lower radius is
needed to achieve connectivity in the random scheme as we will see in
Corollaries~\ref{cor:vb-gamma-DC-C-WSN} and
\ref{cor:vb-gamma-DC-R-WSN} below.
%\ifrevision
%{\bf [revision id \# 12] 
%\fi

\subsection{Experimental setup}
\label{sec:simulation:setup}
In the rest of this section, we  test the weak and the optimal radius by
measuring the percentage of sensors belonging to the largest component in \dcc\
and \dcr\  for various values of $n$, $L$ and $\delta$.
For both duty-cycle schemes, we discuss the benefit 
of the optimal radius over the weak radius as well as
the difference between the optimal radius and the lowest possible radius,
i.e., the RGG radius for connectivity. 
%From now on, the figures  plot on the $Y$ axis 
%the  percentage of
%sensors belonging to the largest component.
%From now on, in the figures, we use the
%expression `Percentage Connectivity' to denote the percentage of
%sensors belonging to the largest component.
%
%\ifcomments
%\comment{If there is time, with gedit I'll change 'Percentage Connectivity' with 'Connectivity Percentage'}
%\fi
%
%In the following, a DC-WSN whose sensors follow the contiguous (resp.,
%random) model will be denoted as {\em DC-C-WSN} (resp., {\em
%  DC-R-WSN}).  In the experiments, we have considered a number of
%sensors $N$ varying between $10^5$ to $10^6$, uniformly and randomly
%distributed in a disk of unit area.  For varying the duty-cycle ratio
%$\delta$, depending on the experiments, we have either fixed $L=100$
%or $d=5$.  Then, the other parameter (resp, $d$ and $L$) is selected
%conveniently with $\delta$.
%Moreover, when $\delta$ is fixed, it is set to $\delta=0.05$.

Our experiments are performed on a workstation equipped with a 4 GHz
Intel processor and 4 GB of main memory.  We implemented the algorithm
in C++. We generates $n$ sensors placed uniformly at random
in a unit disk, with $10^5 \le n \le 10^6$.
To generate uniformly
distributed points we place the points one at a time. For each point
we first choose an angle $\theta$ at random.  Then, for a fixed small
value $\epsilon \theta$, we choose a point uniformly in the triangle
which has one point at the origin $O=(0,0)$, one point at $A=(cos
\theta, sin \theta)$ and one point at $B=(cos (\theta+ \epsilon
\theta), sin (\theta + \epsilon \theta))$.
%CRISTINA
\ifrevisiontwo {\bf\fi Note that, since we generate $n$ sensors in a unit disk, the
  expected number of points that reside in a circle of area $\pi r^2$
  is $\pi r^2 n$.\ifrevisiontwo }\fi
%CRISTINA

%Precisely,
%we first consider the parallelogram $\diamond~ABOO'$ placing a point $O'$ outside this triangle such that  $AOO'$ forms a triangle
%which is congruent (exactly the same) as $\OAB$.
%A uniform point in this parallelogram is selected by
%choosing a point X along the long side and a point Y along the short side
%and then finding the point $Z$ inside the parallelogram such that $\diamond~OXYZ$ forms
%a small parallelogram with the same interior angles as $\diamond ABOO'$. If $Z$ is
%inside $ABO$ then $Z$ is the random point we are looking for, otherwise $Z$ is in
%$AOO'$ and we can find $Z'$ which is the mirror image point of $Z$ inside $ABO$.
%In this case $Z'$ is our random point.

For efficient processing, we store the sensors (i.e., points) in a
$kd$-tree, a spatial data structure that recursively subdivides the
unit disk into boxes till each individual box at the leaf level
contains at most a predetermined number of points. The points
contained in each box were stored in a file. The number of points in a
leaf box is determined such that at least two files could be
simultaneously stored in the main memory.

In regard to the duty-cycle parameters, we select $\delta$ varying
between $0.05$ and $0.5$ since by Fact~\ref{fct:deltahalf}, the model
reduces to \rgg\ for $\delta > 0.5$. Then either we calculate
$d=\lceil \delta L \rceil$ by fixing $L=100$, or we derive $L=\lceil
d/\delta \rceil$ by fixing $d=5$.  In addition to $d$, $L$, and
$r(n)$, some information is stored in the file for each sensor,
depending on the duty-cycle scheme.  For \dcc,\ we store for each
sensor its start time, generated at random.  For \dcr,\ we store a
bitmap of length $L$.  This bitmap will have $d$ of its bits set to
$1$, which implies that the sensor is awake in that time slot, and the
rest of them set to $0$.  To add an edge between two sensors, we check
if they lie within  the connectivity radius and if they share a common
awake slot.
%This property
%ensures that if one extreme of one edge receives the message
%broadcasted by the source also the other extreme will receive the
%message during that cycle or during the next cycle.
To find the connected component, we make use of the Union Find
algorithm.  We initially assign unique flags to each point (i.e., each
sensor belongs to an isolated component). Initially these flags point
to themselves. As and when we get an edge between two points, we
combine the connected components of both the points by pointing the
head flag of one component to the head flag of the other component.
When all edges have been added, the number of the distinct connected
components in the graph and their sizes are traced.  Each experiment
is repeated at least five times and the average value and standard
deviation are reported. The $kd$-tree spatial data structure was used
to process pairs of points in $\theta(n)$ time, and the connected
components were created in $\theta(n\log n)$ time. As a result testing
connectivity for a random graph model with $10^6$ nodes took
approximately 20 minutes on the hardware mentioned above.

\subsection{Weak connectivity condition}
\label{sec:simulation:weak}
Let the {\em weak radius} be the radius $r(n)$ that satisfies
Theorem~\ref{thm:conn-weak}.  Since both contiguous and random selection schemes
satisfies the uniform time coverage situation with $\delta= \lceil d/L
\rceil$, the weak radius for both schemes yields:
\begin{equation}
 \label{eq:radius-weak-formula}
r(n) = \sqrt{\frac{\log n +c(n)}{\pi n \delta}}
\end{equation}
with $c(n) \rightarrow +\infty$ as $n \rightarrow +\infty$. 
In our experiments, we set $c(n)$ to $\log \log n$ if not otherwise stated.

We find that the percentage of connectivity achieved by adopting the
weak radius is extremely high. In all our experiments for \dcr s using
the weak radius all the nodes are part of the largest component (i.e.,
a percentage of connectivity equal to 100\% is reached).  For \dcc s,
as depicted in Figure~\ref{fig:weak-condition}, the percentage of
sensors that belong to the largest connected component is always above
$99\%$.  This result holds already for $n=2*10^5$ although it presents
a slightly larger standard deviation than for the highest values of
$n$ (see Figure~\ref{fig:weak-condition-DC-C-WSN-varN}) and it is true
for any the value of $\delta$ (see
Figure~\ref{fig:weak-condition-DC-C-WSN-varDELTA}). These results,
%CRISTINA
which converge to $1$ much more rapidly than those for all awake RGG reported in Figure~\ref{fig:rgg-radius},
%CRISTINA
make us feel that the weak radius is larger than required and are the first
motivation for our trying to find a better radius.

\begin{figure*}[t]
	\centering
	\subfloat[\label{fig:weak-condition-DC-C-WSN-varN}]{
		\includegraphics[scale=0.38]{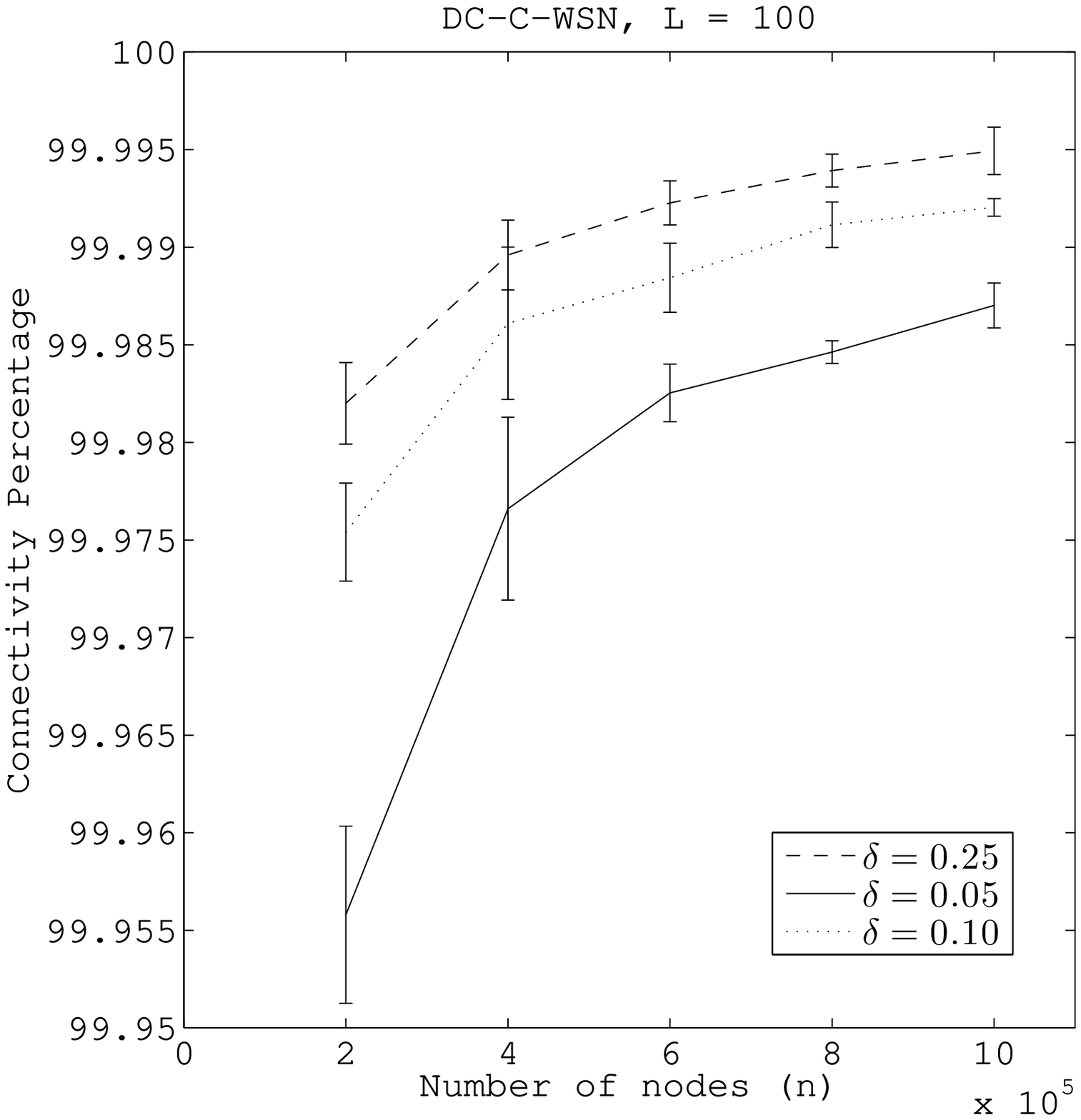}
	}
	%\qquad
	\subfloat[\label{fig:weak-condition-DC-C-WSN-varDELTA}]{
		\includegraphics[scale=0.38]{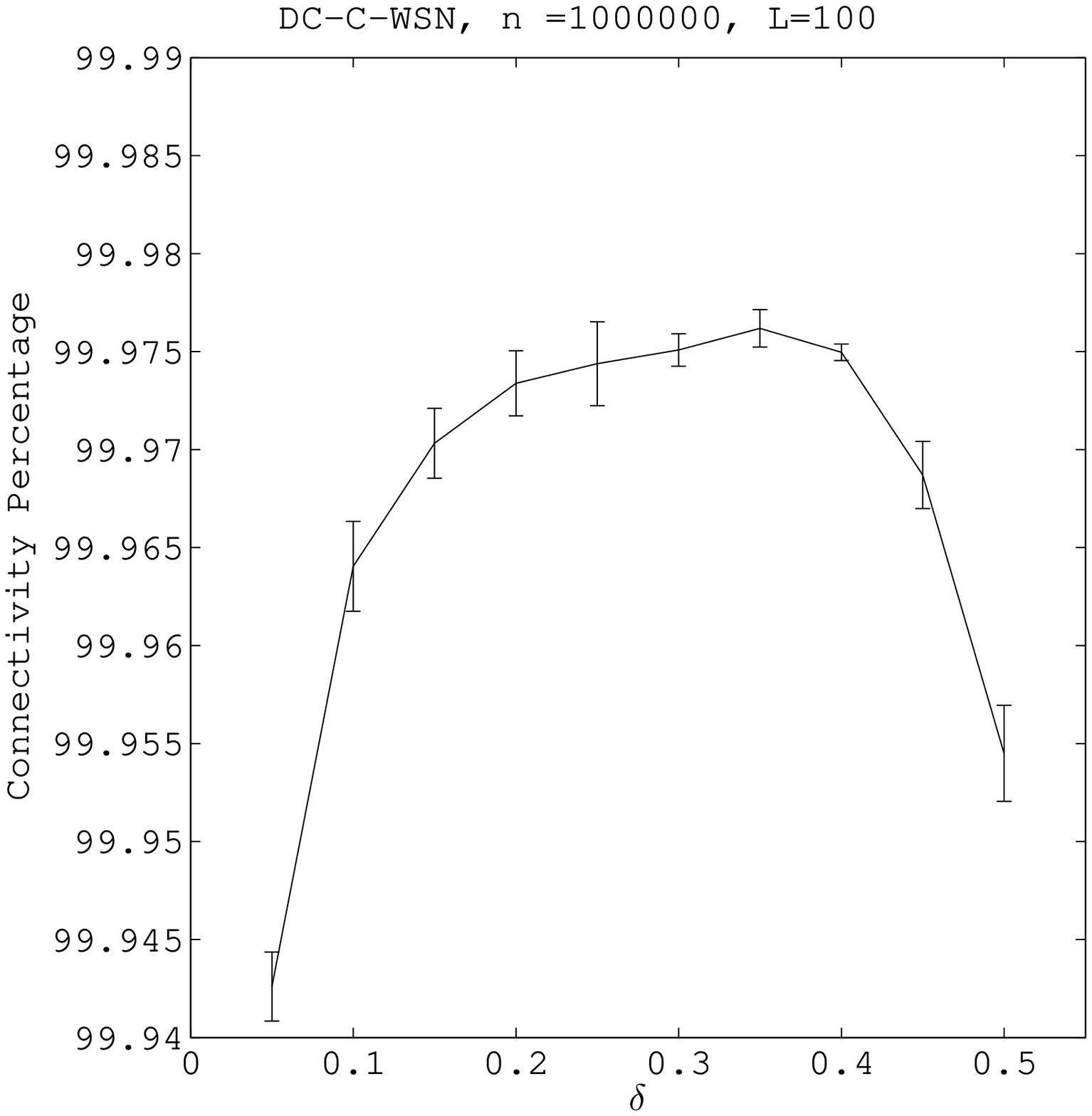}
	}
	\caption{ Weak radius: Percentage of sensors in the largest connected component in DC-C-WSN  when \protect \subref{fig:weak-condition-DC-C-WSN-varN} $n$ varies, or \protect  \subref{fig:weak-condition-DC-C-WSN-varDELTA} $\delta$ varies. The vertical bars represent the standard deviation.} \label{fig:weak-condition}
\end{figure*}
%
%This percentage is even larger than that achieved using the  \rgg
%radius in all awake wireless sensor networks.

Comparing this percentage with the percentage of sensors in the largest component when the radius of connectivity is the \rgg\ radius in regular (i.e. all awake) wireless sensor networks
 (a radius that is $\frac{1}{\sqrt{\delta}}$ lower than the weak radius) we find that at the weak radius this percentage is from 5\% to 10\% higher
(see Figure~\ref{fig:rgg-radius}).
%In fact the weak radius is %much larger than the transmission radius
%required to achieve the connectivity in \rgg. Precisely, it is
% $\frac{1}{\sqrt{\delta}}$ times the \rgg\ radius given in
%~(\ref{eq:radius}). 
In other words the weak radius is not a very useful theoretical result
since we get a connectivity that is not that much higher than the
connectivity at the \rgg\ radius, but we have to transmit
$1/\sqrt{\delta}$ times the distance.  Indeed, on average since the
power spent by each node in transmission is proportional to the square
of the radius and since $n \delta$ sensors are awake in one time slot,
the overall energy spent is almost the same as in regular WSNs,
thereby negating the effect of duty-cycling.

This duty-cycling energy inefficiency is the second motivation that
drove us to find a better theoretical result than that provided by the
weak radius result of~\cite{das-icuimc:2012}. We now move on to
presenting that better result.
\begin{figure*}[t]\centering{
		\includegraphics[scale=0.38]{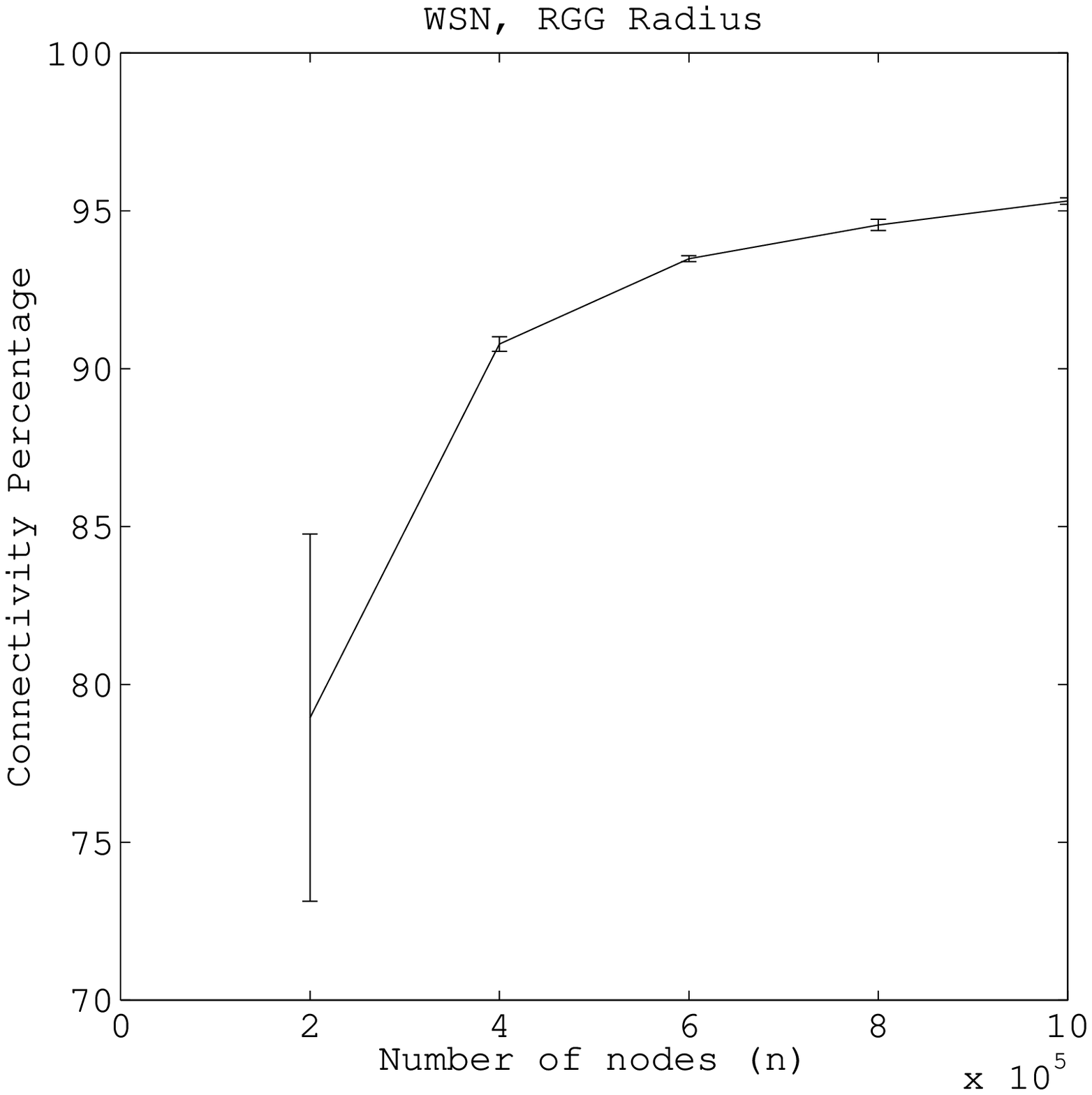}
	}
	\caption{Percentage of sensors in the largest connected component for all awake WSNs adopting the lowest possible radius for connectivity,
i.e., the RGG radius. }\label{fig:rgg-radius}
\end{figure*}

\subsection{Optimal connectivity condition}
\label{sec:simulation:optimal}
Let now concentrate on the strong connectivity result that leads to the {\em optimal} radius, that is  the radius $r(n)$ that satisfies
Theorem~\ref{thm:vb-gamma}. 
To compute the optimal radius
%To apply the optimal connectivity condition in
%Theorem~\ref{thm:vb-gamma} 
for the contiguous and random duty-cycle scheme, we need to compute
the probability $\gamma$ for each of them in the \vbrgg\ model.

In \dcc s, a sensor $v$ will share at least one slot with node $u$ if
$v$ chooses as its starting point any of the slots $i_u-(d-1) \bmod L,
\ldots, i_u, \ldots i_u+(d-1) \bmod L$. Hence, two sensors have
probability $\gamma=\frac{2d-1}{L}$ of sharing a slot.  Now, let the
{\em optimal \dcc\ radius} be the radius $r(n)$ that satisfies
Theorem~\ref{thm:vb-gamma} when $\gamma=\frac{2d-1}{L}$. We have:
\begin{corollary}
\label{cor:vb-gamma-DC-C-WSN}
When $\gamma=\frac{2d-1}{L}<1$,
$\pr(\dcc\ $ $\mbox{ is connected})\rightarrow 1$ as $n
\rightarrow \infty$ if and only if
\begin{equation}
\label{eq:vb-radius-DC-C}
 r(n) = \sqrt{\frac{\log n + c(n)}{(2\delta-1/L) \pi n}},
\end{equation}
where $c = \lim_{n \rightarrow \infty} c(n) = \infty$.
\end{corollary}
Note that if $\gamma = \frac{2d-1}{L} > 1$, the radius in
(\ref{eq:vb-radius-DC-C}) goes below the \rgg\ radius
and is no longer meaningful.
However, $\gamma = \frac{2d-1}{L} > 1$ implies  $\delta > 1/2$, and by Fact~\ref{fct:deltahalf}  
the \rgg\ radius given in~(\ref{eq:radius}) guarantees the connectivity property.

In \dcr s, when a node
$u$ has chosen $d$ slots, another node $v$ has $d$ possibilities to choose
one slot in common with $u$ and the probability of doing that is at least $\delta$ each time.
Hence, the probability that two sensors share one slot is
$\gamma > \left ( 1-(1-\delta)^d\right )$.
Therefore,  by Theorem~\ref{thm:vb-gamma}, the {\em optimal \dcr\  radius} must satisfy:
\begin{corollary}
\label{cor:vb-gamma-DC-R-WSN}
$\pr(\dcr\ \mbox{ is connected})\rightarrow 1$ as $n
\rightarrow \infty$ if and only if
\begin{equation}
\label{eq:vb-radius-DC-R}
 r(n) = \sqrt{\frac{\log n + c(n)}{(1-(1-\delta)^d) \pi n}},
\end{equation}
where $c = \lim_{n \rightarrow \infty} c(n) = \infty$.
\end{corollary}

We first study the size of the largest component under the optimal
radius. In Figure~\ref{fig:strong-n-condition} we plot the the
percentage of sensors that belong to the largest component on the
$y$-axis versus $n$ on the $x$-axis when $L=100$, $10^5 \le n \le
10^6$, $\delta=0.05, 0.15$ and $0.25$. For both schemes, fixing a
value of $\delta$, the size of the largest connected component
increases when $n$ increases. In both \dcr\ and \dcc\ we note that the
percentage of nodes in the largest component decreases as $\delta$
increases, which is expected since $\gamma$ increases as $\delta$
increases and so the optimal radius decreases (see
Table~\ref{table:radius-RGG} for more details).  Nonetheless, for all
the experiments on \dcc s, more than 90\% of the sensors belong to the
largest component.  This is also true for \dcr s for small values of
$\delta$ and even for larger values of $\delta$ for large enough $n$
(see Figure~\ref{fig:strong-n-condition-DC-R-WSN}).  The substantial
drop off in the size of the largest component of \dcr\ for $\delta =
0.25$ is due to the fact that $\gamma$ rises very quickly to 1 in this
case and so the optimal radius of DC-R-WSN
falls quickly down to the \rgg\ radius (see
Table~\ref{table:radius-RGG}).  
%Even in this case, the percentage of
%connectivity is above the percentage reached by regular sensors
%networks (see Figure~\ref{fig:rgg-radius}).
%This shows that the reliance on the asymptotic condition for $n$ is stronger for \dcr s than for \dcc s.

%In order to explain the different behaviour of the two schemes, which have the same
%percentage of connectivity for the extreme values of $\delta$ but decrease differently,

\begin{figure*}[h]
	\centering
	\subfloat[\label{fig:strong-n-condition-DC-C-WSN}]{
		\includegraphics[scale=0.38]{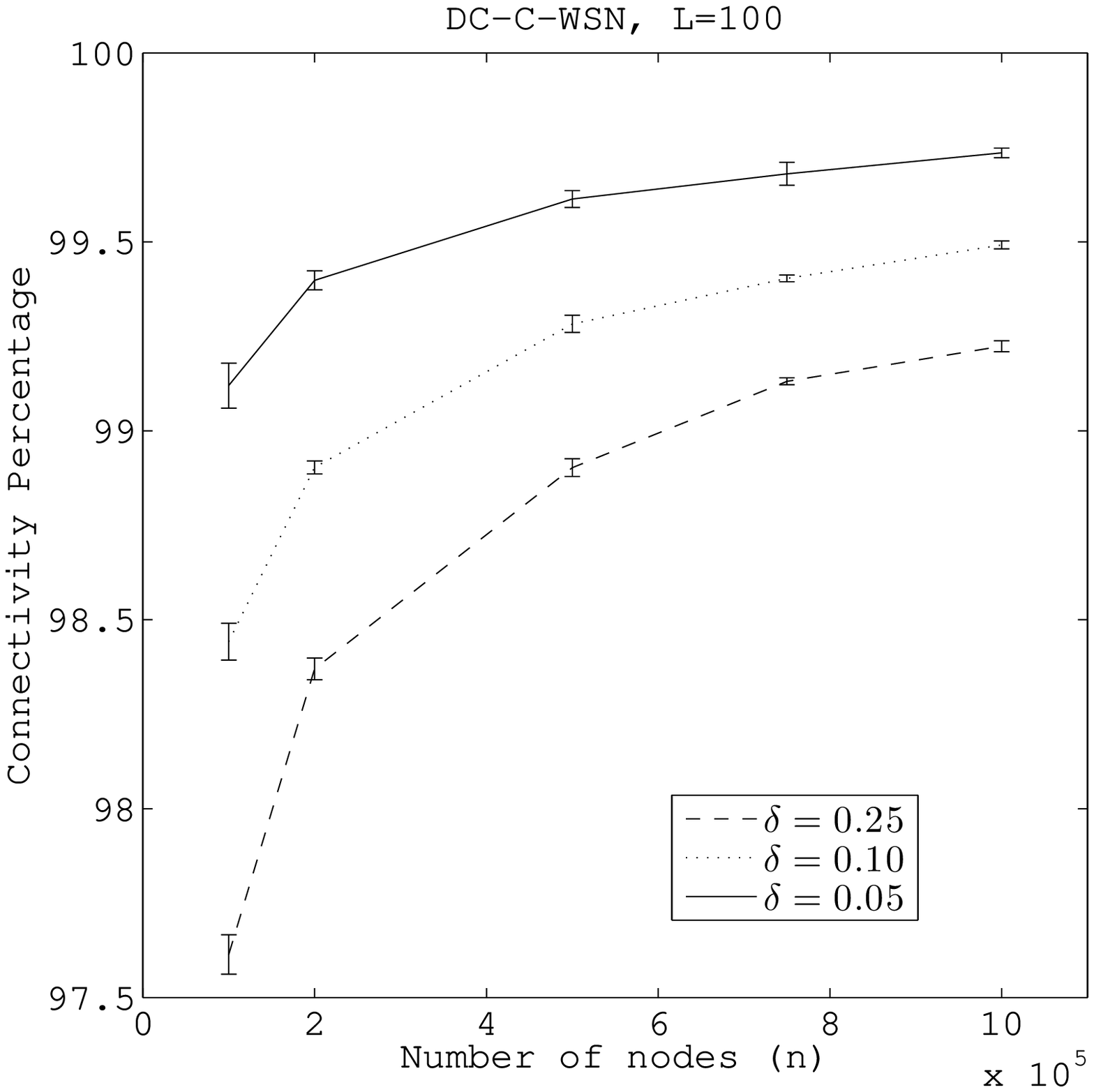}
	}
	%\qquad
	\subfloat[\label{fig:strong-n-condition-DC-R-WSN}]{
		\includegraphics[scale=0.38]{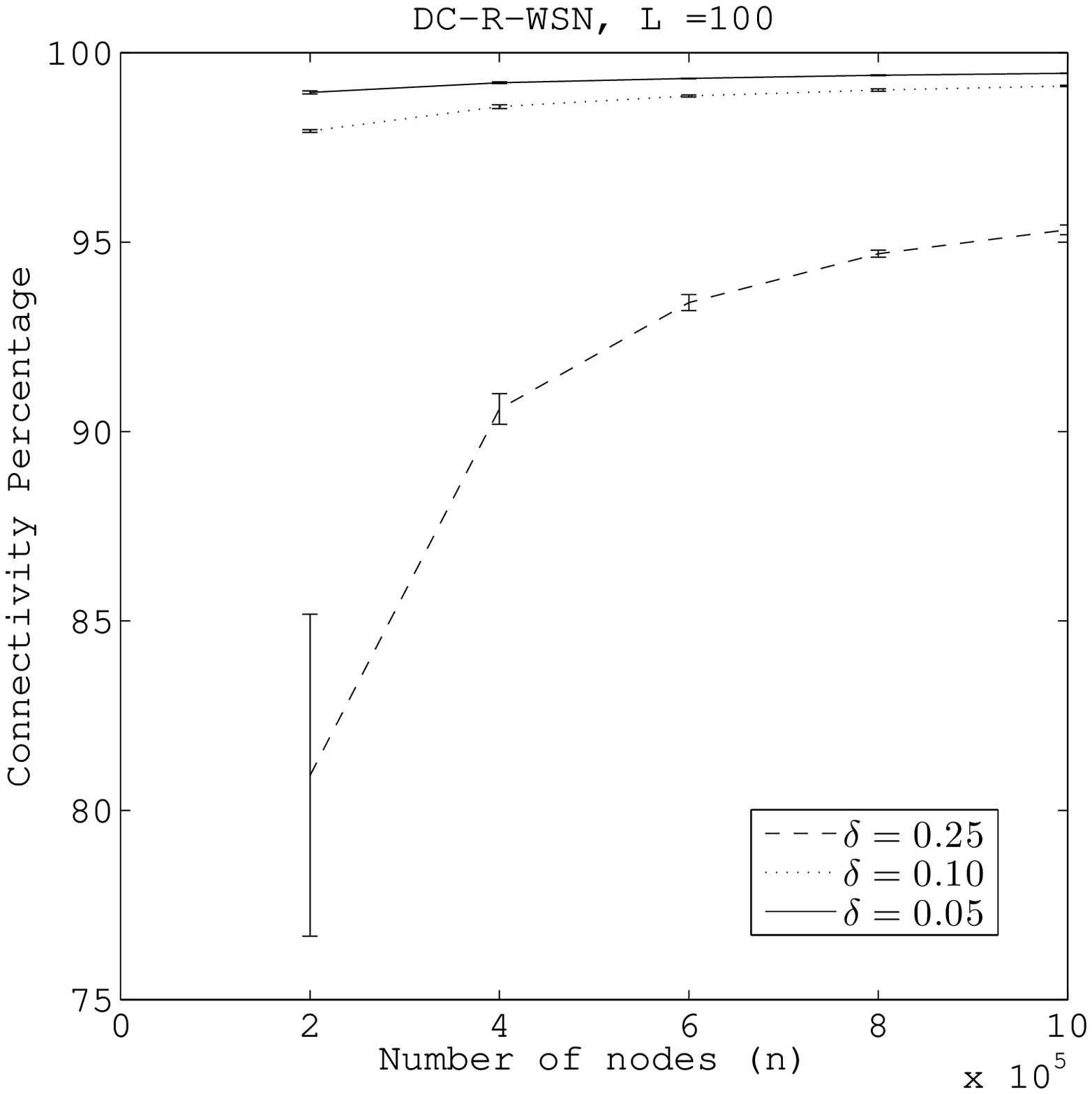}
	}
	\caption{Optimal radius: Percentage of sensors in the largest connected component for different values of $\delta$ when $n$ varies \protect \subref{fig:strong-n-condition-DC-C-WSN} in DC-C-WSN \protect \subref{fig:strong-n-condition-DC-R-WSN}  in DC-R-WSN.}\label{fig:strong-n-condition}
\end{figure*}
\begin{table}[h]
\centering
\begin{tabular}{|l|c|c|c|c|}
\hline
 &  \multicolumn{2}{c|} {$L=200$} & \multicolumn{2}{c|} {$L=100$} \\ \hline
$\delta$ &  \dcc\  & \dcr & \dcc\ & \dcr\ \\ \hline
0.02&  1.322 & 1.970 &  1.224 & 1.407\\ \hline 
0.05&  1.378 & 2.832 &  1.341 & 2.127\\ \hline  
0.10&  1.396 & 2.963 &  1.378 & 2.552\\ \hline %there was a mistake in the previous table, sorry!
0.15&  1.402 & 2.572 &  1.390 & 2.466\\ \hline
0.20&  1.405 & 2.236 & 1.396 & 2.249\\ \hline
0.50&  1.410 & 1.414 & 1.407 & 1.414\\ \hline
\end{tabular}
\caption{The ratio between the weak and the optimal radius in \dcc\ and  \dcr\ for different $\delta$ and $L$.}
\label{table:radius-G-RGG}
\vspace{-.3cm}
\end{table}

We tabulate in Table~\ref{table:radius-G-RGG} the ratio of the weak
and optimal radii in both schemes to explain
Figures~\ref{fig:weak-condition} and ~\ref{fig:strong-n-condition}.
We note that whereas the weak radius is the same for both \dcc\ and
\dcr\ whenever $\delta$ is fixed, there is a radical difference in the
optimal radius. %Our simulation results bear out this finding.  We
observe that as suggested by~(\ref{eq:vb-radius-DC-C}), for \dcc\ the
weak radius is approximately a $\sqrt{2}$ factor longer than the
optimal radius. Such a factor decreases when the cycle length $L$
decreases, but it remains always below $\sqrt{2}$.  For \dcr\ in contrast when $0.05 \le \delta \le 0.4$ and $d > 4$, the ratio
between the weak and the optimal radius
%ratio weak/strong
%$\sqrt{\frac{1/\delta}{1/(1-(1-\delta)^d)}}$ 
is always above
$\sqrt{2}$, implying that \dcr s require a smaller transmission radius
to be connected than the \dcc s. We also note that the ratio is
consistently higher for \dcr\ although as $\delta$ reaches 0.5 the
ratio becomes about the same since $\gamma$ reaches close to 1 for
both schemes.

\begin{figure*}
	\centering
	\subfloat[\label{fig:strong-delta-fixNL-connectivity-DC-C}]{
		\includegraphics[scale=0.35]{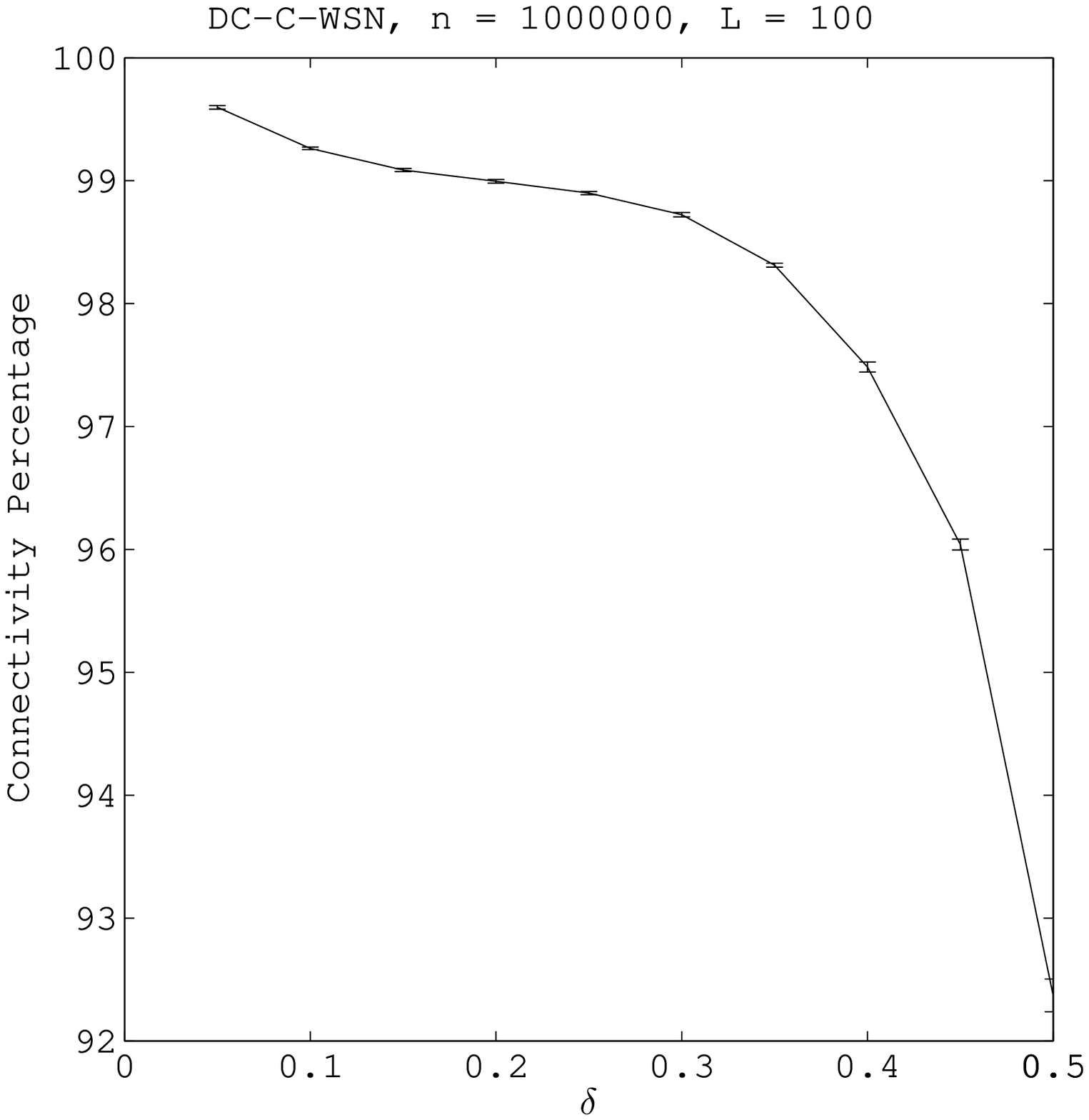}}
	%\qquad
	\subfloat[\label{fig:strong-delta-fixNL-connectivity-DC-R}]{
		\includegraphics[scale=0.35]{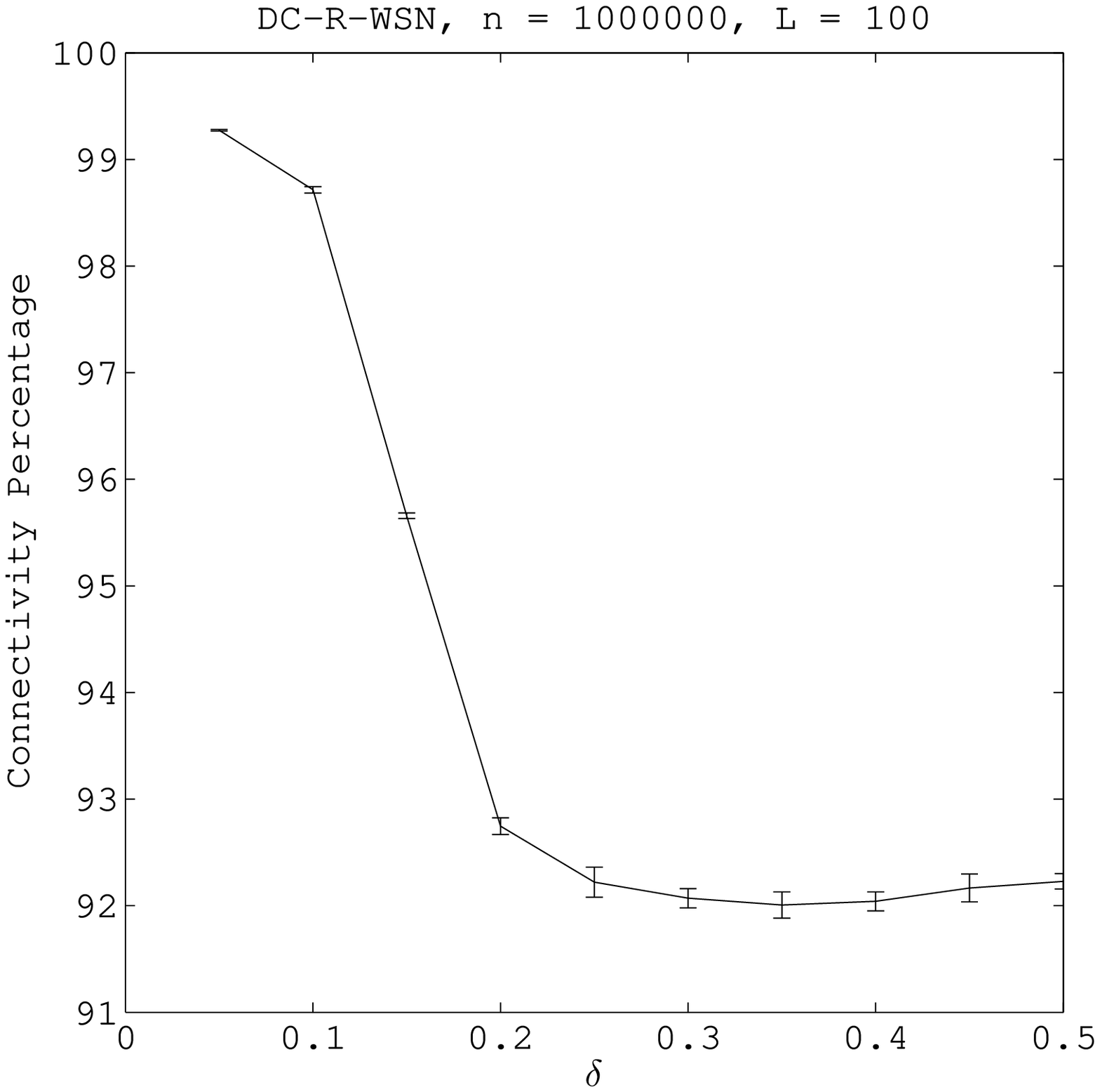}}
	\caption{Optimal radius: Percentage of sensors in the largest connected component when $\delta$ varies, $n$ and $L$ are fixed in \protect  \subref{fig:strong-delta-fixNL-connectivity-DC-C} \dcc s 
 \protect \subref{fig:strong-delta-fixNL-connectivity-DC-R} \dcr s.}\label{fig:strong-delta-fixNL-condition}
\end{figure*}

Figure~\ref{fig:strong-delta-fixNL-condition} depicts the connectivity
percentage of \dcc s and \dcr s under the optimal radius when
$n=10^6$, $\delta$ varies in $[0.1,0.5]$ and $L$ is fixed to $100$. We
see that as $\delta$ increases up to $0.15$, the optimal radius in
\dcr s decreases fast towards the \rgg\ radius, and so the
connectivity of \dcr s in
Figure~\ref{fig:strong-delta-fixNL-connectivity-DC-R} decreases. Past
this point, i.e. for 
larger values of $\delta$, the optimal radius is close to the \rgg\ 
radius and the connectivity in \dcr s remains stable and substantially
close to that of regular sensor networks.  In
Figure~\ref{fig:strong-delta-fixNL-connectivity-DC-C} the connectivity
remains stable and 5\% above that of regular sensors as expected since
the optimal radius is approximately $3/2$ times the \rgg\ radius and
decreases slowly.  Clearly, it appears that, for both schemes, the
connectivity performance decrease happens when the radius approaches
the \rgg\ radius model, but the performance is never worse than that
of regular networks.

This observation is backed up by Table~\ref{table:radius-RGG} where
the weak and the optimal radius are compared with the minimum possible
radius given by the \rgg\ model for two different values of $L$.
\begin{table}[h]
\centering
\begin{tabular}{|l|c|c|c|c|}\hline
 &  \multicolumn{2}{c|} {$L=200$} & \multicolumn{2}{c|} {$L=100$} \\ \hline
$\delta$ & \dcc\ & \dcr\ & \dcc\ & \dcr\  \\ \hline
0.02& 5.345 &  3.589  & 5.773 & 5.025\\ \hline
0.05& 3.244 &  1.578 & 3.333 & 2.102\\ \hline
0.10& 2.264 & 1.067 &2.294 & 1.239\\ \hline
0.15& 1.841  & 1.003 & 1.857 & 1.046\\ \hline
0.20& 1.591 & 1.000 & 1.601 & 1.005\\ \hline
0.50& 1.002  & 1 & 1.005 & 1.005\\ \hline
\end{tabular}
\caption{The ratio between the optimal radius and the minimum possible transmission radius (i.e, the \rgg\ radius).}
\label{table:radius-RGG}
\vspace{-.3cm}
\end{table}

We also studied the influence of $L$ on the percentage of connectivity
in Figure~\ref{fig:strong-delta-fixNd-condition}. Here we fix $d$,
while $L$ changes accordingly to $d/\delta$.  One immediately notes
that the drop off in \dcr s is less than the drop off in
Figure~\ref{fig:strong-delta-fixNL-condition}.  In the case of \dcr s,
this different behaviour is due to the fact that, as reported in
Table~\ref{table:optimal-dcc-delta-varies}, the optimal radius for the
\dcr s decreases slower when $L$ is fixed than when $d$ is fixed.  For
\dcc s, the variation of the optimal radius when $\delta$ changes and
either $d$ or $L$ are fixed is minimal showing that \dcc s are less
influenced by $L$.  To be precise, the optimal radius when $d=5$ is
slightly greater than when $L=100$ and so is the connectivity, which
remains above 97\% even when $\delta=0.5$.
%So both schemes perform better when $d$ is fixed.

\begin{figure*}
	\centering
	\subfloat[\label{fig:strong-delta-fixNd-connectivity-DC-C}]{
		\includegraphics[scale=0.35]{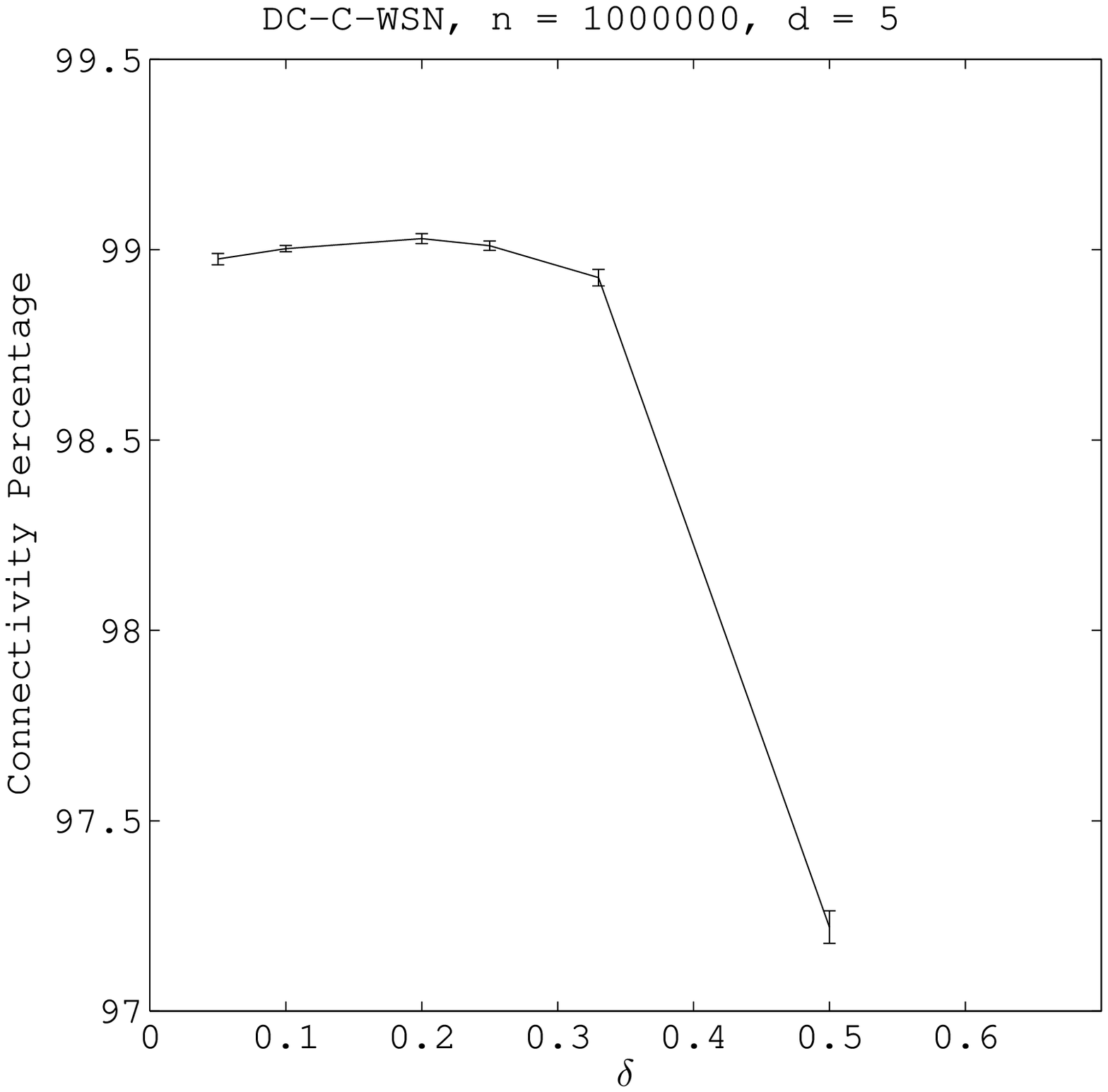}}
	%\qquad
	\subfloat[\label{fig:strong-delta-fixNd-connectivity-DC-R}]{
		\includegraphics[scale=0.35]{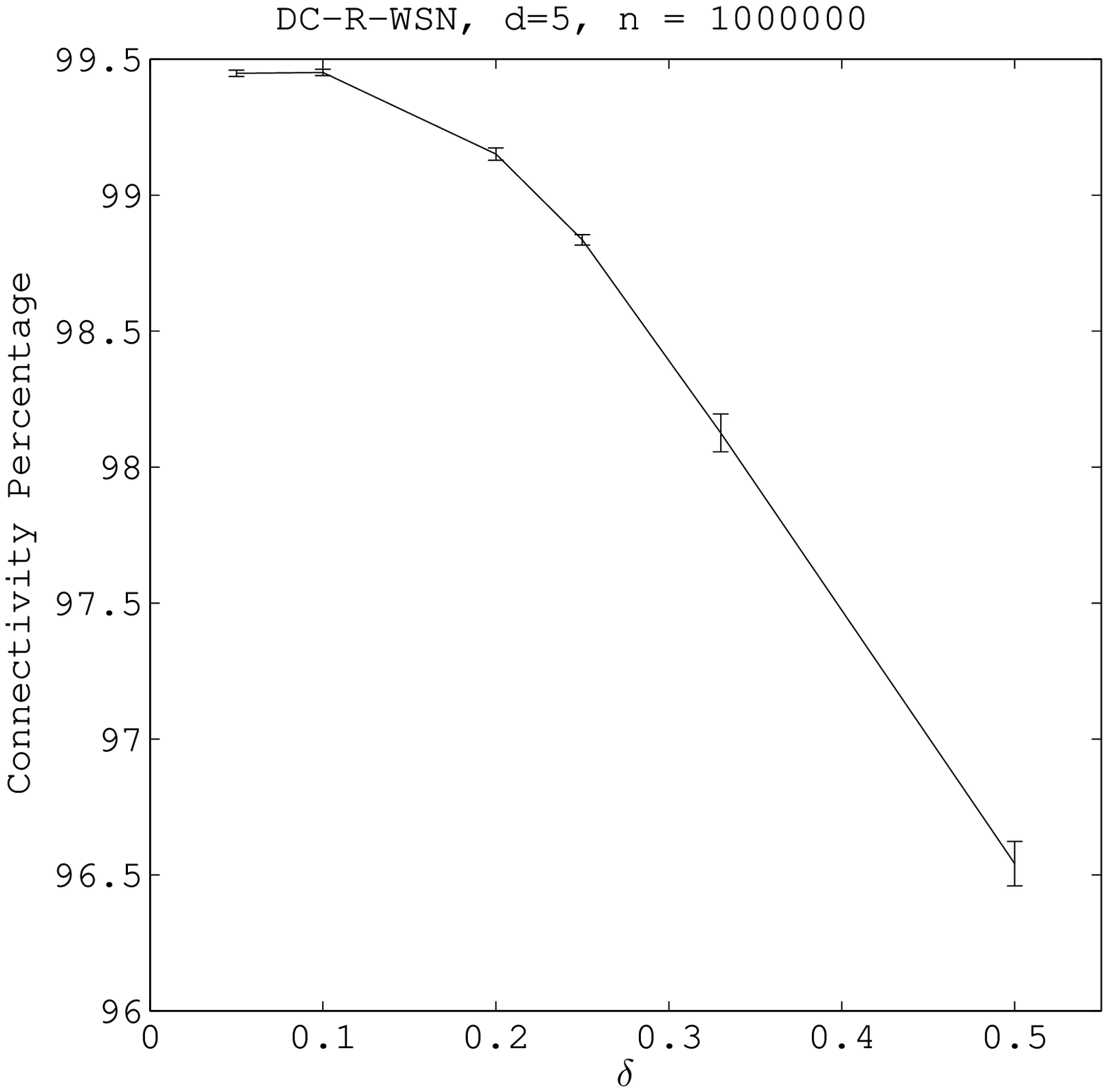}}
	\caption{Optimal radius: Percentage of sensors in the largest connected component when $\delta$ varies, $n$ and $d$ are fixed in \protect \subref{fig:strong-delta-fixNd-connectivity-DC-C} \dcc s 
\protect \subref{fig:strong-delta-fixNd-connectivity-DC-R} \dcr s.}\label{fig:strong-delta-fixNd-condition}
\end{figure*}

\begin{table}
\centering
\begin{tabular}{|l|c|c|c|c|}\hline
 &  \multicolumn{2}{c|} {\dcc s} & \multicolumn{2}{c|} {\dcr s} \\ \hline
$\delta$ & $d=5$ & $L=100$ & $d=5$ & $L=100$  \\ \hline
0.05& 3.333 &  3.333 & 2.102 & 2.102\\ \hline
0.10& 2.357 & 2.294 & 1.562 & 1.239\\ \hline
0.20& 1.667 &  1.601 &  1.219 & 1.005\\ \hline
0.40& 1.178  & 1.125 & 1.042 & 1 \\ \hline
\end{tabular}
\caption{The ratio between the optimal radius and the minimum possible transmission radius (i.e, the \rgg\ radius) when $\delta$ varies 
and either $d$ or $L$ are fixed.}
\label{table:optimal-dcc-delta-varies}
\vspace{-.3cm}
\end{table}

\ignore{

Indeed, when $L$ is fixed and $\delta$ increases, $d$ increases too.
To corroborate our finding we run an experiment with $n=10^6$, $d=5, 10$ and $20$.
The results are reported in Figure~\ref{fig:dcc-sensible-to-d}.

\begin{figure*}
	\centering
	{
		\includegraphics[scale=0.35]{finalfigures/optimal-DC-C-fixN-vardDELTAnew.eps}
	}
	\caption{Percentage of sensors in the largest connected component in \dcc s when $d$ and $\delta$ varies, and $n$ is fixed.}\label{fig:dcc-sensible-to-d}
\end{figure*}

}

We now show that under the strong connectivity condition, the energy
saving is effective.  On average, \dcc s spend half the energy of the
regular WSNs since there are $n \delta$ awake sensors and each sensor
transmits with energy proportional to $\frac{1}{2 \delta}$ times the
energy spent by a sensor in an always-awake WSN.  A higher saving is
possible for \dcr s.  In fact, each awake sensor transmits with energy
proportional to the energy spent by a sensor in an always awake WSN
multiplied by $\frac{1}{1-(1-\delta)^d} \approx
\frac{1}{1-e^{-d\delta}}$, which becomes very close to $1$ when
$\delta > 0.20$.  Hence, \dcr s spend energy proportional to the
number $n \delta$ of awake sensors, which is the most desirable
situation.  In conclusion, the optimal connectivity condition
undoubtedly leads to a great gain in the radius length, and thus leads
to a great energy saving in power transmission for both schemes, but
especially for \dcr s.

\begin{figure*}[t]
	\centering
	\subfloat[\label{fig:cc1-1}]{
		\includegraphics[scale=0.35]{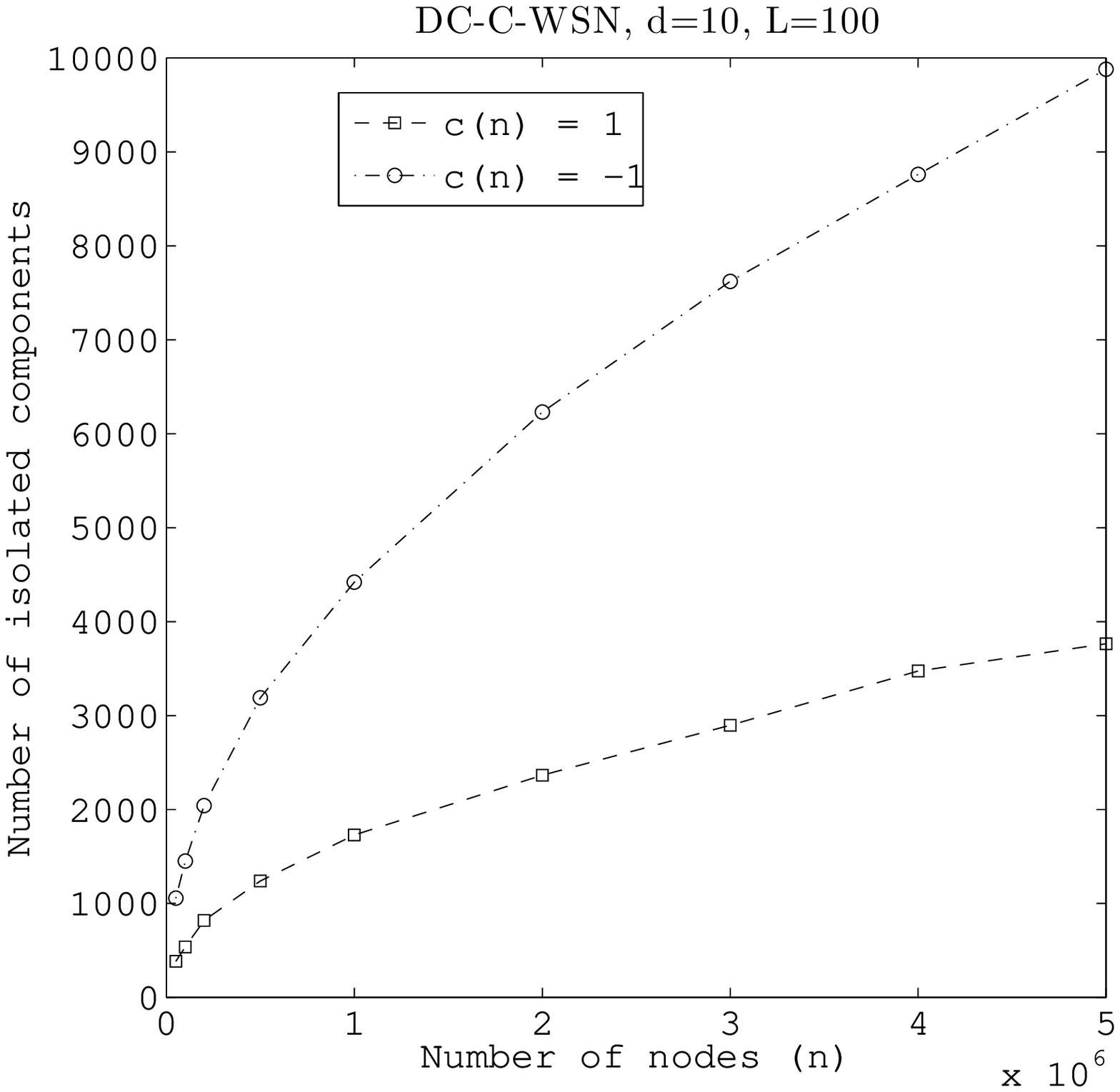}
	}
	%\qquad
	\subfloat[\label{fig:largestsecond}]{
		\includegraphics[scale=0.35]{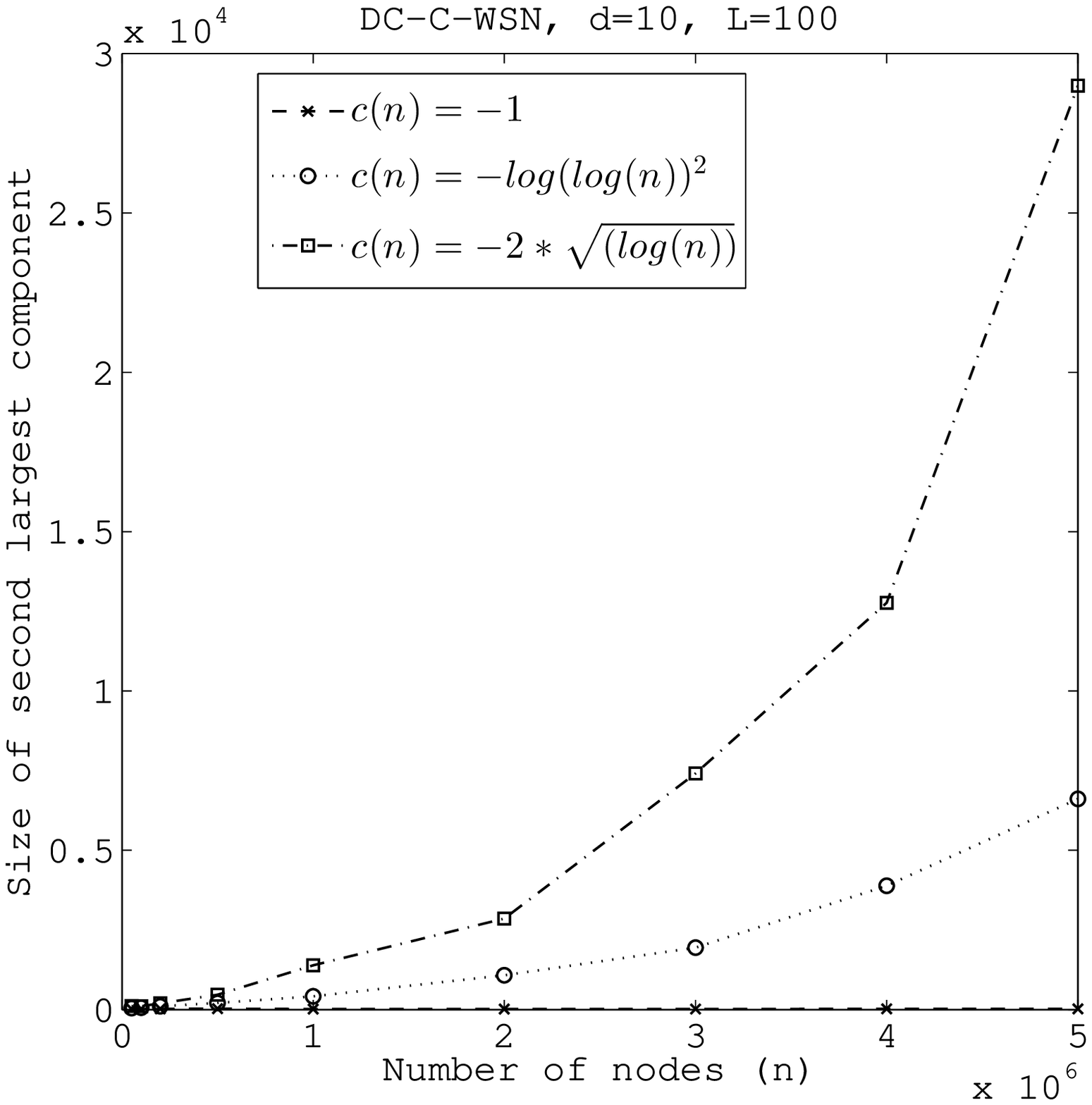}
	}
	\caption{\protect \subref{fig:cc1-1} The number of isolated points in \dcc . \protect \subref{fig:largestsecond} The size of the largest second component in \dcc s.}\label{fig:connected component}
\end{figure*}

To establish that the radius condition given in
Corollary~\ref{cor:vb-gamma-DC-C-WSN} is necessary we studied the
situation where $c(n)$ does {\em not} grow to $\infty$ as $n$
grows (Figure~\ref{fig:connected component}).
Figure~\ref{fig:cc1-1} shows that although the connectivity is still
high in \dcc\ when $c(n)=1$ and when $c(n) = -1$, the number of
isolated nodes is rapidly increasing, the increase
being faster in the case where $c(n) = -1$. This indicates that the
percentage connectivity is continuously dropping.  Moreover,
Figure~\ref{fig:largestsecond} shows that when $c(n) < 0$, the size of
the second largest component increases as $|c(n)|$ increases which
implies that the probability of connectivity tends to 0 as $n
\rightarrow \infty$ if $\lim_{n \rightarrow \infty} c(n) =
-\infty$. These results show the necessity of the condition on
$c(n)$.

We conducted a series of experiments to establish the optimality of
the optimal radius.  In particular, we varied the additive factor
$c(n)$ in the optimal connectivity condition.  As we expect from the
previous discussion, since the optimal \dcr\ radius falls more sharply
than the optimal radius for \dcc\ (and in fact is not far from the
minimum \rgg\ radius), the connectivity in \dcr s drops before it does
in \dcc s.  Figure~\ref{fig:final-percentage} shows that when
$c(n)=-\log \log (n)$ \dcr s are below 10\% of connectivity
independent of $n$, while \dcc s still reach a good percentage of
connectivity, especially for large $n$.  Figure~\ref{fig:constant}
shows for which values of $c(n)$ both schemes experience a comparable
and drastic loss of connectivity, dropping below $0.2$\%.  For \dcc s,
this happens between $c(n)=-(\log{\log{n}})^2$ and
$c(n)=-2\sqrt{\log{n}}$ (see Table~\ref{table:cn} for the absolute
values), while for \dcr s this happens when $c(n)=-5$.
\begin{table}
\centering
\begin{tabular}{|c|c|c|c|}\hline
& \multicolumn{3}{c|}{$c(n)$} \\ \hline
n & -$(\log \log n)^2$ & -2*$\sqrt{\log n}$ & -2.5*$\sqrt{\log n}$ \\ \hline
$0.5*10^6$ & -6.60 & -5.25 & -6.57\\ \hline
$10^6$ & -6.86& -7.43  & -9.29 \\ \hline
$1.5*10^6$ & -7.02& -9.10& -11.38 \\ \hline
$2*10^6$ & -7.12 & -10.51 & -13.14\\ \hline
$2.5*10^6$ & -7.23 & -11.75 & -14.69\\ \hline
\end{tabular}
\caption{The values of $c(n)$ when $n$ varies.}
\label{table:cn}
\vspace{-.3cm}
\end{table}

\begin{figure*}

	\centering
	\subfloat[\label{fig:percentage-constant-DC-C}]{
		\includegraphics[scale=0.33]{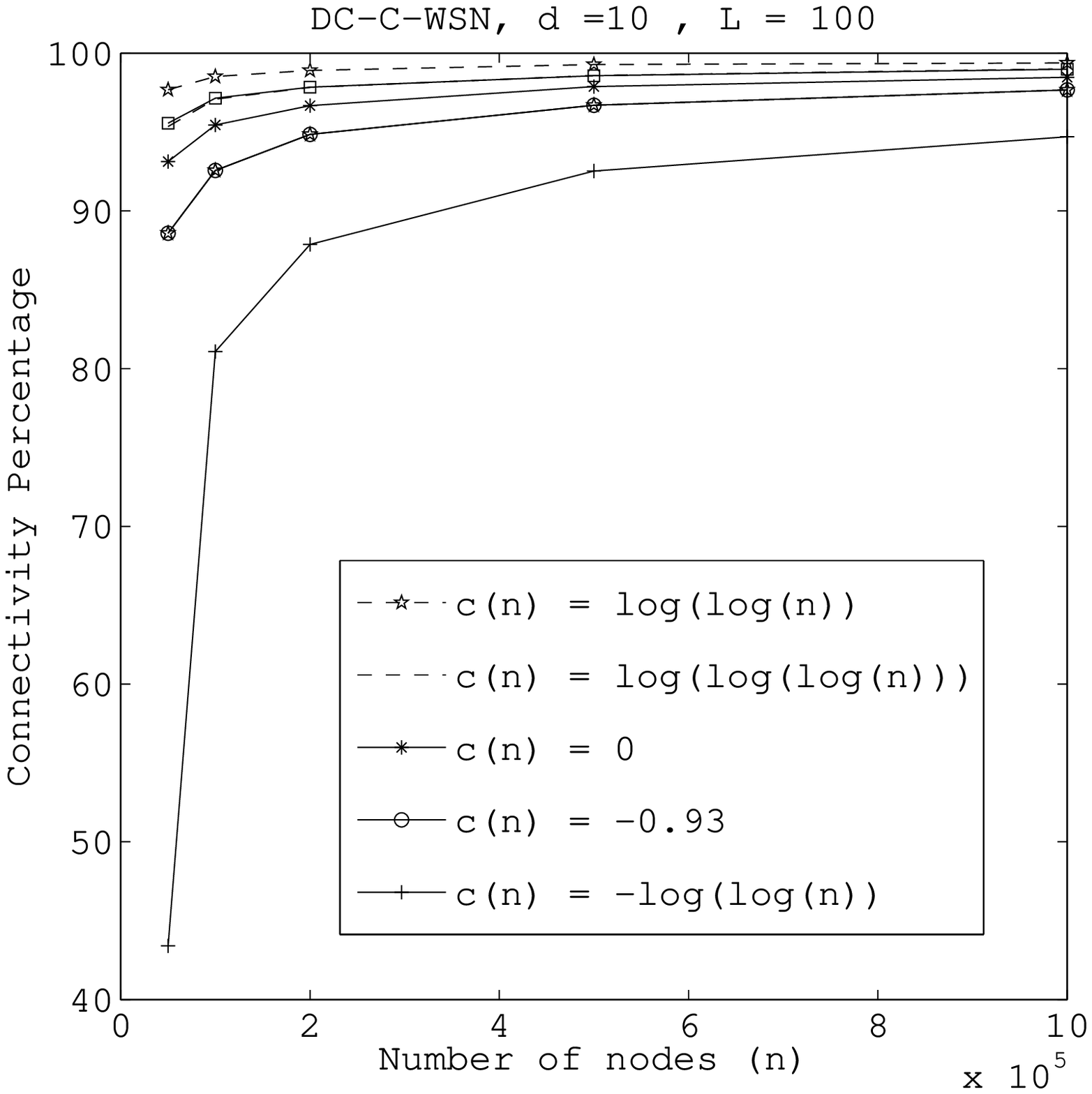}
	}
	%\qquad
	\subfloat[\label{fig:percentage-constant-DC-R}]{
		\includegraphics[scale=0.33]{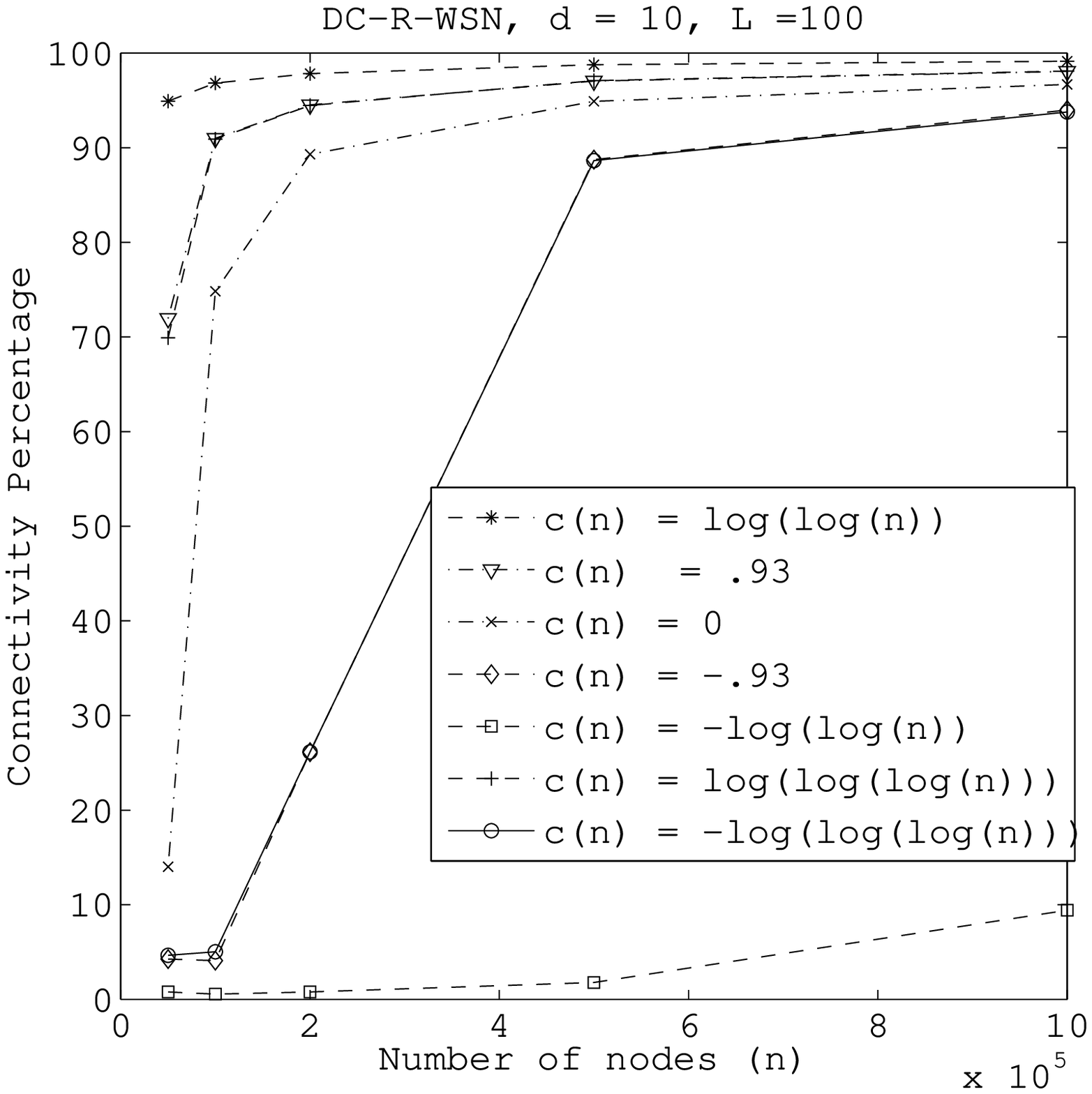}
	}
	\caption{ Percentage of connectivity \protect \subref{fig:percentage-constant-DC-C} in DC-C-WSN \protect \subref{fig:percentage-constant-DC-R}  in DC-R-WSN.}\label{fig:final-percentage}
\end{figure*}

\begin{figure*}

	\centering
	\subfloat[\label{fig:constant-DC-C}]{
		\includegraphics[scale=0.33]{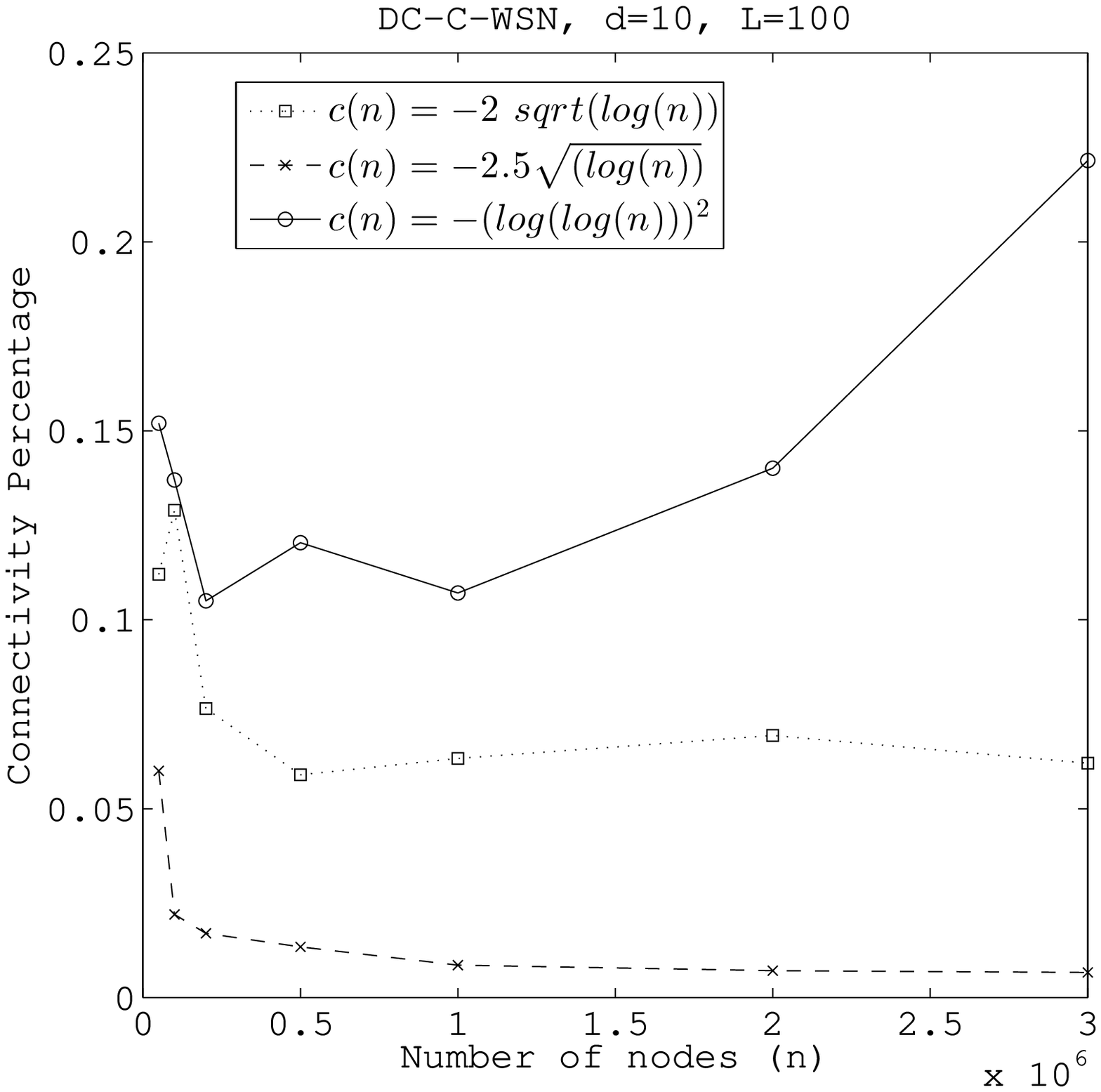}
	}
	%\qquad
	\subfloat[\label{fig:constant-DC-R}]{
		\includegraphics[scale=0.33]{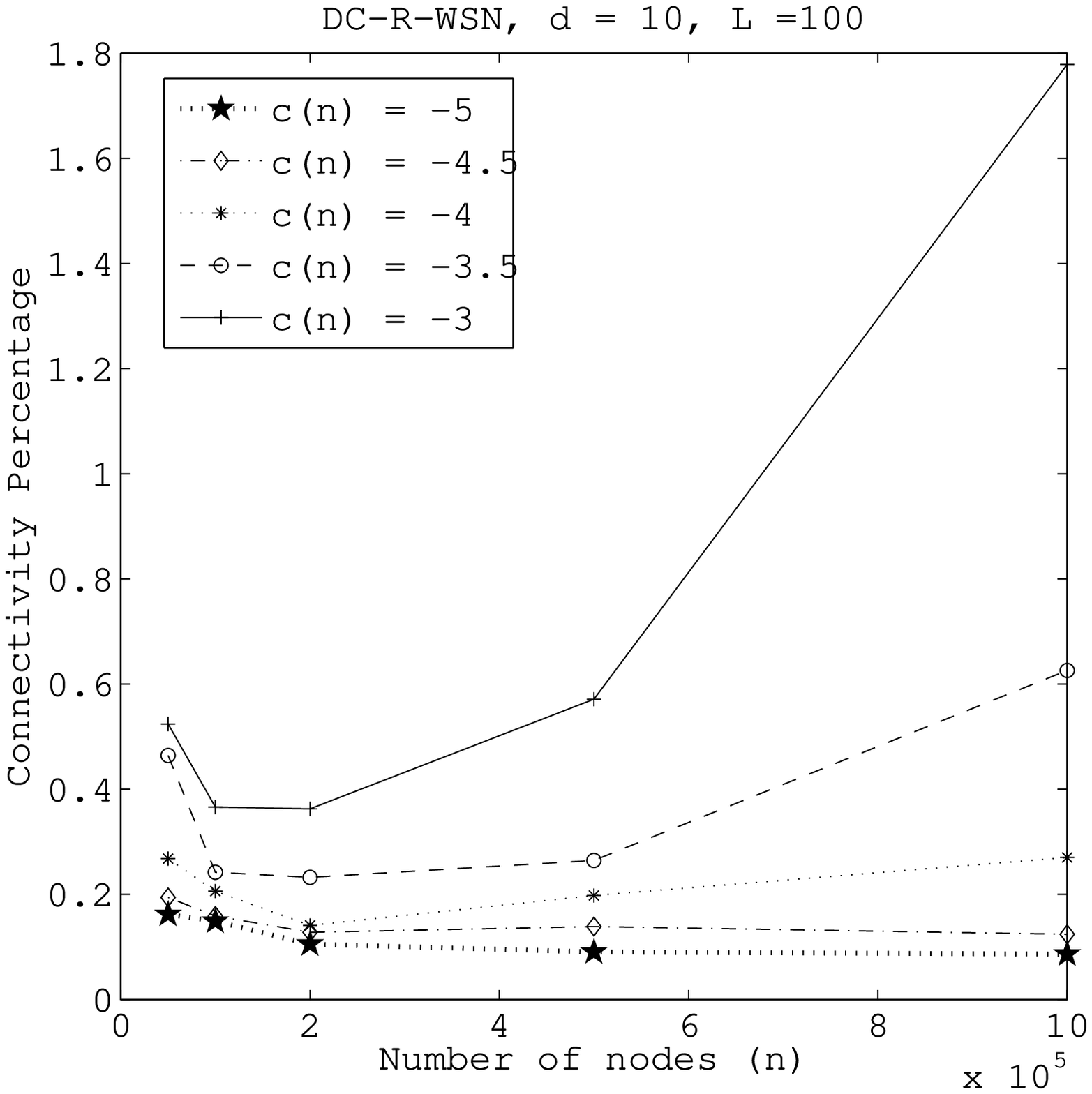}
	}
	\caption{Percentage of connectivity with different values of $c(n)$ in \protect \subref{fig:constant-DC-C} \dcc \protect \subref{fig:constant-DC-R} \dcr.}\label{fig:constant}
\end{figure*}

\subsection{Power consumption in DC-C-WSNs and DC-R-WSNs}
\label{sec:simulation:power}

The main focus of this paper is to study the minimum radius that
guarantees connectivity of the duty-cycling graph for a large family
of duty-cycling schemes, i.e. to study the minimum {\em local}
transmission power spent by a duty-cycled sensor. It is not clear that
minimizing local power leads to minimum {\em global} power consumption
for any given set of packets being sent across the network. We do not
attempt any comprehensive characterization of the power consumption of
different routing problems in the duty-cycled case since that is
outside the scope of this paper. Instead we pick a simple
point-to-point communication task and compare the power consumption of
the weak and optimal radii of \dcr\ and \dcc\ for that task. We also
study the power consumption of the ``always on'' network that uses the
\rgg\ radius to connect which corresponds to a local power level that
is lower than the local power levels of the duty-cycled networks we
study.

\paragraph{Simulation setup}
In the experiments related to global power consumption, we use the
sensor power consumption levels given in
Table~\ref{tab:powerconsum}~\cite{Shin2014,TCS08}. To calibrate our
experiments we set $P=50 mW$ as the power spent for transmitting up to
the \rgg\ radius $r(n)$ (see Equation~\ref{eq:radius}) with $n=2*10^5$
and $c(n)= \log \log n$.  The transmission power $P(\overline{r}(n'))$
to cover an arbitrary radius $\overline{r}(n')$ is then given by
$P(\overline{r}(n')) = P(r(n)) \cdot
\frac{\overline{r}^2(n')}{r^2(n)}$. Note that the radius
$\overline{r}(n')$ varies according to the duty-cycle scheme and the
number of sensors.  Moreover, since the transmission is considered
successful if the SINR (Signal-to-Interference-plus-Noise-Ratio) is
above a certain threshold, we assume that the power for receiving $M$
is independent of the distance between the transmitting and the
receiving sensor.

\begin{table}[t]
\caption{Estimate of sensor power consumption in 
different operational modes at $2.5$ Volt and with a sensor \rgg\ transmission radius $r(n)= \sqrt\frac{\log n + \log\log n}{n}$  and $n=2*10^5$.}
\label{tab:powerconsum}
\footnotesize
\begin{center}
\begin{tabular}{|c|c|c|}\hline
{\bf Sensor Mode} & {\bf Current Drawn} & {\bf Power Consumed} \\ \hline
Sleep  & & \\ 
(CPU inactive, timer on, radio off)  & 6 $\mu A$ & 0.015 $mW$ \\ \hline
CPU switch on, radio startup  & 3 $mA$ & 15 $mW$ \\ \hline
CPU switch off, radio shutdown & 3 $mA$ & 15 $mW$ \\ \hline
Awake   & &  \\ 
(CPU active, radio listening or RX) & 12 $mA$ & 32 $mW$ \\ \hline
CPU active, radio TX  & 20 $mA$ & 50 $mW$ \\ \hline
\end{tabular}
\end{center}
\end{table}

We assume that if a sensor is awake, the CPU is active and the sensor
listens to the radio.  In this way, awake sensors can receive the
incoming message $M$ without extra overhead~\cite{Shin2014}.
Moreover, we do not consider interference in our simplified
simulations. 
%It is likely to say that $M$ is a binary message
%'hear' or 'not hear'. Since no content has to be transferred, interferences do not destroy it.
Power is also spent to switch between sleep and awake modes, i.e., to
switch on/off both the CPU and the radio.  Note that, in this model,
the only extra power for performing any network task is the power
spent in actually transmitting packets which is directly proportional to
the square of the current transmission radius. We will see ahead that
the power consumed by transmissions in our network task is the least
in the ``always on'' case which is expected since the radius of
transmission is the lowest. But, the power required to operate the
network in the ``always on'' case is much higher. For example, a
simple calculation shows that over 100 times slots the ``always on''
network consumes 3.2 W per node only for operation without any
transmission. On the other hand for the duty-cycled case, setting
$\delta = 0.05$ and $L = 100$ i.e. with $d = 5$, assuming one
transition from sleep to waking and one from waking to sleep in the
\dcc\ case we find that the total power consumed at a node is 0.19W
which is 16 times lower than the ``always on'' case. Since in
\dcr\ the number of transitions between sleep and waking are more the
power consumed is slightly higher but is no more than 0.31W (counting
five transitions from sleep to wake and five in the opposite
direction) which is more than 10 times lower than the ``always on''
case. Therefore we see an order of magnitude difference in the
operation of the network that more than compensates (as we will see)
for the extra power spent in transmission in most network tasks.

\paragraph{The network task}
The network operation we will study, $\Broad(M,S,D)$ involves sending
message $M$ from source node $S$ to destination node $D$. We define
the {\em total transmission power} of this operation as the sum of the
power spent {\em in transmission only} by each sensor that receives
and retransmits $M$ during the broadcast operation. We experimentally
measure by simulations the total transmission power for the two
concrete \dcc\ and \dcr\ schemes when the weak and the optimal radius
are used.  We also consider the total transmission power for
$\Broad(M,S,D)$ on always awake WSNs, where each sensor transmits
adopting the \rgg\ radius.  The results presented, unlike our previous
results in Theorems~\ref{thm:conn-weak} and ~\ref{thm:vb-gamma}, have
no general validity. We claim them only for the two specific
duty-cycling schemes we have studied so far, and that too only for the
operation $\Broad(M,S,D)$.

%
% BEGIN IGNORE
%
\ignore{
\begin{figure*}
	\centering
		\includegraphics[scale=0.40]{finalfigures/5_no-operation_power_bw.eps}
	\caption{Power consumed by the duty-cycled and always awake sensors during a period of $L=100$ time slots when $\delta=0.05$.}\label{fig:5_nooperation}
\end{figure*}
}
%
% END IGNORE
%

In our experiments, we select the source $S$ for the $\Broad(M,S,D)$
operation in the center of the deployment area, which is a disk of
unit radius, and the destination $D$ is an arbitrary sensor at
Euclidean distance $d(S,D)=0.1$ from $S$.  We assume that each sensor
knows its polar coordinates in the deployment. 

In order to a simulate realistic situation we implement
$\Broad(M,S,D)$ using a greedy algorithm which is a kind
of partial flooding of the network as follows: 
\begin{enumerate}
\item $S$ sends the message $M$ to all its neighbors. 
\item Node $u$ retransmits $M$ only the first time it receives the
  message under the following condition: u transmits M for d slots
  (i.e. for one duty-cycle) if it receives from a sender v which is
  further from D than u itself. 
\end{enumerate}

We call this algorithm the {\em greedy directional} algorithm since it
tries to move the message in the direction of the destination i.e. the
distance between transmitting sensor and $D$ is guaranteed to decrease
as the hop count of the message increases. We also implement a {\em
  relaxed greedy directional} (or simply, {\em relaxed greedy})
algorithm which is the same as the greedy directional algorithm except
that a sensor retransmits $M$ if its Euclidean distance from $D$ is no
larger than $1.2$ times the distance of the sender i.e. in the relaxed
greedy algorithm, the distance between the transmitting sensor and $D$
is not guaranteed to decrease at every hop, although it increases by a
bounded amount 

Whenever not otherwise specified, $\Broad(M,S,D)$ is implemented by
the greedy directional algorithm. We use the relaxed greedy for the
cases where greedy directional does not find a path (despite the
network being connected) e.g. in the ``always on'' case that uses the
\rgg\ radius and \dcr\ when it uses the optimal radius. 
\ifrevision
{\bf @Cristina: Please see this explanation for relaxed greedy:::::
\fi
The failure
of the greedy to find a path in these situations is due to the fact
that since the radii here are very low, the paths from $S$ to $D$ can
sometimes encounter twists and turns and may not always move in the
desired direction.
\ifrevision
}
\fi

\paragraph{Results and analysis}
To contextualize the power consumption of $\Broad(MSD)$ we began by
plotting the number of hops needed to complete the task under the
\dcr\ and \dcc\ with both the optimal radius and the weak radius, as
well as the number of hops needed by the ``always on'' network using
the \rgg\ radius (Figure~\ref{fig:1_hops}). Additionally
we plotted the number of time slots taken to complete the task in the
five cases considered above (Figure~\ref{fig:4_timer}). We found that
the number of hops increases more or less as the radius decreases,
with the \rgg\ scheme having the maximum number hops. Although the
weak radius is the same for \dcc\ and \dcr\, the number of hops to
reach $D$ is smaller in the latter case because of the greater
probability of connection in \dcr. We find that for the duty-cycled
schemes the completion times report in Figure~\ref{fig:4_timer} are in
the same order as the number of hops i.e. lower the number of hops for
a scheme, the quicker the message reaches the destination. The
\rgg\ scheme has a much lower completion time than the duty-cycled
schemes, which is to be expected since sensors are always on in this
case.
\begin{figure*}
	\centering
	\subfloat[\label{fig:1_hops}]{
		\includegraphics[scale=0.33]{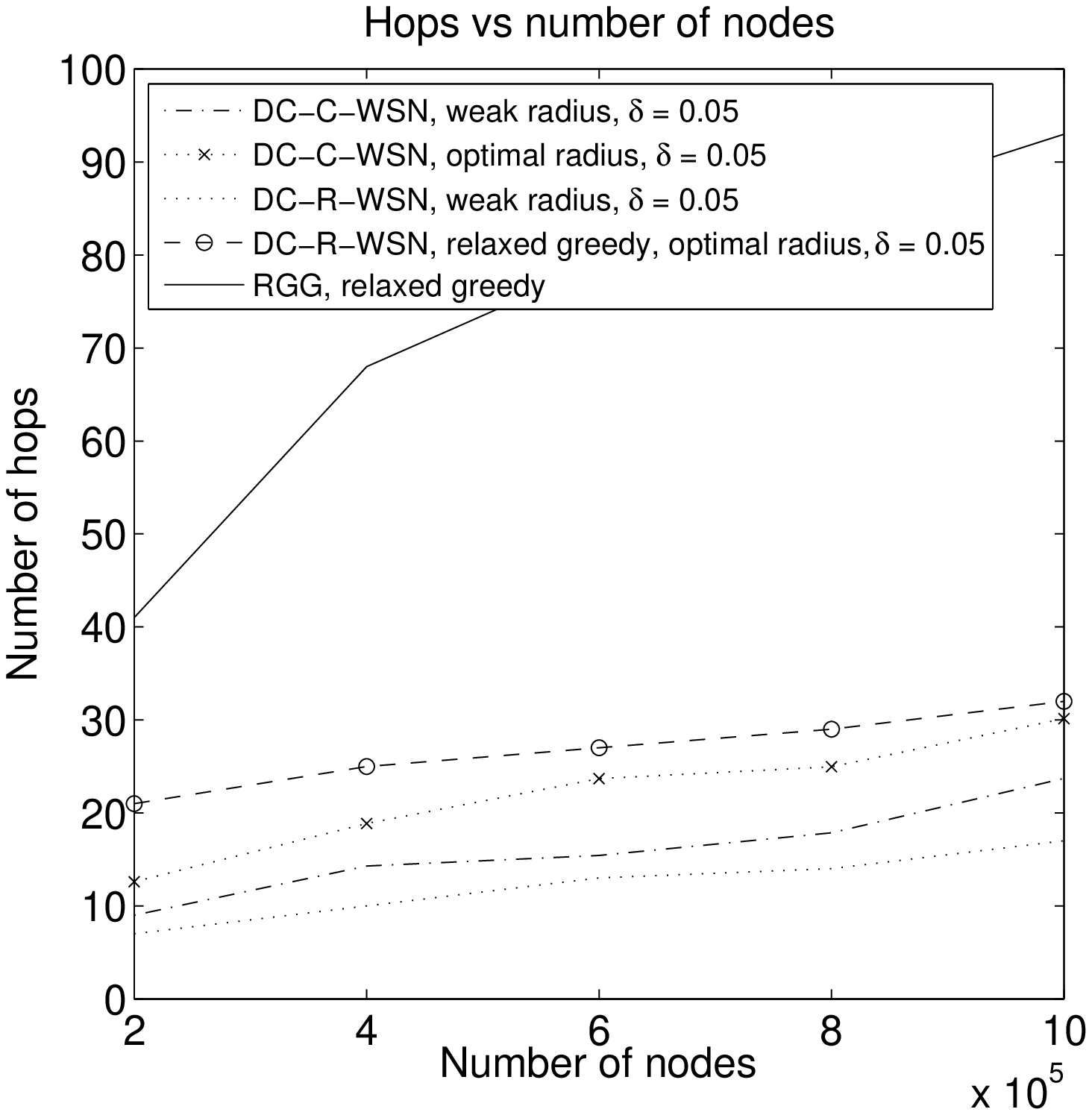}
	}
	%\qquad
	\subfloat[\label{fig:4_timer}]{
		\includegraphics[scale=0.33]{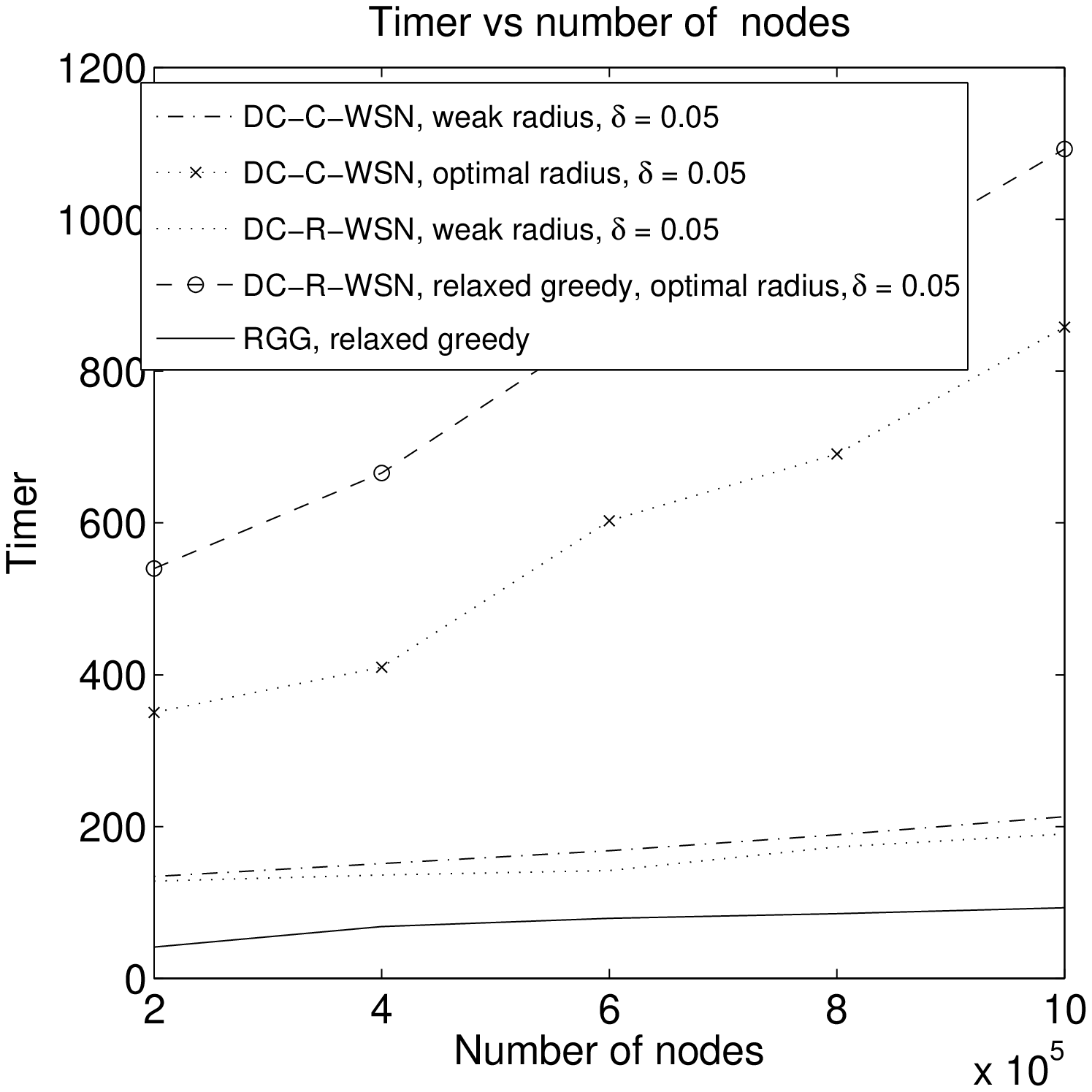}
	}
\caption{When $\delta=0.05$, $L=100$ and $d=5$: \protect \subref{fig:1_hops} Number of hops to reach $D$ from $S$
\protect \subref{fig:4_timer} Total completion time for $\Broad(M,S,D)$.}\label{fig:1-4}
\end{figure*}

\begin{figure*}
	\centering
		\includegraphics[scale=0.33]{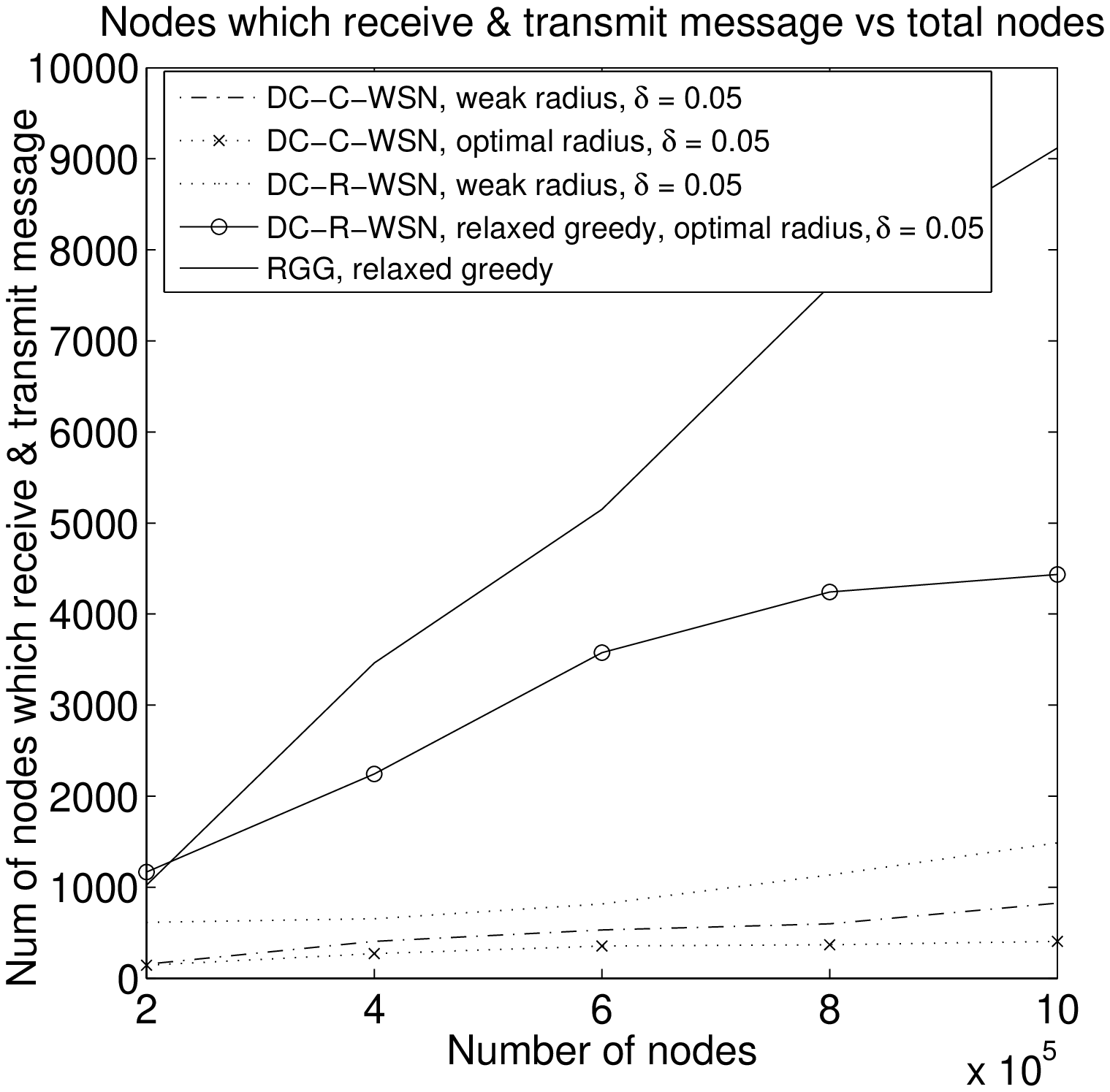}
	\caption{Number of sensors that receive $M$ and retransmit it during the $\Broad(M,S,D)$ operation when $\delta=0.05$.}\label{fig:2_N'}
\end{figure*}

Figure~\ref{fig:2_N'} plots the number $N'$ of sensors that receive
$M$ and retransmit it during the $\Broad(M,S,D)$ operation from $S$ to
$D$. In general $N'$ increases with the number of hops with the
``always on'' sensors registering the highest values of the five
scenarios studied.  When the radius is the same for both \dcc\ and
\dcr, i.e. the weak radius, we find that $N'$ is larger for
\dcr\ although we had seen in Figure~\ref{fig:1_hops} that number of
hops for \dcr\ with weak radius is lower than that of \dcc\ with weak
radius. Here we see that a node's ability to make more connections
under the \dcr\ scheme becomes a disadvantage in terms of power
consumption.

\iffalse
\begin{figure*}
	\centering
	\subfloat[\label{fig:3_diff}]{
		\includegraphics[scale=0.33]{finalfigures/3_diff_bw.eps
	}
	}
	%\qquad
	\subfloat[\label{fig:4_timer}]{
		\includegraphics[scale=0.33]{finalfigures/8_diffpowerN'_bw.eps}
	}
\caption{For $\Broad(M,S,D)$, when $d(S,D)=0.1$, $\delta=0.05$, and $L=100$: \protect \subref{fig:1_hops} Total Transmission Power;
\protect \subref{fig:4_timer} Transmission power for each sensor.}\label{fig:1-4}
\end{figure*}
\fi

%
% OMITTING THIS FIGURE. AB.
%
\iffalse
\begin{figure*}
	\centering
		\includegraphics[scale=0.33]{finalfigures/8_diffpowerN'_bw.eps}
\caption{Transmission Power for each single sensor involved in $\Broad(M,S,D)$, when $d(S,D)=0.1$, $\delta=0.05$, and $L=100$}\label{fig:8_diffpowerN'}
\end{figure*}
\fi
%
% END OMIT
%

%\iffalse
\begin{figure*}
	\centering
	\subfloat[\label{fig:3_diff}]{
		\includegraphics[scale=0.26]{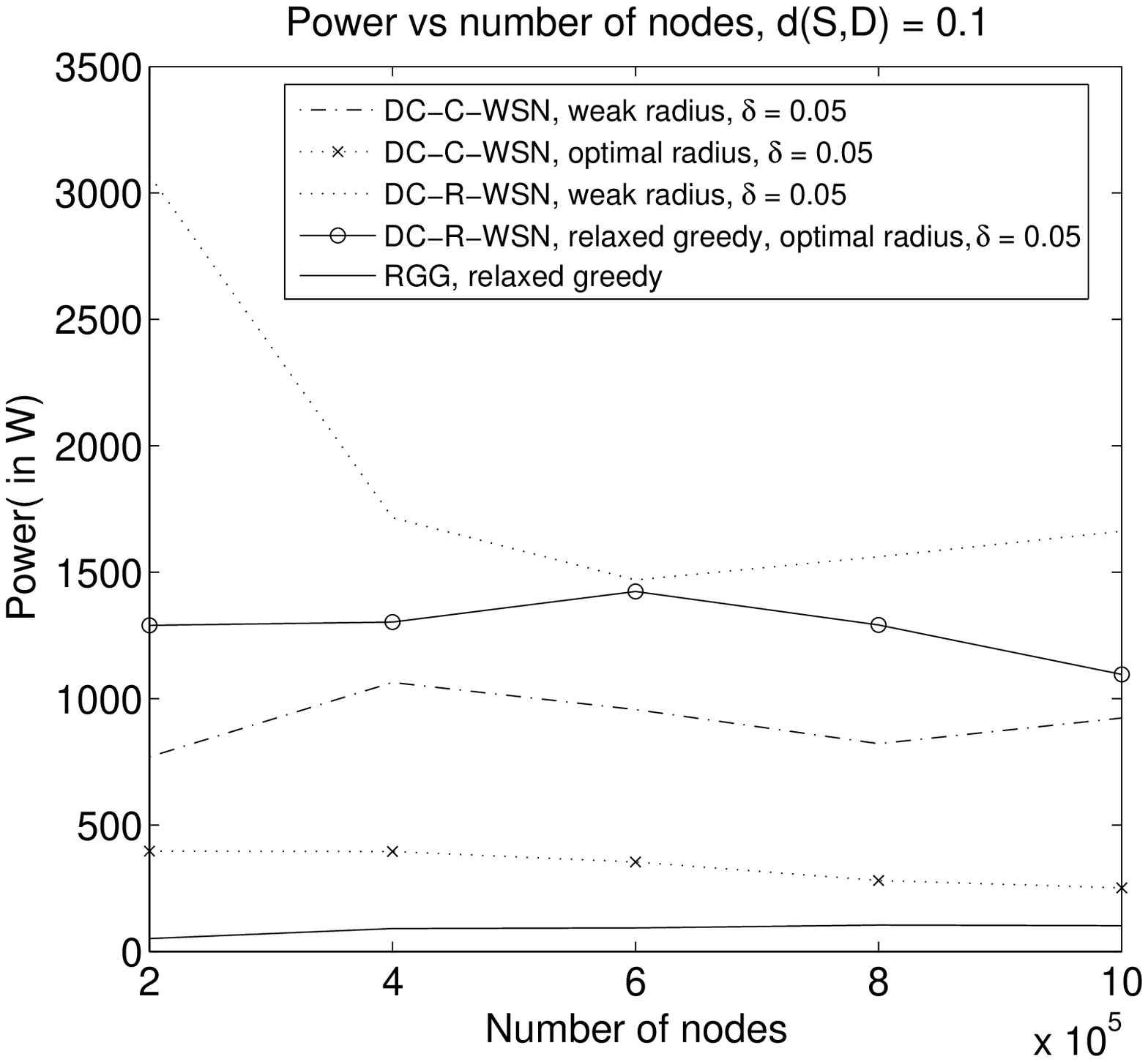}
	}
	%\qquad
	\subfloat[\label{fig:7_power}]{
		\includegraphics[scale=0.25]{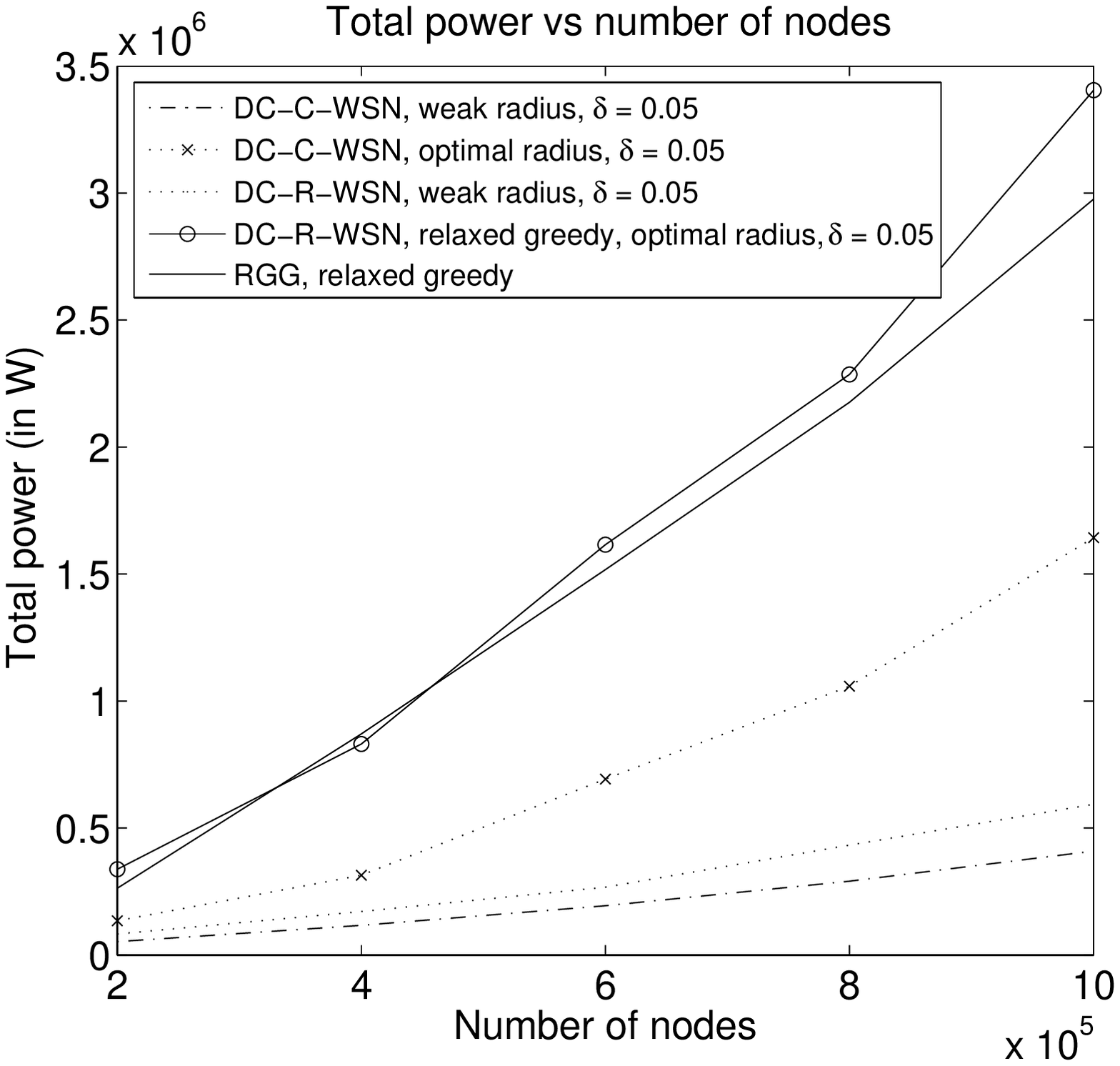}
	}
\caption{For $\Broad(M,S,D)$, when $d(S,D)=0.1$, $\delta=0.05$, and $L=100$: \protect \subref{fig:3_diff} Total transmission power
\protect \subref{fig:7_power} Total power consumed in task.}\label{fig:3-7}
\end{figure*}
%\fi

Figure~\ref{fig:3_diff} plots the total transmission power, i.e. the
power spent {\em only in transmission} by all the sensors that participate
in $\Broad(M,S,D)$. Among the duty-cycled schemes the winner here is
\dcc\ with optimal radius which reflects the fact that it has a low
number of transmitting nodes $N'$. The \dcr\ schemes both perform
worse than the \dcc\ schemes which is a clear consequence of having
made a larger number of transmissions (as we saw in
Figure~\ref{fig:2_N'}). But, on the other hand, \dcr\ with optimal
radius beats \dcr\ with weak radius, despite having made more
transmissions which shows the benefit of having a lower local power
consumption.  Always awake WSNs spend very little amount of energy for
transmission, but, as we see in Figure~\ref{fig:7_power}, the total
power they consume to complete the task is among the highest which is
clearly a consequence of their high cost of operation.  We note in
Figure~\ref{fig:7_power} that the duty-cycled networks that adopt the
weak radius are those that consume the minimum total power.  Among
these the \dcc s demand less power networks since they need lower
transmission power, as seen in Figure~\ref{fig:3_diff}, and because
they consume less power for their basic operation since they just
switch once between sleep and awake mode in each period. Overall the
worst performer is \dcr\ with the optimal radius which is due to the
fact that the poor connectivity it offers forces us to use the relaxed
greedy algorithm, thereby incurring a greater number of transmissions
(as we saw in Figure~\ref{fig:2_N'}) that offsets the benefit of the
low power consumption in local transmissions.

%
% OMITTING FIGURE
%
\iffalse
\begin{figure*}
	\centering
		\includegraphics[scale=0.33]{finalfigures/7_power_bw.eps}
\caption{Power consumed by the network until the $\Broad(M,S,D)$ is completed, when $d(S,D)=0.1$, $\delta=0.05$, and $L=100$}\label{fig:7_power}
\end{figure*}
\fi
%
% END OMIT
%

In conclusion we see that the global power consumption for a given
network task depends on several factors, many of which are specific to
the task that we are attempting. Minimizing the local power
consumption does not automatically minimize the global power
consumption for a given task. However the low cost of basic operation
of duty-cycled networks benefits them as long as they offer good
connectivity that allows us to carry out the given task. 

\section{A minimum-radius duty-cycling scheme}
\label{sec:deterministic}
\ifrevisiontwo
{\bf
\fi
So far we have studied the situation where given $d$ and $L$ and some
scheme for selecting $d$ waking slots out of $L$ we have a probability
$\gamma$ for connection between two neighboring nodes and hence a
minimum transmission radius for achieving connectivity in the network.
This transmission radius is $1/\sqrt{\gamma}$ higher than the radius
required for connectivity in RGGs. Now we show that with a careful
choice of such a scheme for selecting the slots we can achieve
connectivity even at the RGG radius. Namely, in this section we move away
from the general approach spanning a whole family of duty-cycling
schemes of the previous section and give an algorithm for finding a
particular awake scheduling scheme that ensures that the duty-cycled
network achieves connectivity using the RGG radius.
\ifrevisiontwo
}
\fi
%whenever the underlying network is
%connected, i.e. when the network formed by connecting nodes that are
%within transmission range of each other is connected.

Given an integer $k > 1$, consider a duty cycling scheme with the
following properties:
\begin{enumerate}
\item Each node chooses a duty cycle from one of $k$ predefined
  options, we call them $C_1, \ldots, C_k$ where each $C_i \subseteq
  \{0,1,\ldots,L\-1\}$ and $|C_i| = \delta \cdot L = d$.
\item All $k$ duty-cycle options overlap i.e. for each $1 \leq i \ne
  ,j \leq k$, $C_i \cap C_j \ne \emptyset$. 
\item No time instance is left uncovered i.e. $\cup_{i=1}^k C_i = \{0,1,\ldots,L-1\}$.
\end{enumerate}

Note that if such a scheme were to exist, we would be able to send
data from any node to node provided the base network is connected
i.e. if the Gupta-Kumar bound is satisfied (without any extra factor)
we achieve connectivity, since a node with any of the $k$ duty-cycles
can send a message to any of its neighbors. It just has to wait for
the overlap point of time to come.  Also, from a time coverage point
of view this scheme is good since there is no point of time when all
the sensors are off. And, in fact, on average $1/k$ fraction of all
sensors at least are guaranteed to be on at any time step.

The question is: Does such a schedule exist? The answer is
yes. Consider the following simple definition of a schedule:
\begin{quote}
Given $k$ we build $k$ schedules $A_1, \ldots A_k$ by randomly picking
$d$ time slots for each one of them independently of the others. 
\end{quote}

Clearly if $\delta > 1/2$ then any $k$ duty cycles we choose have the
property that all of them overlap, so we focus on the case where
$\delta \leq 1/2$.
\begin{claim}
For a given $k$, the random selection schedule described above has the
property that all the $k$ schedules overlap, with probability at least
\[ 1 - \frac{k(k+1)}{2} \cdot e^{-\delta d}.\]
\end{claim}
\begin{proof}
The probability that two schedules $A_i$ and $A_j$ do not overlap is
computed by calculating the probability that $A_j$ is picked only from
$\{0,1,2,\ldots,L-1\} \setminus A_i$ i.e.

\[\pr(A_i \cap A_j = \emptyset) = {L-d \choose d}/{L \choose d} =
\frac{L-d \cdot (L-d-1)\cdots (L-2d+1)}{L \cdot (L-1) \cdots (L-d+1)}
\leq \left(1 - \frac{d}{L}\right)^d \leq e^{-\delta d}.\]

Since there are $k(k+1)/2$ such pairs, we get the result claimed.
\end{proof}

We turn to the time coverage property. Clearly, the probability that a given slot $i$, $0 \leq i < L$ is not
covered by a given schedule is $(1 - \delta)$. Since all the schedules
are independent, slot $i$ is not covered by any schedule is $(1 -
\delta)^k$. Hence, using the union bound over all the $L$ slots, we
get:
\begin{claim}
The probability that every time slot is covered by at least one
schedule is at least
\[1 - L \cdot e^{-\delta k}.\]
\end{claim}

\begin{figure*}
	\centering{
		\includegraphics[scale=0.35]{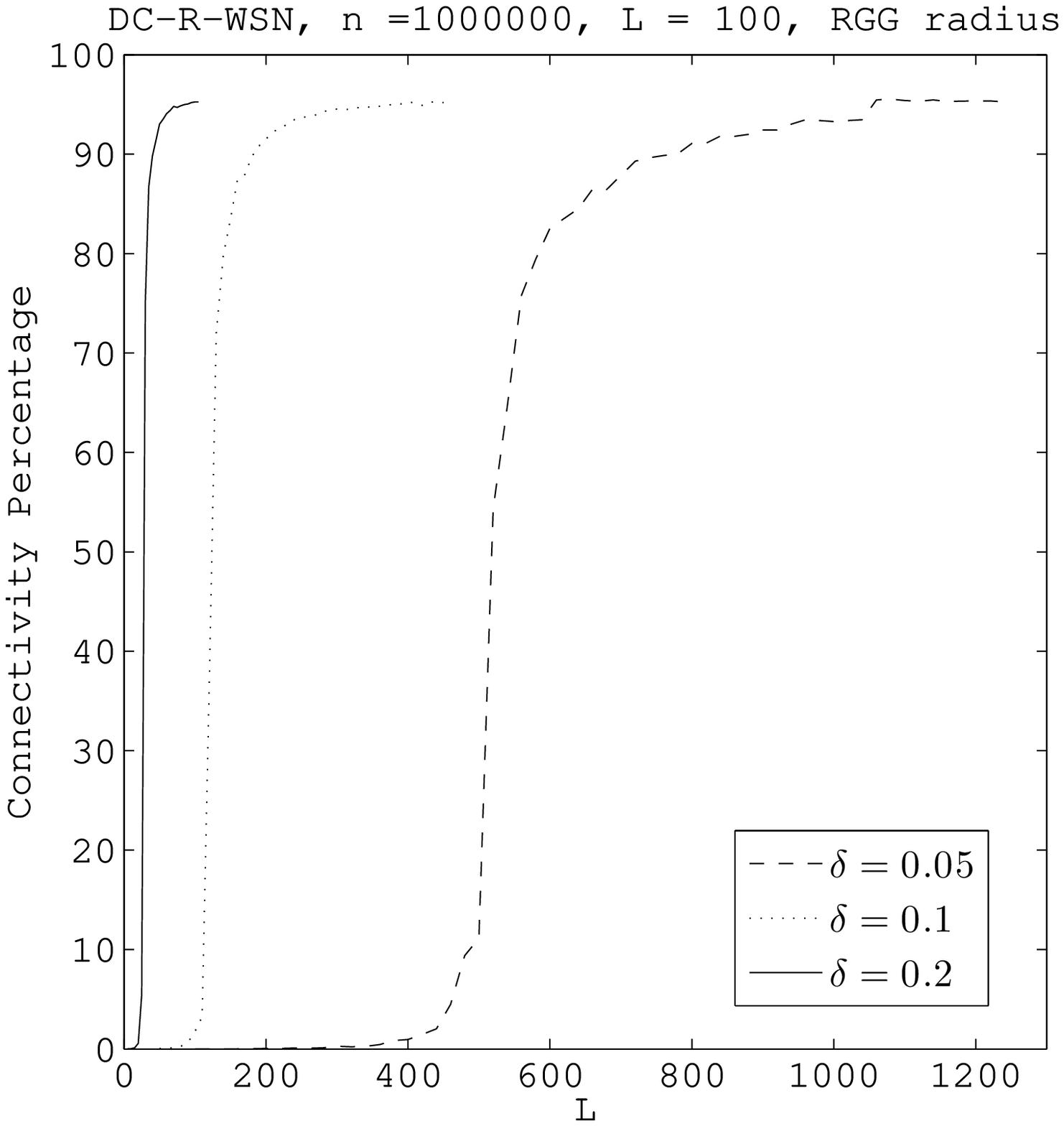}}
	%\qquad
        \
	\caption{Connectivity Percentage in a deterministic duty-cycle scheme when $L$ and $\delta$ vary.}\label{fig:DC-R-WSN-rgg}
\end{figure*}

Hence we get that a random choice of the $k$ duty-cycle schedules
gives us a schedule with both properties with probability:
\[ 1 - \left(\frac{k(k+1)}{2} \cdot e^{-\delta d} + L \cdot e^{-\delta k}\right).\]

This probability can be made arbitrarily close to 1 by choosing $L$
large enough and setting $k$ to a value that is $\theta(\log L)$. But
even if this probability is non-zero that is fine, because we need to
find these schedules offline and then hard code them into the
sensors. So, a randomized algorithm that keeps selecting $k$-schedules
and then testing them to see if they have both properties will finish
in polynomial time with high probability since testing for the two
properties can easily be done in linear time.

Hence, we see that there is a way of organizing duty-cycles for any
given value of $\delta$ such that if the Gupta-Kumar bound is achieved
then connectivity if achieved. The parameter whose value suffers to
achieve this is $L$ which may have to be made large.
In Figure~\ref{fig:DC-R-WSN-rgg}, the percentage of connectivity is plotted when $L$ and $\delta$ vary. 
For example, for $\delta=0.05$, the 90\% of connectivity is reached for $L \ge 600$ and hence $d=30$ is a good choice 
for the deterministic duty-cycle scheme.

\section{Conclusion: The wider implications of our results}
\label{sec:conclusion}

In this paper we have studied the duty-cycled wireless sensor network
setting and provided a necessary and sufficient condition on the
radius of transmission for connectivity in such networks. The work we
built on \cite{das-icuimc:2012} provided only a
sufficient condition which was a shortcoming that we have
rectified. In the process we have defined a new random connection
model which has not, to the best of our knowledge, been proposed
earlier. The most important contribution of this paper is
Theorem~\ref{thm:vb-gamma} which is a general theorem with
implications beyond the duty-cycling setting. 

An important setting in which Theorem~\ref{thm:vb-gamma} is applicable
is in wireless network security, specifically key-predistribution for
secure communication. In this setting, Eschenauer and Gligor proposed
a scheme in which each node of the network chooses $K$ keys at random
from a pool of $P$ available keys, and secure communication is
possible if two nodes share a common
key~\cite{eschenauer-ccs:2002}. Clearly if these nodes are part of a
sensor network with limited transmission power at each node, the
question of the radius of connectivity arises. It is also fairly clear
that the model here is {\em exactly} similar to that of our {\em
  random selection} duty cycle with $P$ playing the role of $L$ and
$K$ playing the role of $d$. 

In fact, in a recent paper~\cite{yagan-tit-a:2012} it was pointed
out that K. Krzywdzi\'{n}ski and
K. Rybarczyk~\cite{krzywdzinski-mfcs:2011} had proved that if
\[ \pi r(n)^2 \cdot \alpha_n = c \frac{\log n}{n}, \]
where $\alpha_n$ is the probability of two nodes sharing a key
(corresponding to $\gamma$ in our case), then the random geometric
graph with the Eschenauer-Gligor scheme is connected with probability
tending to 1 if $c > 8$ and is disconnected with probability tending
to 1 if $c < 1$. In~\cite{yagan-tit-a:2012} the author also claims
that Yi. et. al.~\cite{yi-dmaa:2010} had conjectured a stronger result
on the lines of Gupta and Kumar's result ~\cite{gupta-chapter:1999},
i.e. the Eschenauer-Gligor scheme on a graph is connected with
probability tending to 1 if and only if
\[ \pi r(n)^2 \cdot \alpha_n =  \frac{\log n + c(n)}{n}, \]
and $c(n) \rightarrow \infty$ as $n \rightarrow \infty$.  This
conjecture (Eq. (3) in~\cite{yagan-tit-a:2012}) can now be
considered closed since it is nothing other than our
Corollary~\ref{cor:vb-gamma-DC-R-WSN} which, as we have seen, follows
easily from Theorem~\ref{thm:vb-gamma}. It is our belief that for the
more abstract problem of key-predistribution on a complete graph (see
e.g.~\cite{yagan-tit:2012}), our method of defining a vertex-based
model of connectivity and proving the requisite properties may help
improve the current best known results in that area as well.

We feel that as in the case of key-predistribution on \rgg s, there
may be other settings where Theorem~\ref{thm:vb-gamma} may be
applicable, for example the study of connectivity in WSNs with
directional antennas where the direction is fixed at random
independently at each node. Our contribution, therefore, is a general
and foundational contribution, as well as a detailed and in-depth
study of the particular setting of duty-cycled WSNs.

\bibliographystyle{abbrv}
\bibliography{dutycycle}

%\appendices
\appendix
\section*{APPENDIX}

\section{Proof of Lemma~\ref{lem:penrose-modified}}

In order to prove this theorem we will show that Penrose's proof
technique can be followed for our model as well. To show this we
define some notation first defined in~\cite{penrose-aap:1991}.

Suppose $U = \{x_1, \ldots, x_k\}$ is finite set of points in $\RR^2$
  and $x_0$ is another point in $\RR^2$. Suppose we form a random
  graph $G$ with the points $U \cup \{x_0\}$ using our connection function
  $g(\cdot,\cdot)$ i.e. we associate a copy of $Z$ with each of the
  points and then draw edges as defined in~(\ref{eq:def}). 
\begin{itemize}
\item Define
  $g_1(x_0; U)$ to be the probability that $x_0$ is {\em not} isolated
  in this graph i.e. that $x_0$ is connected to at least one of the
  points in $U$. 
\item Also define $g_2(x_0, x_1, \ldots, x_k)$ to be the
  probability that the graph $G$ is connected.
\end{itemize}

The proof technique partitions the event $\{ |W| = k\}$ by
mapping the points of $W$ to a lattice of points of the form
$\delta z, z \in \ZZ^2$ and studying the number of such points in the
mapping. This is done as follows: We define a map $F_\delta : \RR^2
\rightarrow \delta \ZZ^2$ which sends a point of $\RR^2$ to the
nearest point of the lattice $\delta \ZZ^2$. Since in general there
may be more than one such point, we note that in a Poisson point
process this does not happen with probability 1, and hence this map is
well defined with probability 1. What could happen, however is that a
number of points of $\RR^2$ get mapped to the same point of $\delta
\ZZ^2$ and in fact we can see that a box of width $\delta$ centred 
at a point of the lattice is mapped to the lattice point at the centre
by $F_\delta$. To be able to describe such boxes, we denote by $B_l$
the box $[-l,l] \times [-l,l]$ i.e. the box of width $2l$ centred at
the origin. Further we will denote by $S_\delta$ the set of points of
$\delta \ZZ^2$ which are images of the points in $W$. 

In order to prove the theorem, we will prove the following lemmas that
Penrose demonstrated are true for the random connection model.

\begin{lemma}
\label{lem:penrose-1}
For sufficiently small $\delta$,
\[ \lim_{\lambda \rightarrow \infty} \frac{\sum_{k=2}^{\infty}
  \pr_\lambda(|W| = k \cap |S_\delta| = 1) }{q_1(\lambda)} = 0.
\]
\end{lemma}
This lemma is clearly not enough to prove
Lemma~\ref{lem:penrose-modified}, but it helps us prove the
following lemma from which the theorem follows:
\begin{lemma}
\label{lem:penrose-m}
For sufficiently small $\delta$ and for each fixed $m$,
\[ \lim_{\lambda \rightarrow \infty} \frac{\sum_{k=1}^{\infty}
  \pr_\lambda(|W| = k \cap |S_\delta| = m) }{q_1(\lambda)} = 0.
\]
\end{lemma}

We will need the following general characterization:
\begin{proposition}
\label{prp:penrose}
For any $k \in \NN$,
\begin{align}
\pr_\lambda(|W| = k \cap W \subset B_l) = & \frac{\lambda^{k-1}}{(k-1)!} \int_{B_l} \cdots
\int_{B_l} g_2(0,x_1,\ldots, x_k-1)\nonumber \\
& \cdot \exp\left\{ -\lambda \int_{\RR^2}
    g_1(y;\{0,x_1,\ldots,x_{k-1})dy \right\} dx_1\cdots dx_{k-1}.
\end{align}
\end{proposition}
\begin{proof}
Let $E(n,l)$ denote the event that $|W| = k$ and $W \subseteq
B_l$. Let us condition on the event that the number of points of the
Poisson process in $B_l$ is $m$, we call this event $|V(B_l)| =
m$. Recall that the $m$ points inside $B_l$ are uniformly distributed
when we condition on $|V(B_l)| = m$. Hence, we have that:
\begin{equation}
\label{eq:e_n_l}
\pr(E(n,l) | |V(B_l)| = m) = {m \choose k-1}\cdot
\left(\frac{1}{2l}\right)^{2m} \cdot \int_{B_l} \cdots \int_{B_l}
\pr'(W = \{0,x_1,\ldots,x_{k-1}\}) dx_1\cdots dx_{k-1},
\end{equation}
where
$\pr'(0,x_1,\ldots,x_{k-1})$ is the probability measure in the
subspace where there are $m$ uniformly distributed points in $B_l$,
there is a point at the origin, and there is a Poisson point process
of density $\lambda$ in $\RR^2 \setminus B_l$. The integral on the
right hand side is obtained by conditioning on the events that each of
the $m$ points is in a square of area $dx_i, 1 \leq i \leq m$ (with
probability $dx_i/(2l)^2$ and then choosing $k-1$ of them to be part
of $W$. We decondition by integrating over $B_l$ for all $m$
points. The $m-(k-1)$ that don't get chosen simply contribute a factor
of $1/(2l)^2$ to the integral.

If $W = \{0,x_1,\ldots,x_{k-1}\}$ then it must be the case that no
point from $\RR^2 \setminus B_l$ is connected to the points
$0,x_1,\ldots,x_l$. We use the function $g_1(\cdot;\cdots)$ defined
above in this case. As remarked above, given a fixed set of points,
$0,x_1,\ldots,x_{k-1}$ in this case, the collection of events that any
point of $\RR^2\setminus B_l$ 
is connected to (or isolated from) these points is an independent
collection. Hence, by Proposition 1.3 of Meester and
Roy~\cite{meester:1996}, the set of points connected to these points
forms an (inhomogeneous) thinning of the Poisson point process $V(\RR^2
\setminus B_l$, with a thinning factor which is exactly equal to
$g_1(y;0,x_1,\dots,x_{k-1})$ for any point $y \in \RR^2 \setminus
B_l$. Hence, the probability that no point of $\RR^2 \setminus B_l$ is
connected to any of $0,x_1,\ldots,x_{k-1}$ is 
\[\exp\left\{ -\lambda \int_{\RR^2 \setminus B_l}
  g_1(y;0,x_1,\dots,x_{k-1}) dy\right\}.\]
From the definition of $g_2(\cdots)$ we obtain that
\begin{align*}
\pr'(W = \{0,x_1,\ldots,x_{k-1}\}) & = g_2(0,x_1,\ldots,x_{k-1}) \cdot \exp\left\{ -\lambda \int_{\RR^2 \setminus B_l}
  g_1(y;0,x_1,\dots,x_{k-1}) dy\right\}\\
&\cdot \prod_{i=k}^m (1 - g_1(x_i;0,x_1,\ldots,x_{k-1})),
\end{align*}
where again we use the fact that the events that $x_i$ is isolated from
$0,x_1,\ldots,x_{k-1}$, for $k \leq i \leq m$ form an independent
collection. Now, substituting into~(\ref{eq:e_n_l}), we get 
\begin{align}
\pr(E(n,l) \cap |V(B_l)| = m) & =  e^{-\lambda (2l)^2} \frac{(\lambda (2l)^2)^m}{m!}
{m \choose k-1}\cdot
\left(\frac{1}{2l}\right)^{2m} \nonumber \\
& \cdot \int_{B_l} \cdots \int_{B_l}  g_2(0,x_1,\ldots,x_{k-1}) \nonumber \\
& \cdot \exp\left\{ -\lambda \int_{\RR^2 \setminus B_l}
  g_1(y;0,x_1,\dots,x_{k-1}) dy\right\} \nonumber \\
& \cdot \left( \int_{B_l} (1 -
  g_1(z;0,x_1,\ldots,x_{k-1}))dz\right)^{m-(k-1)} dx_1\cdots
dx_{k-1}. \label{eq:e_n_l_cond}
\end{align}
Now consider the quantity:
\begin{align*}
a_m & = e^{-\lambda (2l)^2}\frac{(\lambda (2l)^2)^m}{m!}
{m \choose k-1}\cdot
\left(\frac{1}{2l}\right)^{2m}  \cdot \left( \int_{B_l} (1 -
  g_1(z;0,x_1,\ldots,x_{k-1}))dz\right)^{m-(k-1)}\\
& = e^{-\lambda (2l)^2}\frac{\lambda^m}{(m-(k-1))!(k-1)!} \cdot \left( \int_{B_l} (1 -
  g_1(z;0,x_1,\ldots,x_{k-1}))dz\right)^{m-(k-1)}.
\end{align*}
Summing $a_m$ over all $m \geq k-1$, we get 
\[\sum_{i=k-1}^\infty a_m = \frac{e^{-\lambda (2l)^2}
  \lambda^{k-1}}{(k-1)!}\cdot \exp - \left\{ \lambda \int_{B_l} (1 -
  g_1(z;0,x_1,\ldots,x_{k-1}))dz\right\}.\]
Combining this with~(\ref{eq:e_n_l_cond}), we get
\begin{align*}
\pr(E(n,l)) & = \frac{e^{-\lambda (2l)^2}
  \lambda^{k-1}}{(k-1)!}  \cdot \int_{B_l} \cdots \int_{B_l}
g_2(0,x_1,\ldots,x_{k-1})  \\
& \cdot \exp\left\{ -\lambda \int_{\RR^2 \setminus B_l}
  g_1(y;0,x_1,\dots,x_{k-1}) dy\right\} \cdot \exp\left\{ -\lambda \int_{B_l}
  g_1(y;0,x_1,\dots,x_{k-1}) dy\right\}  \\
& \cdot \exp\left\{\int_{B_l} dz\right\}
\end{align*}
The result follows by observing that the last term is equal to $e^{\lambda (2l)^2}$.
\end{proof}

We now move to the proof of the two lemmas.

\begin{proofof}{Lemma~\ref{lem:penrose-1}}
Denote by $q^\delta_k(\lambda)$ the quantity $\pr(|W| = k \cap
|S_\delta| =1)$. Since the condition that $|S_\delta| = 1$ is the same
as saying that all the points of $W$ are contained in $B_{\delta/2}$
i.e. the box with width $\delta$ centred at the origin, we can use
Proposition~\ref{prp:penrose} to say that 
\begin{align*}
\frac{q_k^\delta(\lambda)}{q_1{\lambda}} &= \frac{\lambda^{k-1}}{(k-1)!} \int_{B_{\delta/2}} \cdots
\int_{B_{\delta/2}} g_2(0,x_1,\ldots, x_k-1) \\
& \cdot \exp\left\{ -\lambda \int_{\RR^2}
    g_1(y;\{0,x_1,\ldots,x_{k-1})dy \right\} dx_1\cdots dx_{k-1}.\\
& \leq \frac{\lambda^{k-1}}{(k-1)!} \int_{B_{\delta/2}} \cdots
\int_{B_{\delta/2}} \exp\left\{ -\lambda \int_{\RR^2}
    g_1(y;\{0,x_1,\ldots,x_{k-1})dy \right\} dx_1\cdots dx_{k-1}.
\end{align*}
The second step following by the fact that the function
$g_2(\cdot;\cdots)$ takes value at most 1. 
Now, consider a $\delta$ small enough that $B_{\delta}$ is completely
contained in a circle centred at 0 with radius $r/2$, where $r$ is the
radius defined in the connection function. Let us denote this circle
$C(0,r/2)$. In order to apply the connection
diversity condition we use the lower bound
\[\int_{\RR^2}
    g_1(y;\{0,x_1,\ldots,x_{k-1})dy \geq \int_{C(0,r/2)}
    g_1(y;\{0,x_1,\ldots,x_{k-1})dy.\]
Clearly a point in $C(0,r/2)$ is within distance $r$ of
any other point in $C(0,r/2)$, so we can apply the connection
diversity condition~(\ref{eq:conn-div}) to get
\begin{align*}
\frac{q_k^\delta(\lambda)}{q_1{\lambda}} & \leq \frac{\lambda^{k-1}}{(k-1)!} \int_{B_{\delta/2}} \cdots
\int_{B_{\delta/2}} \exp\left\{ -\lambda \int_{C(0,r/2)} c dy \right\}
dx_1\cdots dx_{k-1}.\\
& = \frac{(\lambda\cdot\delta^2)^{k-1}}{(k-1)!}\cdot e^{-c\pi r^2 /2}.
\end{align*}
Summing over $k\geq 2$, we get
\begin{align*}
\frac{\sum_{k=2}^\infty q_k^\delta(\lambda)}{q_1{\lambda}} & \leq
e^{-c\pi r^2 /2}. \sum_{k=2}^\infty
\frac{(\lambda\cdot(\delta)^2)^{k-1}}{(k-1)!}\\
& = e^{\lambda \delta^2 -c\pi r^2 /2},
\end{align*}
Which tends to 0 as $\lambda \rightarrow \infty$ for all $\delta$
satisfying $\lambda \delta^2 -c\pi r^2 /2 < 0$.
\end{proofof}

\begin{proofof}{Lemma~\ref{lem:penrose-m}}
Since our connection function allows no connection beyond a fixed
length (the so-called ``bounded support'' condition), there are only
finitely many configurations of $S_\delta$ with $|S_\delta| = m$. So,
we will show that for any $\eta$ which is a finite subset of $\delta
\ZZ^2$, 
\[ \sum_{k=1}^\infty \pr_\lambda(|W| = k \cap S_\delta = \eta)/q_1(\lambda)
\rightarrow 0 \mbox{ as }\lambda \rightarrow \infty.\]
We denote by $W_\eta$, the connected component containing the origin
when we remove all the points of $V$ that lie outside
$F^{-1}_\delta(\eta)$ i.e. the area of the plane outside the set of
squares of side $\delta$ depicted in Figure~\ref{fig:penrose-k}
\begin{figure}[htbp]
\begin{center}
\input{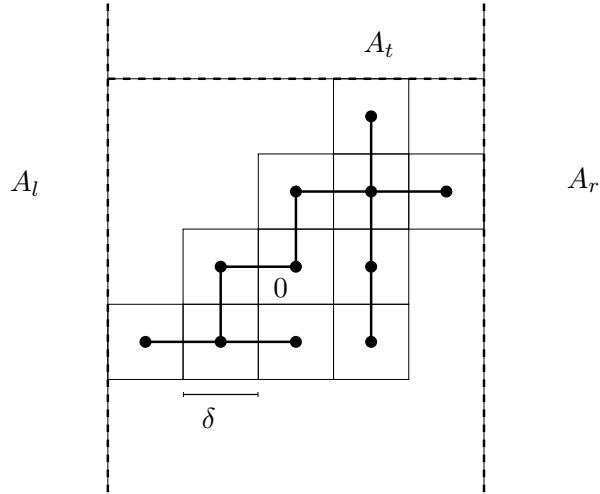}
\caption{$F_\delta$ with each point shown surrounded by a box of the
  form $B_{\delta/2}$. $A_l$, $A_r$ and $A_t$ are also shown.}
\label{fig:penrose-k}
\end{center}
\end{figure}
Let $E(\eta,k)$ be the event that $W_\eta$ has $k$ points in it and
for each point of $\eta$ at least one point of $W_\eta$ lies in the
square of side $\delta$ centred at that point. Let $H_\eta$ be the
event that there is no point of $V$ in $\RR^2 \setminus
F^{-1}_\delta(\eta)$ that is connected to any point of $W_\eta$.  Then
the event $\{ |W| = k \cap S_\delta = \eta\}$ is the same as the event
$E(\eta,k) \cap H_\eta$. We will estimate the probability
\begin{equation}
\label{eq:enk}
\pr_\lambda(H_\eta | E(\eta,k)) = \exp\left\{-\lambda\int_{\RR^2
    \setminus F^{-1}_\delta(\eta)} g_1(y;W_\eta)dy\right\}.
\end{equation}
Where, as before, the equality follows from Proposition 1.3 of Meester
and Roy~\cite{meester:1996}.  Suppose we have $\delta_1$ small enough
for the conclusion of Lemma~\ref{lem:penrose-1} to hold. In that case
if $F^{-1}_\delta(\eta)$ can be contained in $B_{\delta_1/2}$ then the
we are done. So we will assume that the width of $F^{-1}_\delta(\eta)$
is at least $\delta_1$. An argument symmetric to the one we will show
can be made if the height is more than $\delta_1$.

Define $x_l = \inf \{x : (x,y) \in \eta\}$, $x_r = \sup \{x : (x,y)
\in \eta\}$ and $x_t = \sup \{y : (x,y) \in \eta\}$ to be the $x$
coordinates of the leftmost and rightmost points, and the $y$
coordinate of the topmost points of $\eta$. Using these, define the
sets $A_l = \{(x,y) : x < x_l - \delta/2\}$, $A_r = \{(x,y) : x < x_l
+ \delta/2\}$ and $A_t = \{(x,y) : y > x_t + \delta/2, x_l - \delta/2
\leq x \leq x_r + \delta/2\}$ (see Figure~\ref{fig:penrose-k}.) Note
that $A_l$, $A_r$ and $A_t$ are disjoint regions of $\RR^2\setminus
F^{-1}_\delta(\eta)$. 

By the definition of $E(\eta,k)$ and $x_l$, if $E(\eta,k)$ occurs
there must exist a point $u_1 = (x_1,y_1) \in W_\eta$ such that $|x_1
- x_l| \leq \delta/2$. We lower bound the probability of a point in
$A_l$ connecting to any point in $W_\eta$ by the probability that it
connects to $u_1$. Hence
\[ \int_{A_l} g_1(y,W_\eta)dy \geq \int_{A_l} g(y,u_1)dy\]
Since the distance between the boundary of $A_l$ and $u_1$ could be as
large as $\delta$ (but not larger), and shifting $u_1$ to the origin, we have
\[ \int_{A_l} g_1(y,W_\eta)dy\geq \int_{(-\infty,-\delta)\times(-\infty,\infty)} g(y,0)dy\]
The same argument holds for $A_r$ with a similarly chose $u_2$ and so
\begin{equation}
\label{eq:al-ar}
\int_{A_l \cup A_r} g_1(y,W_\eta)dy \geq
\int_{\RR^2} g(y,0)dy - \int_{(-\delta,\delta)\times(-\infty,\infty)} g(y,0)dy
\end{equation}
Also, if $E(\eta,k)$ occurs, there is a point $u_3 = (x_3,y_3) \in
W_\eta$ such that $|x_t - y_3| \leq \delta/2$. And so, as before,
\begin{equation}
\label{eq:at}
\int_{A_t} g_1(y,W_\eta)dy\geq \int_{A_t} g(y,u_3)dy
\end{equation}
Now, we define two sets in $\RR^2$, $A_+ = (0,\delta_1/2) \times
(\delta,\infty)$ and $A_- (-\delta_1,0) \times (\delta,\infty)$. 
From the non-triviality condition on $f$ it is easy to see that for any
value of $r > 0$,
\[\min\left\{ \int_{(0,\delta_1/2)\times(0,\infty)} g(0,x)dx,
  \int_{(-\delta_1/2,0)\times(0,\infty)} g(0,x)dx\right\} >0. \]
Hence, if we choose a small enough value of $\delta$, we can find a $c
> 0$ such that 
\begin{equation}
\label{eq:aplus}
\min\left\{ \int_{A_+} g(0,x)dx,
  \int_{A_-} g(0,x)dx\right\}> c + \int_{(-\delta,\delta)\times(-\infty,\infty)} g(0,x)dx.
\end{equation}
Note that $A_-$ and $A_+$ have width $\delta_1/2$ and we are in the
case where the width of $\eta$ is at least $\delta_1$. So, if we
recenter $A_-$ and $A_+$ at $u_3$, at least one of them will be fully
contained in $A_t$. Combining this observation with~(\ref{eq:at})
which bounds the integral of the isolation function in terms of the
connection function around $u_3$ we get that 
\begin{equation*}
\int_{A_t} g_1(y,W_\eta)dy \geq\min\left\{ \int_{A_+} g(0,x)dx,
  \int_{A_-} g(0,x)dx\right\}
\end{equation*}
Further substituting~(\ref{eq:aplus}), this given us
\begin{equation}
\label{eq:at2}
\int_{A_t} g_1(y,W_\eta)dy \geq c + \int_{(-\delta,\delta)\times(-\infty,\infty)} g(0,x)dx.
\end{equation}
Combining this with~(\ref{eq:al-ar})
\[ \int_{A_t\cup A_l \cup A_r} g_1(y,W_\eta)dy \geq c +\int_{\RR^2} g(0,x)dx.\]
Observing that $A_l \cup A_r \cup A_t \subseteq \RR^2
\setminus F^{-1}_\delta(\eta)$, the last equation substituted
into~(\ref{eq:enk}) gives us that
\[\pr_\lambda(H_\eta | E(\eta,k))\leq \exp\left\{-\lambda \left(c +
 \int_{\RR^2} g(0,x)dx\right)\right\}.\]
Since, $q_1(\lambda) = \exp\left\{ -\lambda \int_{\RR^2}
  g(0,x)dx\right\}$, we get that 
\[\frac{\pr_\lambda(H_\eta | E(\eta,k))}{q_1(\lambda)} \leq
  e^{-\lambda c},\]
which in turn means that 
\[\frac{\pr_\lambda(H_\eta \cap E(\eta,k))}{q_1(\lambda)} \leq
  \pr_\lambda(E(\eta,k))\cdot  e^{-\lambda c}.\]
Summing over $k \geq 1$, we have 
\[\frac{\sum_{k=1}^{\infty}\pr_\lambda(H_\eta \cap E(\eta,k))}{q_1(\lambda)} \leq
 e^{-\lambda c} \cdot \sum_{k=1}^{\infty}  \pr_\lambda(E(\eta,k)).\]
For any fixed $\eta$, since $\sum_{k=1}^{\infty}
\pr_\lambda(E(\eta,k))$ is at most 1, the right hand side tends to 0
as $\lambda \rightarrow \infty$.
\end{proofof}

\end{document}